\documentclass[a4paper,UKenglish,cleveref, autoref]{lipics-v2019}
\nolinenumbers

\title{Bialgebraic Semantics for String Diagrams}

\titlerunning{Bialgebraic Semantics for String Diagrams}

\author{Filippo Bonchi}{University of Pisa}{}{}{}

\author{Robin Piedeleu}{University College London}{}{}{}

\author{Pawel Sobocinski}{University of Southampton}{}{}{}

\author{Fabio Zanasi}{University College London}{}{}{}

\authorrunning{F. Bonchi, R. Piedeleu, P. Sobocinski and F. Zanasi}

\Copyright{Filippo Bonchi, Robin Piedeleu, Pawel Sobocinski and Fabio Zanasi}



\category{}

\relatedversion{}

\supplement{}

\EventEditors{X}
\EventNoEds{30}
\EventLongTitle{30th International Conference on Concurrency Theory (CONCUR 2019)}
\EventShortTitle{CONCUR 2019}
\EventAcronym{CONCUR}
\EventYear{2019}
\EventDate{X}
\EventLocation{Amsterdam, The Netherlands}
\EventLogo{}
\SeriesVolume{X}
\ArticleNo{X}

\usepackage{graphicx,epsfig,color}
\usepackage{xypic}
\xyoption{rotate}
\input xy
\xyoption{all}
\usepackage{verbatim}
\usepackage{amssymb}
\usepackage[mathscr]{eucal}
\usepackage{stmaryrd}
\usepackage{multicol}
\usepackage{manfnt}
\usepackage{enumerate}
\usepackage{amsfonts}
\usepackage{mathrsfs}
\usepackage[all,cmtip,2cell]{xy}
\CompileMatrices
\UseAllTwocells
\SilentMatrices
\usepackage{cite}
\usepackage{fancybox}
\usepackage{hyperref}
\usepackage{wrapfig}
\usepackage{amsmath}
\usepackage{bm}
\usepackage{etoolbox}
\def \N {\mathbb{N}}
\newcommand{\Fbeh}{\mathcal{F}}

\newcommand{\EM}[1]{EM(#1)}
\newcommand{\Alg}[1]{Alg(#1)}

\newcommand{\Coalg}[1]{CoAlg(#1)}

\newcommand{\SCOMPSYM}{\mathcal{S}_{\texttt{SM}}}

\newcommand{\SPROP}{\mathcal{S}_{\texttt{PROP}}}
\newcommand{\SCOMPFR}{\mathcal{S}_{\texttt{FR}}}
\newcommand{\SCW}{\mathcal{S}_{\texttt{CW}}}

\newcommand{\SL}{\mathcal{S}_{\texttt{L}}}

\newcommand{\IdFunc}{\iota}

\def \df {\ensuremath{:=}}

\newcommand{\tr}[1]{\xrightarrow{#1}}    
\newcommand{\tl}[1]{\xleftarrow{#1}}    

\def \To {\Rightarrow}

\def \catC {\mathcal{C}}
\def \catD {\mathcal{D}}

\def \funF {\mathcal{F}}
\def \funG {\mathcal{G}}
\def \funT {\mathcal{T}}

\def \Pow {\mathcal{P}_{\scriptscriptstyle{\kappa}}} 
\newcommand{\free}[1]{{#1}^{\dagger}}

\newcommand{\Span}[1]{\mathsf{Span}(#1)}

\def \PROP {\mathbf{PROP}} 
\def \CWP {\mathbf{CW}} 
\def \SET {\mathsf{Set}} 

\def \SIG {\mathsf{Sig}}
\def \poi {\,\ensuremath{;}\,} 
\def \tns {\ensuremath{\oplus}} 

\newcommand{\from}{\mathrel{:}}
\def \id {\mathit{id}} 

\newcommand{\quot}[2]{{#1}_{\! \raisebox{.1em}{$\scriptstyle /$} \! {\scriptscriptstyle #2}}}

\newcommand{\lbbd}{\mathopen{[\![}}
\newcommand{\rbbd}{\mathclose{]\!]}}
\newcommand{\sem}[1]{\lbbd{#1}\rbbd} 

\theoremstyle{plain}
\usepackage[layout=margin,draft]{fixme}
\FXRegisterAuthor{fz}{afz}{Fab}
\FXRegisterAuthor{fb}{afb}{Fil}
\FXRegisterAuthor{ps}{aps}{Paw}

\newcommand{\Circ}[1]{\mathsf{Circ}_{\scriptstyle #1}} 
\newcommand{\Rig}{\mathsf{R}}

\newcommand{\typ}{\mathrel{:}}
\newcommand{\typeJudgment}[3]{{ {#2} \,\typ\, {#3}}}
\newcommand{\sort}[2]{\ensuremath{(#1,\,#2)}}

\let\tinymatrix\smallmatrix

\patchcmd{\tinymatrix}{\scriptstyle}{\scriptscriptstyle}{}{}
\patchcmd{\tinymatrix}{\scriptstyle}{\scriptscriptstyle}{}{}
\patchcmd{\tinymatrix}{\vcenter}{\vtop}{}{}
\patchcmd{\tinymatrix}{\bgroup}{\bgroup\scriptsize}{}{}

\usepackage{prooftree}

 \newcommand{\derivationRule}[3]{{\prooftree{ #1}\justifies{ #2}\using\ruleLabel{#3}\endprooftree}}
\newcommand{\reductionRule}[2]{{\prooftree{\scriptstyle #1}\justifies{\scriptstyle #2}\endprooftree}}

\newcommand\twarr[2]{%
\mathrel{\mathop{\moverlay{\scriptstyle\xrightarrow{\,#1\,}\cr{\lower.2em\hbox{$\scriptstyle{}_{#2}$}}}}}}
\newcommand\twarrw[2]{%
\mathrel{\mathop{\moverlay{\scriptstyle\Longrightarrow\cr{\lower-.6em\hbox{$\scriptstyle{}_{#1}$}}
\cr{\lower.3em\hbox{$\scriptstyle{}_{#2}$}}}}}}

\newcommand{\dtrans}[2]{\hbox{$\;\twarr{#1}{#2}\;$}}
\newcommand{\dtransw}[2]{\raise1pt\hbox{$\;\twarrw{#1}{#2}\;$}}

\newcommand{\ruleLabel}[1]{#1}
\newcommand{\labelSep}{\,}

\makeatletter
\def\moverlay{\mathpalette\mov@rlay}
\def\mov@rlay#1#2{\leavevmode\vtop{%
\baselineskip\z@skip \lineskiplimit-\maxdimen
\ialign{\hfil$#1##$\hfil\cr#2\crcr}}}
\makeatother

\newcommand{\stcong}{\equiv}

\newcommand{\procvar}[1]{\mathtt{#1}}
\newcommand{\syntaxdf}{\mathtt{:=}}
\newcommand{\ArEq}{\mathcal{A}}
\newcommand{\ArEqSM}{\ArEq}
\newcommand{\Eq}{\mathbb{E}}
\newcommand{\EqSM}{\Eq_{\scriptscriptstyle{\mathsf{SM}}}}
\newcommand{\EqCW}{\Eq_{\scriptscriptstyle{\mathsf{CW}}}}
\newcommand{\ArEqCW}{\ArEq_{\scriptscriptstyle{\mathsf{CW}}}}
\newcommand{\lCW}{l_{\scriptscriptstyle{\mathsf{CW}}}}
\newcommand{\rCW}{r_{\scriptscriptstyle{\mathsf{CW}}}}
\newcommand{\lambdaseqcomp}{\lambda^{\scriptstyle{\mathsf{sq}}}}
\newcommand{\lambdaparcomp}{\lambda^{\scriptstyle{\mathsf{mp}}}}
\newcommand{\lambdaid}{\lambda^{\scriptstyle{\mathsf{id}}}}
\newcommand{\lambdazero}{\lambda^{\scriptstyle{\mathsf{\epsilon}}}}
\newcommand{\lambdasym}{\lambda^{\scriptstyle{\mathsf{sy}}}}
\newcommand{\lambdaquot}[1]{\lambda^\dagger {\scriptscriptstyle{/}_{\scriptscriptstyle{\mathsf{#1}}}}}

\newcommand{\Powf}{\mathcal{P}_{\omega}}

\usepackage{tikz}
\usetikzlibrary{arrows,snakes,automata,backgrounds,petri, positioning,calc}
\usetikzlibrary{shapes.geometric,shapes.misc,matrix,arrows.meta,decorations.markings}
\tikzset{x=1em, y=1.5ex, baseline=-0.5ex}
\tikzset{ihbase/.style={inner sep=0,circle,draw,fill=lightgray,minimum size=0.4em,node contents={}}}
\tikzset{ihblack/.style={ihbase,fill=black}}
\tikzset{ihwhite/.style={ihbase,fill=white}}
\tikzset{ihgray/.style=black-faded}
\tikzset{mat/.style={draw,fill=white,rectangle,node font=\scriptsize}}
\tikzset{ha/.style={mat,rounded rectangle,rounded rectangle left arc=none}}
\tikzset{haop/.style={mat,rounded rectangle,rounded rectangle right arc=none}}
\tikzset{blackha/.style={mat,rounded rectangle,rounded rectangle left arc=none,font=\color{white},fill=black}}
\tikzset{blackhaop/.style={mat,rounded rectangle,rounded rectangle right arc=none,font=\color{white},fill=black}}
\tikzset{anti/.style={inner sep=0,isosceles triangle,fill=black,draw=black, minimum width=0.75em, node contents={}}}
\tikzset{antiop/.style={anti,shape border rotate=180}}
\tikzset{antisq/.style={inner sep=0,rectangle,fill=black, minimum height=1em, minimum width=0.6em, node contents={}}}
\tikzset{count/.style={above,inner ysep=0.15em,font=\scriptsize}}
\tikzset{axiom/.style={above,font=\small}}
\tikzset{dir/.style={-Latex}}
\tikzset{st/.style={decoration={markings,
    mark={at position 0.5 with {\draw (0, 2pt) to (0, -2pt);}}},
    postaction=decorate}}

\newcommand{\genericcomult}[2]{
  \begin{tikzpicture}
    \node at (1, 0) [ihbase,solid,name=copy,#1];
    \draw[#2] (copy) .. controls (1.25, 0.75) .. (2, 0.75);
    \draw[#2] (0, 0) -- (copy);
    \draw[#2] (copy) .. controls (1.25, -0.75) .. (2, -0.75);
  \end{tikzpicture}
}
\newcommand{\genericcomultn}[2]{
\begin{tikzpicture}
	\begin{pgfonlayer}{nodelayer}
		\node [style=none] (0) at (1.25, -0.75) {};
		\node [style=none] (1) at (1.25, 0.75) {};
		\node [style=#1] (2) at (0, -0) {};
		\node [style=none] (3) at (-1, -0) {};
		\node [style=none] (4) at (1, -0.25) {\tiny $#2$};
		\node [style=none] (5) at (1, 1.25) {\tiny $#2$};
		\node [style=none] (6) at (-0.75, 0.5) {\tiny $#2$};
	\end{pgfonlayer}
	\begin{pgfonlayer}{edgelayer}
		\draw [in=180, out=60, looseness=1.00] (2) to (1.center);
		\draw [in=180, out=-60, looseness=1.00] (2) to (0.center);
		\draw (3.center) to (2);
	\end{pgfonlayer}
\end{tikzpicture}
}
\newcommand{\genericcounit}[2]{
  \tikz \draw[#2] (0, 0) -- (1, 0) node[ihbase,#1, solid];
}
\newcommand{\genericcounitn}[2]{
  \tikz \draw (0, 0) -- node[count] {#2} (1, 0) node[ihbase,#1];
}
\newcommand{\genericmult}[2]{
  \tikz {
    \node at (1,0) (copy) [ihbase,#1,solid];
    \draw[#2] (0,  0.75) .. controls (0.75,  0.75) .. (copy);
    \draw[#2] (0, -0.75) .. controls (0.75, -0.75) .. (copy);
    \draw[#2] (copy) -- (2, 0);
  }
}

\newcommand{\genericunit}[2]{
  \tikz \draw[#2] (0, 0) node[ihbase,#1, solid] -- (1, 0);
}

\newcommand{\Bcomult}{\genericcomult{ihblack}{}}
\newcommand{\Bcomultn}[1]{\genericcomultn{black}{#1}}
\newcommand{\Bcounit}{\genericcounit{ihblack}{}}
\newcommand{\Bcounitn}[1]{\genericcounitn{ihblack}{#1}}
\newcommand{\Bmult}{\genericmult{ihblack}{}}

\newcommand{\Bunit}{\genericunit{ihblack}{}}

\newcommand{\Wmult}{\genericmult{ihwhite}{}}

\newcommand{\Wunit}{\genericunit{ihwhite}{}}

\newcommand{\Wcomult}{\genericcomult{ihwhite}{}}

\newcommand{\Wcounit}{\genericcounit{ihwhite}{}}
\newcommand{\Wcounitn}[1]{\genericcounitn{ihwhite}}

\newcommand{\Gmult}{\genericmult{ihgray}{}}

\newcommand{\Gunit}{\genericunit{ihgray}{}}

\newcommand{\Gcomult}{\genericcomult{ihgray}{}}

\newcommand{\Gcounit}{\genericcounit{ihgray}{}}
\newcommand{\Gcounitn}[1]{\genericcounitn{ihgray}}

\newcommand\idzero{
\tikzset{x=1em, y=2.1ex}
\begin{tikzpicture}
	\begin{pgfonlayer}{nodelayer}
		\node [style=none] (0) at (-0.5, 0.75) {};
		\node [style=none] (1) at (0.75, 0.75) {};
		\node [style=none] (2) at (-0.5, -0.5) {};
		\node [style=none] (3) at (0.75, -0.5) {};
	\end{pgfonlayer}
	\begin{pgfonlayer}{edgelayer}
		\draw [densely dotted] (0.center) to (1.center);
		\draw [densely dotted] (1.center) to (3.center);
		\draw [densely dotted] (3.center) to (2.center);
		\draw [densely dotted] (2.center) to (0.center);
	\end{pgfonlayer}
\end{tikzpicture}}
\tikzset{x=1em, y=1.5ex}
} 

\newcommand{\idone}{
  \tikz \draw (0, 0) -- (1, 0);
}

\newcommand{\sym}{
  \tikz {
    \draw (0,  0.5) .. controls (0.5,  0.5) and (0.5, -0.5) .. (1, -0.5);
    \draw (0, -0.5) .. controls (0.5, -0.5) and (0.5,  0.5) .. (1,  0.5);
  }
}

\newcommand\scalar[1]{
  \tikz {
    \node[ha] (ha) {#1};
    \draw (ha.west) -- ++(-0.75, 0);
    \draw (ha.east) -- ++(0.75, 0);
  }
}

\newcommand\coscalar[1]{
  \tikz {
    \node[haop] (haop) {#1};
    \draw (haop.west) -- ++(-0.75, 0);
    \draw (haop.east) -- ++(0.75, 0);
  }
}

\tikzset{x=1em, y=1.5ex}
}

\tikzset{x=1em, y=1.5ex}
}

\pgfdeclarelayer{edgelayer}
\pgfdeclarelayer{nodelayer}
\pgfsetlayers{background,edgelayer,nodelayer,main}

\definecolor{light-gray}{gray}{.5}
\definecolor{backcolour}{gray}{1}
\tikzstyle{none}=[inner sep=0pt]
\tikzstyle{plain}=[inner sep=0pt]
\tikzstyle{black}=[circle, draw=black, fill=black, inner sep=0pt, minimum size=4pt]
\tikzstyle{black-faded}=[circle, draw=light-gray, fill=light-gray, inner sep=0pt, minimum size=4pt]
\tikzstyle{white}=[circle, draw=black, fill=white, inner sep=0pt, minimum size=4.5pt]
\tikzstyle{white-faded}=[circle, draw=light-gray, fill=white, inner sep=0pt, minimum size=4.5pt]
\tikzstyle{delay}=[fill=black, regular polygon, regular polygon sides=3,rotate=-90, scale=.55]
\tikzstyle{delay-op}=[fill=black, regular polygon, regular polygon sides=3,rotate=90, scale=.55]
\tikzstyle{reg}=[draw, fill=white, rounded rectangle, rounded rectangle left arc=none, minimum height=1.2em, minimum width=1.4em, node font={\scriptsize}]
\tikzstyle{coreg}=[draw, fill=white, rounded rectangle, rounded rectangle right arc=none, minimum height=1.2em, minimum width=1.4em, node font={\scriptsize}]
\tikzstyle{box}=[draw, fill=white, rectangle, minimum height=1.6em, minimum width=1em, node font={\scriptsize}]
\tikzstyle{tallbox}=[draw, fill=white, rectangle, minimum height=3.5em, minimum width=1em, node font={\scriptsize}]
\tikzstyle{rn}=[circle, draw=red, fill=red, inner sep=0pt, minimum size=4pt]
\tikzstyle{place}=[circle, draw=black, fill=white, inner sep=0pt, minimum size=8pt]

\tikzstyle{pl}=[circle,thick,draw=black!75,fill=white,minimum size=9pt]
\tikzstyle{port}=[circle, fill,inner sep=1.2pt]
\tikzstyle{transition}=[rectangle,thick,draw=black!75,
  			  fill=black!20,minimum size=7pt]

\tikzstyle{arrow}=[->]

\newcommand{\regvalue}[1]{
\begin{tikzpicture}
	\begin{pgfonlayer}{nodelayer}
		\node [style=none] (0) at (-1.5, -0) {};
		\node [style=none] (1) at (1.5, -0) {};
		\node [style=reg] (2) at (0, -0) {\scriptsize $x$};
		\node [style=none] (3) at (0, 1.5) {\scriptsize $#1$};
	\end{pgfonlayer}
	\begin{pgfonlayer}{edgelayer}
		\draw (0.center) to (2);
		\draw (2) to (1.center);
	\end{pgfonlayer}
\end{tikzpicture}}
\newcommand{\coregvalue}[1]{
\begin{tikzpicture}
	\begin{pgfonlayer}{nodelayer}
		\node [style=none] (0) at (-1.5, -0) {};
		\node [style=none] (1) at (1.5, -0) {};
		\node [style=coreg] (2) at (0, -0) {\scriptsize $x$};
		\node [style=none] (3) at (0, 1.5) {\scriptsize $#1$};
	\end{pgfonlayer}
	\begin{pgfonlayer}{edgelayer}
		\draw (0.center) to (2);
		\draw (2) to (1.center);
	\end{pgfonlayer}
\end{tikzpicture}}

\newcommand{\diagState}[2]{
\begin{tikzpicture}[background rectangle/.style={fill=backcolour},
    show background rectangle]
	\begin{pgfonlayer}{nodelayer}
		\node [style=reg] (0) at (1.75, -0) {\scriptsize $#2$};
		\node [style=none] (1) at (-0.25, -0) {};
		\node [style=none] (2) at (0, 0.6) {\scriptsize $#1$};
	\end{pgfonlayer}
	\begin{pgfonlayer}{edgelayer}
		\draw (1.center) to (0);
	\end{pgfonlayer}
\end{tikzpicture}
}

\newcommand{\semProc}[2]{
\begin{tikzpicture}[background rectangle/.style={fill=backcolour},
    show background rectangle]
	\begin{pgfonlayer}{nodelayer}
		\node [style=reg] (0) at (2.5, -0) {\tiny $\enc{#2}$};
		\node [style=none] (1) at (-0.25, -0) {};
		\node [style=none] (2) at (0, 0.6) {\scriptsize $#1$};
	\end{pgfonlayer}
	\begin{pgfonlayer}{edgelayer}
		\draw (1.center) to (0);
	\end{pgfonlayer}
\end{tikzpicture}
}

\newcommand{\parProc}[4]{
\begin{tikzpicture}[background rectangle/.style={fill=backcolour},
    show background rectangle]
	\begin{pgfonlayer}{nodelayer}
		\node [style=reg] (0) at (1.5, 1.25) {$\enc{#3}$};
		\node [style=none] (1) at (1.25, 1.25) {};
		\node [style={#1}] (2) at (-0.75, -0) {};
		\node [style=none] (3) at (-2.25, -0) {};
		\node [style=reg] (4) at (1.5, -1.25) {$\enc{#4}$};
		\node [style=none] (5) at (1.25, -1.25) {};
		\node [style=none] (6) at (-2, 0.5) {\scriptsize $#2$};
	\end{pgfonlayer}
	\begin{pgfonlayer}{edgelayer}
		\draw [in=75, out=180, looseness=1.00] (1.center) to (2.center);
		\draw (2) to (3.center);
		\draw [in=-75, out=180, looseness=1.00] (5.center) to (2);
		\draw (1.center) to (0);
		\draw (5.center) to (4);
	\end{pgfonlayer}
\end{tikzpicture}}

\newcommand{\parStates}[4]{
\begin{tikzpicture}[background rectangle/.style={fill=backcolour},
    show background rectangle]
	\begin{pgfonlayer}{nodelayer}
		\node [style=reg] (0) at (1, 1.25) {$#3$};
		\node [style=none] (1) at (0.75, 1.25) {};
		\node [style={#1}] (2) at (-0.75, -0) {};
		\node [style=none] (3) at (-2.25, -0) {};
		\node [style=reg] (4) at (1, -1.25) {$#4$};
		\node [style=none] (5) at (0.75, -1.25) {};
		\node [style=none] (6) at (-2, 0.5) {\scriptsize $#2$};
	\end{pgfonlayer}
	\begin{pgfonlayer}{edgelayer}
		\draw [in=75, out=180, looseness=1.00] (1.center) to (2.center);
		\draw (2) to (3.center);
		\draw [in=-75, out=180, looseness=1.00] (5.center) to (2);
		\draw (1.center) to (0);
		\draw (5.center) to (4);
	\end{pgfonlayer}
\end{tikzpicture}}

\newcommand{\hideProc}[3]{
\begin{tikzpicture}[background rectangle/.style={fill=backcolour},
    show background rectangle]
	\begin{pgfonlayer}{nodelayer}
		\node [style=none] (0) at (1.25, -0.5) {};
		\node [style={#1}] (1) at (-0.75, -0.5) {};
		\node [style=none] (2) at (-2, 0.5) {};
		\node [style=none] (3) at (1.25, 0.5) {};
		\node [style=none] (4) at (-1.5, 1) {\scriptsize $#2$};
		\node [style=reg] (5) at (1.5, -0) {$\enc{#3}$};
	\end{pgfonlayer}
	\begin{pgfonlayer}{edgelayer}
		\draw (0.center) to (1.center);
		\draw (2.center) to (3.center);
	\end{pgfonlayer}
\end{tikzpicture}}

\newcommand{\hideState}[3]{
\begin{tikzpicture}[background rectangle/.style={fill=backcolour},
    show background rectangle]
	\begin{pgfonlayer}{nodelayer}
		\node [style=none] (0) at (1.25, -0.5) {};
		\node [style={#1}] (1) at (-0.75, -0.5) {};
		\node [style=none] (2) at (-2, 0.5) {};
		\node [style=none] (3) at (1.25, 0.5) {};
		\node [style=none] (4) at (-1.5, 1) {\scriptsize $#2$};
		\node [style=reg] (5) at (1, -0) {$#3$};
	\end{pgfonlayer}
	\begin{pgfonlayer}{edgelayer}
		\draw (0.center) to (1.center);
		\draw (2.center) to (3.center);
	\end{pgfonlayer}
\end{tikzpicture}}

\newcommand{\permProc}[3]{
\begin{tikzpicture}[background rectangle/.style={fill=backcolour},
    show background rectangle]
	\begin{pgfonlayer}{nodelayer}
		\node [style=reg] (0) at (1.75, -0) {$\enc{#3}$};
		\node [style=none] (1) at (0, 0.5) {\scriptsize $#1$};
		\node [style=box] (2) at (-1.25, -0) {\scriptsize $\overline{#2}$};
		\node [style=none] (3) at (-3.25, -0) {};
		\node [style=none] (4) at (-2.75, 0.5) {\scriptsize $#1$};
	\end{pgfonlayer}
	\begin{pgfonlayer}{edgelayer}
		\draw (3.center) to (2.center);
		\draw (2.center) to (0);
	\end{pgfonlayer}
\end{tikzpicture}}

\newcommand{\permState}[3]{
\begin{tikzpicture}[background rectangle/.style={fill=backcolour},
    show background rectangle]
	\begin{pgfonlayer}{nodelayer}
		\node [style=reg] (0) at (1.25, -0) {$#3$};
		\node [style=none] (1) at (0, 0.5) {\scriptsize $#1$};
		\node [style=box] (2) at (-1.25, -0) {\scriptsize $\overline{#2}$};
		\node [style=none] (3) at (-3.25, -0) {};
		\node [style=none] (4) at (-2.75, 0.5) {\scriptsize $#1$};
	\end{pgfonlayer}
	\begin{pgfonlayer}{edgelayer}
		\draw (3.center) to (2.center);
		\draw (2.center) to (0);
	\end{pgfonlayer}
\end{tikzpicture}}

\newcommand{\newvarProc}[3]{
\begin{tikzpicture}[background rectangle/.style={fill=backcolour},
    show background rectangle]
	\begin{pgfonlayer}{nodelayer}
		\node [style=none] (0) at (-1.75, -0.4) {};
		\node [style={#1}] (1) at (-0.5, -0.4) {};
		\node [style=none] (2) at (-0.75, 0.5) {};
		\node [style=reg] (3) at (1.5, -0) {$\enc{#3}$};
		\node [style=none] (4) at (-1.25, 1) {\scriptsize $#2$};
		\node [style=none] (5) at (-1.75, 0.5) {};
	\end{pgfonlayer}
	\begin{pgfonlayer}{edgelayer}
		\draw (0.center) to (1.center);
		\draw [in=180, out=0, looseness=1.25] (2.center) to (3.center);
		\draw (5.center) to (2.center);
	\end{pgfonlayer}
\end{tikzpicture}}

\newcommand{\newvarState}[3]{
\begin{tikzpicture}[background rectangle/.style={fill=backcolour},
    show background rectangle]
	\begin{pgfonlayer}{nodelayer}
		\node [style=none] (0) at (-1.75, -0.4) {};
		\node [style={#1}] (1) at (-0.5, -0.4) {};
		\node [style=none] (2) at (-0.75, 0.5) {};
		\node [style=reg] (3) at (1.25, -0) {$#3$};
		\node [style=none] (4) at (-1.25, 1) {\scriptsize $#2$};
		\node [style=none] (5) at (-1.75, 0.5) {};
	\end{pgfonlayer}
	\begin{pgfonlayer}{edgelayer}
		\draw (0.center) to (1.center);
		\draw [in=180, out=0, looseness=1.25] (2.center) to (3.center);
		\draw (5.center) to (2.center);
	\end{pgfonlayer}
\end{tikzpicture}}

\newcommand{\stateplusone}[2]{
\begin{tikzpicture}[background rectangle/.style={fill=backcolour},
    show background rectangle]
	\begin{pgfonlayer}{nodelayer}
		\node [style=none] (0) at (0.75, -0.5) {};
		\node [style=none] (1) at (-0.75, 0.5) {};
		\node [style=none] (2) at (0.75, 0.5) {};
		\node [style=none] (3) at (-0.5, 1) {\scriptsize $#1$};
		\node [style=reg] (4) at (1.25, -0) {$#2$};
		\node [style=none] (5) at (-0.75, -0.5) {};
	\end{pgfonlayer}
	\begin{pgfonlayer}{edgelayer}
		\draw (1.center) to (2.center);
		\draw (0.center) to (5.center);
	\end{pgfonlayer}
\end{tikzpicture}}

\newcommand{\procplusone}[2]{
\begin{tikzpicture}[background rectangle/.style={fill=backcolour},
    show background rectangle]
	\begin{pgfonlayer}{nodelayer}
		\node [style=none] (0) at (0.75, -0.5) {};
		\node [style=none] (1) at (-0.75, 0.5) {};
		\node [style=none] (2) at (0.75, 0.5) {};
		\node [style=none] (3) at (-0.5, 1) {\scriptsize $#1$};
		\node [style=reg] (4) at (1.7, -0) {$#2$};
		\node [style=none] (5) at (-0.75, -0.5) {};
	\end{pgfonlayer}
	\begin{pgfonlayer}{edgelayer}
		\draw (1.center) to (2.center);
		\draw (0.center) to (5.center);
	\end{pgfonlayer}
\end{tikzpicture}}

\newcommand{\vardisconnect}[4]{
\begin{tikzpicture}[background rectangle/.style={fill=backcolour},
    show background rectangle]
	\begin{pgfonlayer}{nodelayer}
		\node [style=none] (0) at (1, -0.25) {};
		\node [style=none] (1) at (-0.75, -1) {};
		\node [style=none] (2) at (-1.75, -0) {};
		\node [style=reg] (3) at (1.25, -0) {$#4$};
		\node [style=none] (4) at (1, 0.25) {};
		\node [style=#1] (5) at (-0.5, -0) {};
		\node [style=none] (6) at (-1, 1.5) {\scriptsize $#3 - 1$};
		\node [style=none] (7) at (-1.75, 1) {};
		\node [style=none] (8) at (-0.75, 1) {};
		\node [style=none] (9) at (-1.75, -1) {};
		\node [style=none] (10) at (-1, -1.5) {\scriptsize $#2 - #3$};
	\end{pgfonlayer}
	\begin{pgfonlayer}{edgelayer}
		\draw (2.center) to (5);
		\draw [in=180, out=0, looseness=1.25] (8.center) to (4.center);
		\draw (7.center) to (8.center);
		\draw [in=180, out=0, looseness=1.25] (1.center) to (0.center);
		\draw (9.center) to (1.center);
	\end{pgfonlayer}
\end{tikzpicture}}

\newcommand{\diagEffect}[2]{
\begin{tikzpicture}[show background rectangle, {background rectangle/.style}={fill=backcolour}]
	\begin{pgfonlayer}{nodelayer}
		\node [style=coreg] (0) at (-0.25, -0) {\scriptsize $#2$};
		\node [style=none] (1) at (1.75, -0) {};
		\node [style=none] (2) at (1.5, 0.6) {\scriptsize $#1$};
	\end{pgfonlayer}
	\begin{pgfonlayer}{edgelayer}
		\draw (1.center) to (0);
	\end{pgfonlayer}
\end{tikzpicture}}

\begin{document}

\maketitle

\begin{abstract}
Turi and Plotkin's bialgebraic semantics is an abstract approach to specifying the operational semantics of a system, by means of a distributive law between its syntax (encoded as a monad) and its dynamics (an endofunctor). This setup is instrumental in showing that a semantic specification (a coalgebra) satisfies desirable properties: in particular, that it is compositional. 

In this work, we use the bialgebraic approach to derive well-behaved structural operational semantics of \emph{string diagrams}, a graphical syntax that is increasingly used in the study of interacting systems across different disciplines. Our analysis relies on representing the two-dimensional operations underlying string diagrams in various categories as a monad, and their bialgebraic semantics in terms of a distributive law for that monad. 

As a proof of concept, we provide bialgebraic compositional semantics for a versatile string diagrammatic language which has been used to model both signal flow graphs (control theory) and Petri nets (concurrency theory). Moreover, our approach reveals a correspondence between two different interpretations of the Frobenius equations on string diagrams and two synchronisation mechanisms for processes, \`{a} la Hoare and \`{a} la Milner.
 \end{abstract}

\section{Introduction}

Starting from the seminal works of Hoare and Milner, there was an explosion~\cite{Hoare:1978:CSP,Milner:1982:CCS,Sangiorgi:2001:PiCalculus,Cardelli2000,Castagna2005} of interest in process calculi: formal languages for specifying and reasoning about concurrent systems.
The beauty of the approach, and often the focus of research, lies in \emph{compositionality}: essentially, the behaviour of composite systems---usually understood as some kind of process equivalence, the most famous of which is bisimilarity---ought to be a function of the behaviour of its components. The central place of compositionality led to the study of syntactic formats for semantic specifications~\cite{de1985higher,bloom1995bisimulation,groote1993transition}; succinctly stated, syntactic operations with semantics defined using such formats are homomorphic wrt the semantic space of behaviours.

Another thread of concurrency theory research~\cite{Reisig1985,jensen2013coloured,harel1987statecharts} uses graphical formalisms, such as Petri nets. These often have the advantage of highlighting connectivity, distribution and the communication topology of systems. They tend to be popular with practitioners in part because of their intuitive and human-readable depictions, an aspect that should not be underestimated: indeed, pedagogical texts introducing CCS~\cite{Hoare:1978:CSP} and CSP~\cite{Milner:1982:CCS} often resort to pictures that give intuition about topological aspects of syntactic specifications. However, compositionality has not---historically---been a principal focus of research.  

In this paper we propose a framework that allows us to eat our cake and have it too. We use \emph{string diagrams}~\cite{Selinger2009} which have an intuitive graphical rendering, but also come with algebraic operations for composition. String diagrams combine the best of both worlds: they are a (2-dimensional) syntax, but also convey important topological information about the systems they specify. They have been used in recent years to give compositional accounts of quantum circuits\cite{Abramsky2004, Coecke2008}, signal flow graphs~\cite{Bonchi2015,BaezErbele-CategoriesInControl,Fong2015}, Petri nets~\cite{BonchiHPSZ19}, and electrical circuits~\cite{Fong2013,Ghica13}, amongst several other applications.

Our main contribution is the adaptation of Turi and Plotkin's bialgebraic semantics (\emph{abstract GSOS})~\cite{turi1997towards,klin2011bialgebras} for string diagrams. By doing so, we provide a principled justification and theoretical framework for giving definitions of operational semantics to the generators and operations of string diagrams, which are those of monoidal categories. More precisely we deal with string diagrams for symmetric monoidal categories which organise themselves as arrows of a particularly simple and well-behaved class known as \emph{props}. Similar operational definitions have been used in the work on the algebra of Span(Graph)~\cite{Katis1997a}, tile logic~\cite{Gadducci2000}, the wire calculus~\cite{wirecalc} and recent work on modelling signal flow graphs and Petri nets~\cite{Bonchi2015,BonchiHPSZ19}. In each case, semantics was given either monolithically or via a set of SOS rules. The link with bialgebraic framework---developed in this paper---provides us a powerful theoretical tool that (i) justifies these operational definitions and (ii) guarantees compositionality.
%

In a nutshell, in the bialgebraic approach, the syntax of a language is the initial algebra (the algebra of terms) $T_\Sigma$  for a signature functor $\Sigma$. A certain kind of distributive law, an \emph{abstract GSOS} specification~\cite{turi1997towards}, induces a coalgebra (a state machine) $\beta \colon T_\Sigma \to \Fbeh T_\Sigma$ capturing the operational semantics of the language. The final $\Fbeh$-coalgebra $\Omega$ provides the denotational universe: intuitively, the space of all possible behaviours. The unique coalgebra map $\sem{\cdot}_\beta \colon T_\Sigma \to \Omega$ represents the denotational semantics assigning to each term its behaviour.
 \begin{equation}\label{eq:diagintro}
\vcenter{
\xymatrix@C=50pt{
T_\Sigma \ar@{-->}[r]^{\sem{\cdot}_{\beta}} \ar[d]_{\beta} & \Omega \ar[d]\\
\Fbeh(T_\Sigma) \ar[r]_{\Fbeh(\sem{\cdot}_{\beta})} & \Fbeh(\Omega)
}
} 
\end{equation}
The crucial observation is that \eqref{eq:diagintro} lives in the category of $\Sigma$-algebras: $\Omega$ also carries a $\Sigma$-algebra structure and the denotational semantics is an algebra homomorphism. This means that the behaviour of a composite system is determined by the behaviour of the components, e.g.\ $\sem{s+t}= \sem{s}+\sem{t}$, for an operation $+$ in $\Sigma$.

We show that the above framework can be adapted to the algebra of string diagrams. The end result is a picture analogous to \eqref{eq:diagintro}, but living in the category of props and prop morphism. As a result, the denotational map is a prop morphism, and thus guarantees compositionality with respect to sequential and parallel composition of string diagrams.

Adapting the bialgebraic approach to the 2-dimensional syntax of props requires some technical work since, e.g.\ the composition operation of monoidal categories is a dependent operation. For this reason we need to adapt the usual notion of a syntax endofunctor on the category of sets; instead we work in a category $\SIG$ whose objects are spans $\N \tl{} \Sigma \tr{} \N$, with the two legs giving the number of dangling wires on the left and right of each diagram. The operations of props are captured as a $\SIG$-endofunctor, which yields string-diagrams-as-initial-algebra, and a quotient of the resulting free monad, whose algebras are precisely~props.

In addition to the basic laws of props, we also consider the further imposition of the equations of special Frobenius algebras. 
We illustrate the role of this algebraic structure with our running example, a string diagrammatic process calculus~$\Circ{\Rig}$ that has two Frobenius structures and can be equipped with two different semantics, one which provides a compositional account of signal flow graphs for linear time invariant dynamical systems~\cite{Bonchi2015}, and one which is a compositional account of Petri nets~\cite{BonchiHPSZ19}.

We conclude with an observation that ties our work back to classical concepts of process calculi and show that the two Frobenius structures of~$\Circ{\Rig}$ are closely related to two different, well-known synchronisation patterns, namely those of Hoare's CSP~\cite{Hoare:1978:CSP} and Milner's CCS~\cite{Milner:1982:CCS}.

\emph{Structure of the paper.} In \S\ref{sec:motivating} we introduce our main example and recall some preliminaries, followed by a recapitulation of bialgebraic approach in \S\ref{sec:background}. We develop the technical aspects of string-diagrams-as-syntax in \S\ref{sec:hierarchy} and adapt the bialgebraic approach in \S\ref{sec:SOS}. Finally, we exhibit the connection with classical synchronisation mechanisms in \S\ref{BWFrob} and conclude in \S\ref{sec:conclusion}.

\section{Motivating Example}\label{sec:motivating}
\newcommand{\registerk}{\regvalue{k}} 
\newcommand{\coregisterk}{\coregvalue{k}} 
As our motivating example, we recall from \cite{Bonchi2014b,BonchiSZ17,BonchiHPSZ19} a basic language $\Circ{\Rig}$ given by the grammar below. Values $k$ in $\registerk$ and $\scalar{k}$ range over elements of a given semiring~$\Rig$.
\begin{align}\label{eq:syntax1}
\hspace{-.5cm} c,d \ ::= \ & \Bcounit \mid \Bcomult \mid \registerk \mid \scalar{k} \mid \Wmult \mid \Wunit \mid \Bunit \mid \Bmult \mid  \coregisterk \mid \coscalar{k} \mid  \Wcomult \mid \Wcounit  \mid    \\
 \label{eq:syntax2} & \mid \idzero \mid \idone \mid \sym   \mid c\poi d \mid c \tns d 
\end{align}

The language does not feature variables; on the other hand, a simple sorting discipline is necessary. A sort is a pair $\sort{n}{m}$, with $n, m \in \N$.
Henceforth we will consider only terms sortable according to the rules in Figure~\ref{sort}. An easy induction confirms uniqueness of sorting. 

\begin{figure}[t]
$$\begin{aligned}
\typeJudgment{}{\Bcounit}{\sort{1}{0}} \quad
\typeJudgment{}{\Bcomult}{\sort{1}{2}} \quad
\typeJudgment{}{\registerk}{\sort{1}{1}} \quad
\typeJudgment{}{\scalar{k}}{\sort{1}{1}} \quad
\typeJudgment{}{\Wmult}{\sort{2}{1}} \quad
\typeJudgment{}{\Wunit}{\sort{0}{1}} \\
\typeJudgment{}{\Bunit}{\sort{0}{1}} \quad
\typeJudgment{}{\Bmult}{\sort{2}{1}} \quad
\typeJudgment{}{\coregvalue{k}}{\sort{1}{1}} \quad
\typeJudgment{}{\coscalar{k}}{\sort{1}{1}} \quad
\typeJudgment{}{\Wcomult}{\sort{1}{2}} \quad
\typeJudgment{}{\Wcounit}{\sort{1}{0}}
\\
 \typeJudgment{}{\idzero}{\sort{0}{0}} \quad
 \typeJudgment{}{\idone}{\sort{1}{1}} \quad
 \typeJudgment{}{\sym}{\sort{2}{2}} \quad
\reductionRule{ \typeJudgment{}{c}{\sort{k_1}{k_2}} \quad \typeJudgment{}{d}{\sort{k_2}{k_3}} }
{ \typeJudgment{}{c\poi d}{\sort{k_1}{k_3}} }\quad
\reductionRule{ \typeJudgment{}{c}{\sort{k_1}{l_1}} \quad \typeJudgment{}{d}{\sort{k_2}{l_2}} }
{ \typeJudgment{}{c \tns d}{\sort{k_1+k_2}{l_1+l_2}} }
\end{aligned}$$
\caption{Sorting discipline for $\Circ{\Rig}$}\label{sort}
\end{figure}


\begin{figure}[t]
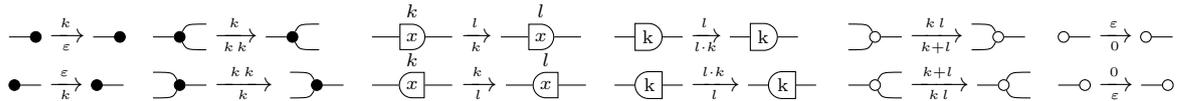

\[\begin{array}{llllll}
\Bcounit \dtrans{k}{\varepsilon} \Bcounit &
\Bcomult \dtrans{\,k\,}{\,k \, k\, } \Bcomult &
\registerk \dtrans{l}{k} \regvalue{l} &
\scalar{k} \dtrans{l}{l\cdot k} \scalar{k} &
\, \Wmult \dtrans{\,k\,l\,}{k+l} \Wmult &
\Wunit \dtrans{\varepsilon}{0} \Wunit \\
\Bunit \dtrans{\varepsilon}{k} {\Bunit} &
\Bmult \dtrans{\, k\, k\,}{k} \ \Bmult &
\coregisterk \dtrans{k}{l} \coregvalue{l} &
\coscalar{k} \dtrans{l\cdot k} {l}\coscalar{k} &
\, \Wcomult \dtrans{k+l}{\,k\,l\,} \Wcomult &
\Wcounit \dtrans{0} {\varepsilon}\Wcounit \\
\end{array}\]
\caption{Structural Operational Semantics for the generators of $\Circ{\Rig}$. Intuitively, from left to right, these are elementary connectors modelling discard, copy, one-place register, multiplication by a scalar, addition, and the constant zero.\label{sos1}}
\end{figure}
\begin{figure}[t]
\[\begin{array}{lllll}
\derivationRule{c\dtrans{\bm{a}}{\bm{b}} c' \quad d\dtrans{\bm{b}}{\bm{c}} d'}
{c\poi d \dtrans {\bm{a}}{\bm{c}} c' \poi d'}{\lambdaseqcomp} &
\derivationRule{\vspace{30pt}}{\idone  \dtrans{k}{k} \idone}{\lambdaid} &
\derivationRule{s\dtrans{\bm{a}_1}{\bm{b}_1} c'\quad d\dtrans{\bm{a}_2}{\bm{b}_2} d'}
{c\tns d \dtrans{\, \bm{a}_1 \bm{a}_2}{ \bm{b}_1 \bm{b}_2} \ d'\tns d'}{\lambdaparcomp}&
\derivationRule{}{\idzero \dtrans{\varepsilon}{\varepsilon} \idzero}{\lambdazero} &
\derivationRule{}{\sym \dtrans{\,k\,l\,}{\,l\,k\,}\ \sym}{\lambdasym}
\end{array}\]
\caption{Structural operational semantics for the operations of $\Circ{\Rig}$.\label{sos2}}
\end{figure}

The operational meaning of terms is defined recursively by the structural rules in Figs.~\ref{sos1} and~\ref{sos2}
where $k,l$ range over~$\Rig$ and $\bm{a},\bm{b},\bm{c}$ over $\Rig^{\star}$, the set of words over $\Rig$. We denote the empty word by $\varepsilon$ and concatenation of $\bm{a},\bm{b}$ by $\bm{a}\bm{b}$. As expected $+,\cdot$ and $0$ denote respectively the sum, the product and zero of the semiring~$\Rig$. 
For any term $c \colon \sort{n}{m}$, the rules yield a labelled transition system where each transition has form $c \dtrans{\bm{a}}{\bm{b}}d$. By induction, it is immediate that $d$ has the same sort as $c$, the word $\bm{a}$ has length $n$, and $\bm{b}$ has length $m$. 

\medskip

Our chief focus in this paper is the study of semantics specifications of the kind given in Figs.~\ref{sos1} and~\ref{sos2}. So far, the technical difference with typical GSOS examples~\cite{bloom1995bisimulation} 
is the presence of a sorting discipline. A more significant difference, which we will now highlight, is that sorted terms are considered up-to the laws of symmetric monoidal categories. As such, they are ``2-dimensional syntax'' and enjoy a pictorial representation in terms of string diagrams.

%
%

\subsection{From Terms to String Diagrams} \label{sec:stringdiag}


In \eqref{eq:syntax1}-\eqref{eq:syntax2} we purposefully used a graphical rendering of the components. Indeed, terms of $\Circ{\Rig}$ are usually represented graphically, according to the convention that $c \poi c'$ is drawn $\lower9pt\hbox{$\includegraphics[height=.8cm]{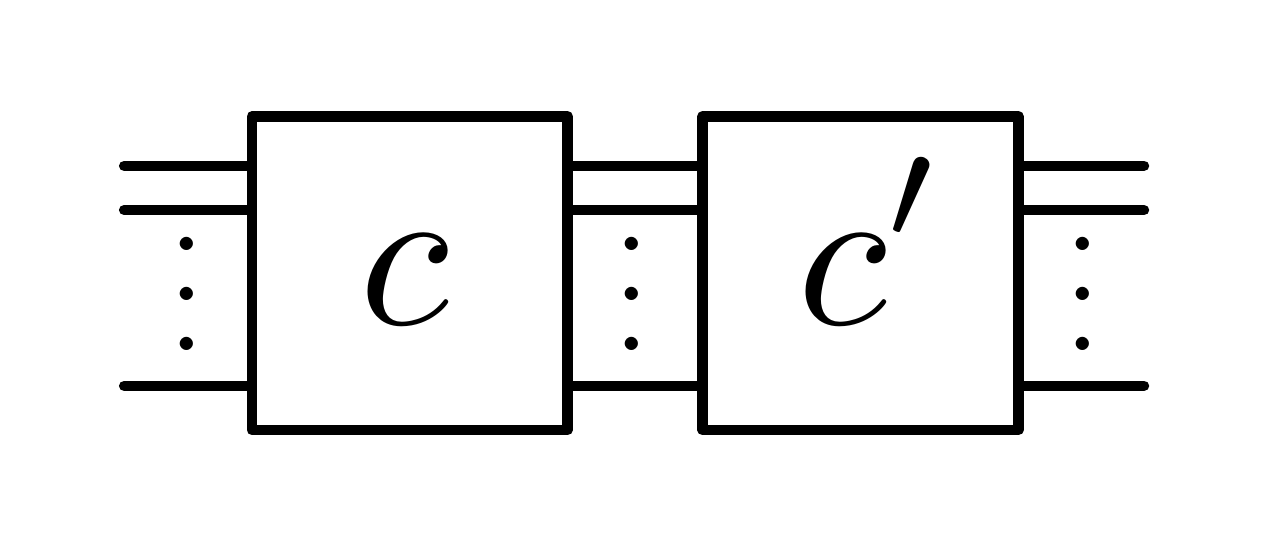}$}$ and $c \tns c'$ is drawn $\lower15pt\hbox{$\includegraphics[height=1.2cm]{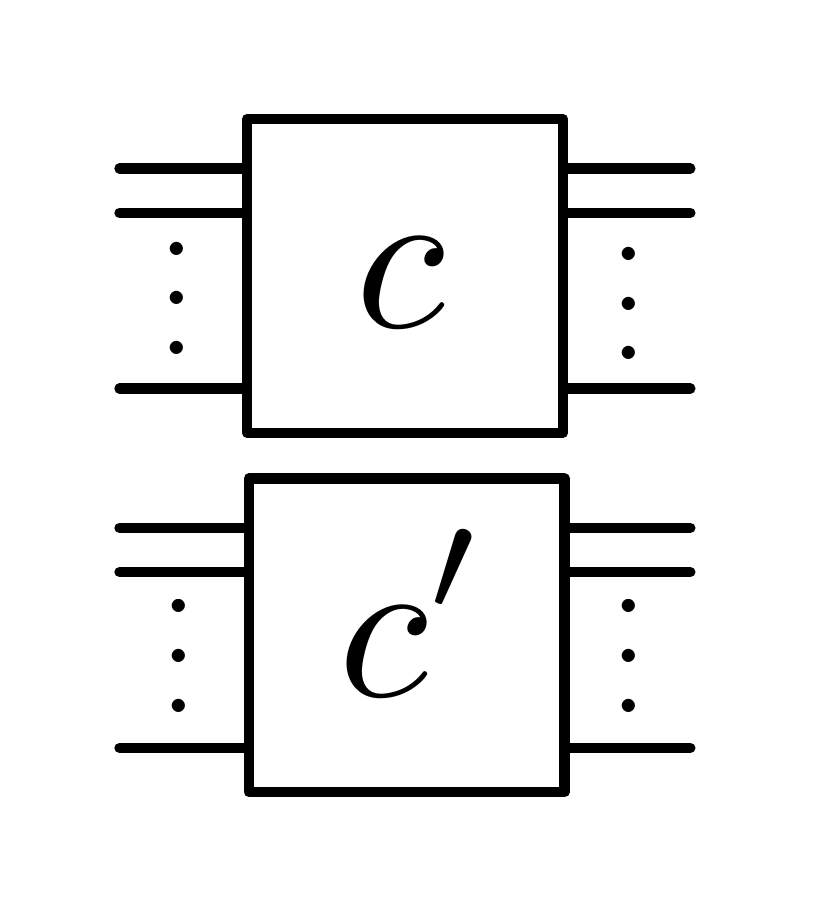}$}$. For instance, the term $(\,(\Bunit \poi \Bcomult )\tns \idone \,) \, \poi \, (\,(\idone \tns (\Wmult \poi \registerk \poi \Wcomult)) \,) \, \poi \, (\,( \Bmult \poi \Bcounit)\tns \idone \,)$ is depicted as the following diagram. 
\begin{equation}\label{petriPlace} p_k::=
\tikzset{x=1em, y=2.1ex}
\begin{tikzpicture}[show background rectangle, {background rectangle/.style}={fill=backcolour}]
	\begin{pgfonlayer}{nodelayer}
		\node [style=none] (0) at (2.5, 1.75) {};
		\node [style=none] (1) at (-2, 0.25) {};
		\node [style=none] (2) at (2.5, 0.25) {};
		\node [style=white] (3) at (-1, -0.5) {};
		\node [style=white] (4) at (1.5, -0.5) {};
		\node [style=none] (5) at (2.5, -1.25) {};
		\node [style=none] (6) at (-2, -1.25) {};
		\node [style=none] (7) at (5, -1.25) {};
		\node [style=none] (8) at (-4.5, -1.25) {};
		\node [style=none] (9) at (-2, 1.75) {};
		\node [style=reg] (10) at (0.25, -0.5) {$x$};
		\node [style=black] (11) at (-3, 1) {};
		\node [style=black] (12) at (-4, 1) {};
		\node [style=black] (13) at (3.5, 1) {};
		\node [style=black] (14) at (4.5, 1) {};
		\node [style=none] (15) at (0.25, 0.5) {\scriptsize $k$};
	\end{pgfonlayer}
	\begin{pgfonlayer}{edgelayer}
		\draw [in=180, out=75, looseness=1.00] (4) to (2.center);
		\draw [in=0, out=105, looseness=1.00] (3) to (1.center);
		\draw [in=-105, out=0, looseness=1.00] (6.center) to (3);
		\draw [in=180, out=-75, looseness=1.00] (4) to (5.center);
		\draw (6.center) to (8.center);
		\draw (5.center) to (7.center);
		\draw (9.center) to (0.center);
		\draw (3) to (10.center);
		\draw (10.center) to (4);
		\draw (12) to (11);
		\draw [in=180, out=-60, looseness=1.00] (11) to (1.center);
		\draw [in=-120, out=0, looseness=1.00] (2.center) to (13);
		\draw (13) to (14);
		\draw [in=120, out=0, looseness=1.00] (0.center) to (13);
		\draw [in=180, out=60, looseness=1.00] (11) to (9.center);
	\end{pgfonlayer}
\end{tikzpicture}}
\tikzset{x=1em, y=1.5ex}
 \end{equation}
Given this graphical convention, 
a sort gives the number of dangling wires on each side of the diagram induced by a term. A transition $c \dtrans{\bm{a}}{\bm{b}}d$ means that $c$ may evolve to $d$ when the values on the dangling wires on the left are $\bm{a}$ and those on the right are $\bm{b}$. When $\Rig$ is the natural numbers, the diagram in \eqref{petriPlace} behaves as a place of a Petri nets containing $k$ tokens: any number of tokens can be inserted  from its left and at most $k$ tokens can be removed from its right. Indeed, by the rules in Figs.~\ref{sos1} and~\ref{sos2}, $p_k\dtrans{i}{o}p_{k'}$ iff $o\leq k$ and $k'=i+k-o$.

\medskip

The graphical notation is appealing as it highlights connectivity and the capability for
resource exchange. However, syntactically different terms 
can yield the same diagram, e.g. $ (\Bunit \tns \idone ) \poi ( \Bcomult \tns \idone) \poi (\idone \tns \Wmult) \poi (\idone \tns \registerk) \poi (\idone \tns \Wcomult) \poi  ( \Bmult \tns \idone) \poi (\Bcounit \tns \idone) $ also yields~\eqref{petriPlace}. 
Indeed, one defines diagrams to be terms modulo \emph{structural congruence}, denoted by $\stcong$. This is the  smallest congruence over terms generated by the equations of strict symmetric monoidal categories (SMCs):
\begin{eqnarray} 
 \qquad  \label{eq:equationCategories}
(f \tns g) \tns h   \, \stcong \,  f \tns (g \tns h)  \qquad (\epsilon \tns f)  \,  \stcong \,  f  \quad (f \tns \epsilon) \, \stcong \,  f\qquad \sigma_{1,1} \poi \sigma_{1,1} \, \stcong \,  id_{2} \label{eq:equationPROPs} \\
 (f \poi g) \tns ( h \poi  i)   \, \stcong \, (f \tns h) \poi ( g \tns  i)  \qquad (f \poi g) \poi h   \, \stcong \, f \poi (g \poi h)  \qquad (f \poi id_m)   \, \stcong \, f
\label{eq:equationPROPs2} \\
 (id_n \! \poi \! f) \, \stcong \, f \qquad (\sigma_{1,n} \!\poi\! ( f \tns id_1)) \stcong (id_1 \tns f) \!\poi\! \sigma_{1,m} \qquad (\sigma_{n,1} \!\poi\! ( id_1 \tns g)) \stcong (g \tns id_1) \!\poi\! \sigma_{m,1}  \label{eq:equationPROPs3}
\end{eqnarray}
where identities $\typeJudgment{}{id_n}{\sort{n}{n}}$ and symmetries $\typeJudgment{}{\sigma_{n,m}}{\sort{n+m}{m+n}}$ can be recursively defined starting from $id_0 \df \idzero$ and $\sigma_{1,1}\df \sym$. 
Therefore, sorted diagrams $\typeJudgment{}{c}{\sort{n}{m}}$ are the arrows $n \to m$ of an SMC with objects the natural numbers, also called a prop~\cite{MacLane1965}.


\subsection{Frobenius Bimonoids} \label{sec:axioms}
We will also consider additional algebraic structure for the black ($\Bcounit$, $\Bcomult$, $\Bunit$, $\Bmult$) and the white ($\Wunit$, $\Wmult$, $\Wcounit$, $\Wcomult$) components. When $\Rig$ is the field of reals, $\Circ{\Rig}$ models linear dynamical systems \cite{BonchiSZ17,Fong2015,BaezErbele-CategoriesInControl} and both the black and the white structures form special Frobenius bimonoids, meaning the axioms of Fig.~\ref{fig:Frob} hold, replacing the gray circles by either black or white. When $\Rig$ is the semiring of natural numbers,  $\Circ{\Rig}$ models Petri nets \cite{BonchiHPSZ19} and only the black structure 
satisfies these equations.
In~\S~\ref{BWFrob}, we shall see that the black Frobenius structure gives rise to the synchronisation mechanism used by Hoare in CSP~\cite{Hoarebook}, while the white Frobenius structure to that used by Milner in CCS~\cite{Milner:1982:CCS}.

\begin{figure}
\begin{align*}
 
\tikzset{x=1em, y=2.1ex}
\begin{tikzpicture}
	\begin{pgfonlayer}{nodelayer}
		\node [style=none] (0) at (1, 1.25) {};
		\node [style=none] (1) at (2.25, -1) {};
		\node [style=black-faded] (2) at (1, -0.25) {};
		\node [style=none] (3) at (2.25, 1.25) {};
		\node [style=black-faded] (4) at (-0.25, 0.5) {};
		\node [style=none] (5) at (-1.5, 0.5) {};
		\node [style=none] (6) at (2.25, 0.5) {};
	\end{pgfonlayer}
	\begin{pgfonlayer}{edgelayer}
		\draw (4) to (5.center);
		\draw [bend right, looseness=1.00] (0.center) to (4);
		\draw [bend left, looseness=1.00] (1.center) to (2);
		\draw [bend right, looseness=1.00] (6.center) to (2);
		\draw [in=180, out=-60, looseness=1.00] (4) to (2);
		\draw (3.center) to (0.center);
	\end{pgfonlayer}
\end{tikzpicture}}
\tikzset{x=1em, y=1.5ex}
\;\stcong\; 
\tikzset{x=1em, y=2.1ex}
\begin{tikzpicture}
	\begin{pgfonlayer}{nodelayer}
		\node [style=black-faded] (0) at (-0.5, -0) {};
		\node [style=none] (1) at (-1.75, -0) {};
		\node [style=none] (2) at (0.75, -0.75) {};
		\node [style=black-faded] (3) at (0.75, 0.75) {};
		\node [style=none] (4) at (2, 1.5) {};
		\node [style=none] (5) at (2, -0) {};
		\node [style=none] (6) at (2, -0.75) {};
	\end{pgfonlayer}
	\begin{pgfonlayer}{edgelayer}
		\draw (0) to (1.center);
		\draw [bend left, looseness=1.00] (2.center) to (0);
		\draw [bend right, looseness=1.00] (4.center) to (3);
		\draw [bend left, looseness=1.00] (5.center) to (3);
		\draw [in=180, out=60, looseness=1.00] (0) to (3);
		\draw (6.center) to (2.center);
	\end{pgfonlayer}
\end{tikzpicture}}
\tikzset{x=1em, y=1.5ex}
\qquad  
\tikzset{x=1em, y=2.1ex}
\begin{tikzpicture}
	\begin{pgfonlayer}{nodelayer}
		\node [style=black-faded] (0) at (0, -0) {};
		\node [style=none] (1) at (0.75, -0.75) {};
		\node [style=none] (2) at (0.75, 0.75) {};
		\node [style=none] (3) at (-1.25, -0) {};
		\node [style=none] (4) at (2.75, -0.75) {};
		\node [style=none] (5) at (2.75, 0.75) {};
	\end{pgfonlayer}
	\begin{pgfonlayer}{edgelayer}
		\draw (0) to (3.center);
		\draw [in=75, out=180, looseness=1.00] (2.center) to (0);
		\draw [in=180, out=-75, looseness=1.00] (0) to (1.center);
		\draw [in=0, out=180, looseness=1.00] (4.center) to (2.center);
		\draw [in=180, out=0, looseness=0.75] (1.center) to (5.center);
	\end{pgfonlayer}
\end{tikzpicture}}
\tikzset{x=1em, y=1.5ex}
\stcong 
\tikzset{x=1em, y=2.1ex}
\begin{tikzpicture}
	\begin{pgfonlayer}{nodelayer}
		\node [style=black-faded] (0) at (0, -0) {};
		\node [style=none] (1) at (-1.25, -0) {};
		\node [style=none] (2) at (1.5, -0.75) {};
		\node [style=none] (3) at (1.5, 0.75) {};
	\end{pgfonlayer}
	\begin{pgfonlayer}{edgelayer}
		\draw (0) to (1.center);
		\draw [bend left, looseness=1.00] (2.center) to (0);
		\draw [bend right, looseness=1.00] (3.center) to (0);
	\end{pgfonlayer}
\end{tikzpicture}}
\tikzset{x=1em, y=1.5ex}
\qquad  
\tikzset{x=1em, y=2.1ex}
\begin{tikzpicture}
	\begin{pgfonlayer}{nodelayer}
		\node [style=black-faded] (0) at (-0.75, 0) {};
		\node [style=black-faded] (1) at (0.5, 0.75) {};
		\node [style=none] (2) at (0.5, -0.75) {};
		\node [style=none] (3) at (-2, 0) {};
		\node [style=none] (4) at (1.5, -0.75) {};
	\end{pgfonlayer}
	\begin{pgfonlayer}{edgelayer}
		\draw (0) to (3.center);
		\draw [bend left, looseness=1.00] (2.center) to (0);
		\draw [in=180, out=60, looseness=1.00] (0) to (1);
		\draw (4.center) to (2.center);
	\end{pgfonlayer}
\end{tikzpicture}}
\tikzset{x=1em, y=1.5ex}
\stcong\;
\tikzset{x=1em, y=2.1ex}
\begin{tikzpicture}
	\begin{pgfonlayer}{nodelayer}
		\node [style=black-faded] (0) at (0.25, 0) {};
		\node [style=black-faded] (1) at (-1, 0.75) {};
		\node [style=none] (2) at (-1, -0.75) {};
		\node [style=none] (3) at (1.5, 0) {};
		\node [style=none] (4) at (-2, -0.75) {};
	\end{pgfonlayer}
	\begin{pgfonlayer}{edgelayer}
		\draw (0) to (3.center);
		\draw [bend right, looseness=1.00] (2.center) to (0);
		\draw [in=0, out=120, looseness=1.00] (0) to (1);
		\draw (4.center) to (2.center);
	\end{pgfonlayer}
\end{tikzpicture}}
\tikzset{x=1em, y=1.5ex}
 \stcong 
\tikzset{x=1em, y=2.1ex}
\begin{tikzpicture}
	\begin{pgfonlayer}{nodelayer}
		\node [style=none] (0) at (1.75, -0) {};
		\node [style=none] (1) at (-1.75, -0) {};
	\end{pgfonlayer}
	\begin{pgfonlayer}{edgelayer}
		\draw (0.center) to (1.center);
	\end{pgfonlayer}
\end{tikzpicture}}
\tikzset{x=1em, y=1.5ex}
\;\stcong
\tikzset{x=1em, y=2.1ex}
\begin{tikzpicture}
	\begin{pgfonlayer}{nodelayer}
		\node [style=none] (0) at (0, 0.75) {};
		\node [style=black-faded] (1) at (1.25, -0) {};
		\node [style=none] (2) at (0, -0.75) {};
		\node [style=none] (3) at (0, 0.75) {};
		\node [style=black-faded] (4) at (-1.25, -0) {};
		\node [style=none] (5) at (0, -0.75) {};
		\node [style=none] (6) at (2.5, -0) {};
		\node [style=none] (7) at (-2.5, -0) {};
	\end{pgfonlayer}
	\begin{pgfonlayer}{edgelayer}
		\draw [bend left, looseness=1.00] (0.center) to (1);
		\draw [in=0, out=-120, looseness=1.00] (1) to (2.center);
		\draw [bend right, looseness=1.00] (3.center) to (4);
		\draw [in=180, out=-60, looseness=1.00] (4) to (5.center);
		\draw (7.center) to (4);
		\draw (1) to (6.center);
	\end{pgfonlayer}
\end{tikzpicture}}
\tikzset{x=1em, y=1.5ex}
 \\

\tikzset{x=1em, y=2.1ex}
\begin{tikzpicture}
	\begin{pgfonlayer}{nodelayer}
		\node [style=black-faded] (0) at (0.75, -0) {};
		\node [style=none] (1) at (2, -0) {};
		\node [style=none] (2) at (-0.5, -0.75) {};
		\node [style=black-faded] (3) at (-0.5, 0.75) {};
		\node [style=none] (4) at (-1.75, 1.5) {};
		\node [style=none] (5) at (-1.75, -0) {};
		\node [style=none] (6) at (-1.75, -0.75) {};
	\end{pgfonlayer}
	\begin{pgfonlayer}{edgelayer}
		\draw (0) to (1.center);
		\draw [bend right, looseness=1.00] (2.center) to (0);
		\draw [bend left, looseness=1.00] (4.center) to (3);
		\draw [bend right, looseness=1.00] (5.center) to (3);
		\draw [in=0, out=120, looseness=1.00] (0) to (3);
		\draw (6.center) to (2.center);
	\end{pgfonlayer}
\end{tikzpicture}}
\tikzset{x=1em, y=1.5ex}
\;\stcong\: 
\tikzset{x=1em, y=2.1ex}
\begin{tikzpicture}
	\begin{pgfonlayer}{nodelayer}
		\node [style=none] (0) at (-0.25, 1) {};
		\node [style=none] (1) at (-1.5, -1.25) {};
		\node [style=black-faded] (2) at (-0.25, -0.5) {};
		\node [style=none] (3) at (-1.5, 1) {};
		\node [style=black-faded] (4) at (1, 0.25) {};
		\node [style=none] (5) at (2.25, 0.25) {};
		\node [style=none] (6) at (-1.5, 0.25) {};
	\end{pgfonlayer}
	\begin{pgfonlayer}{edgelayer}
		\draw (4) to (5.center);
		\draw [bend left, looseness=1.00] (0.center) to (4);
		\draw [bend right, looseness=1.00] (1.center) to (2);
		\draw [bend left, looseness=1.00] (6.center) to (2);
		\draw [in=0, out=-120, looseness=1.00] (4) to (2);
		\draw (3.center) to (0.center);
	\end{pgfonlayer}
\end{tikzpicture}}
\tikzset{x=1em, y=1.5ex}
\qquad  
\tikzset{x=1em, y=2.1ex}
\begin{tikzpicture}
	\begin{pgfonlayer}{nodelayer}
		\node [style=black-faded] (0) at (1.5, -0) {};
		\node [style=none] (1) at (0.75, -0.75) {};
		\node [style=none] (2) at (0.75, 0.75) {};
		\node [style=none] (3) at (2.75, -0) {};
		\node [style=none] (4) at (-1.25, -0.75) {};
		\node [style=none] (5) at (-1.25, 0.75) {};
	\end{pgfonlayer}
	\begin{pgfonlayer}{edgelayer}
		\draw (0) to (3.center);
		\draw [in=105, out=0, looseness=1.00] (2.center) to (0);
		\draw [in=0, out=-105, looseness=1.00] (0) to (1.center);
		\draw [in=180, out=0, looseness=1.00] (4.center) to (2.center);
		\draw [in=0, out=180, looseness=0.75] (1.center) to (5.center);
	\end{pgfonlayer}
\end{tikzpicture}}
\tikzset{x=1em, y=1.5ex}
\stcong\; 
\tikzset{x=1em, y=2.1ex}
\begin{tikzpicture}
	\begin{pgfonlayer}{nodelayer}
		\node [style=black-faded] (0) at (0.25, -0) {};
		\node [style=none] (1) at (1.5, -0) {};
		\node [style=none] (2) at (-1.25, -0.75) {};
		\node [style=none] (3) at (-1.25, 0.75) {};
	\end{pgfonlayer}
	\begin{pgfonlayer}{edgelayer}
		\draw (0) to (1.center);
		\draw [bend right, looseness=1.00] (2.center) to (0);
		\draw [bend left, looseness=1.00] (3.center) to (0);
	\end{pgfonlayer}
\end{tikzpicture}}
\tikzset{x=1em, y=1.5ex}
\qquad 
\tikzset{x=1em, y=2.1ex}
\begin{tikzpicture}
	\begin{pgfonlayer}{nodelayer}
		\node [style=none] (0) at (0, -1.5) {};
		\node [style=none] (1) at (-2.25, -0.75) {};
		\node [style=none] (2) at (2.25, -1.5) {};
		\node [style=none] (3) at (2.25, 0.75) {};
		\node [style=none] (4) at (0, 1.5) {};
		\node [style=none] (5) at (-2.25, 1.5) {};
		\node [style=none] (6) at (0, -0) {};
		\node [style=black-faded] (7) at (1.25, 0.75) {};
		\node [style=none] (8) at (0, -0) {};
		\node [style=black-faded] (9) at (-1.25, -0.75) {};
	\end{pgfonlayer}
	\begin{pgfonlayer}{edgelayer}
		\draw [bend left, looseness=1.00] (4.center) to (7);
		\draw [in=0, out=-120, looseness=1.00] (7) to (8.center);
		\draw [bend right, looseness=1.00] (6.center) to (9);
		\draw [in=180, out=-60, looseness=1.00] (9) to (0.center);
		\draw (1.center) to (9);
		\draw (0.center) to (2.center);
		\draw (7) to (3.center);
		\draw (4.center) to (5.center);
	\end{pgfonlayer}
\end{tikzpicture}}
\tikzset{x=1em, y=1.5ex}
\;\stcong\; 
\tikzset{x=1em, y=2.1ex}
\begin{tikzpicture}
	\begin{pgfonlayer}{nodelayer}
		\node [style=none] (0) at (1.75, 0.75) {};
		\node [style=black-faded] (1) at (0.5, -0) {};
		\node [style=none] (2) at (1.75, -0.75) {};
		\node [style=none] (3) at (-2, 0.75) {};
		\node [style=black-faded] (4) at (-0.75, -0) {};
		\node [style=none] (5) at (-2, -0.75) {};
	\end{pgfonlayer}
	\begin{pgfonlayer}{edgelayer}
		\draw [bend right, looseness=1.00] (0.center) to (1);
		\draw [in=180, out=-60, looseness=1.00] (1) to (2.center);
		\draw [bend left, looseness=1.00] (3.center) to (4);
		\draw [in=0, out=-120, looseness=1.00] (4) to (5.center);
		\draw (1) to (4);
	\end{pgfonlayer}
\end{tikzpicture}}
\tikzset{x=1em, y=1.5ex}
\;\stcong\; 
\tikzset{x=1em, y=2.1ex}
\begin{tikzpicture}
	\begin{pgfonlayer}{nodelayer}
		\node [style=none] (0) at (0, -1.5) {};
		\node [style=none] (1) at (2.25, -0.75) {};
		\node [style=none] (2) at (-2.25, -1.5) {};
		\node [style=none] (3) at (-2.25, 0.75) {};
		\node [style=none] (4) at (0, 1.5) {};
		\node [style=none] (5) at (2.25, 1.5) {};
		\node [style=none] (6) at (0, -0) {};
		\node [style=black-faded] (7) at (-1.25, 0.75) {};
		\node [style=none] (8) at (0, -0) {};
		\node [style=black-faded] (9) at (1.25, -0.75) {};
	\end{pgfonlayer}
	\begin{pgfonlayer}{edgelayer}
		\draw [bend right, looseness=1.00] (4.center) to (7);
		\draw [in=180, out=-60, looseness=1.00] (7) to (8.center);
		\draw [bend left, looseness=1.00] (6.center) to (9);
		\draw [in=0, out=-120, looseness=1.00] (9) to (0.center);
		\draw (1.center) to (9);
		\draw (0.center) to (2.center);
		\draw (7) to (3.center);
		\draw (4.center) to (5.center);
	\end{pgfonlayer}
\end{tikzpicture}}
\tikzset{x=1em, y=1.5ex}
  
\end{align*}
\caption{Axioms of special Frobenius bimonoids}\label{fig:Frob}
\end{figure}
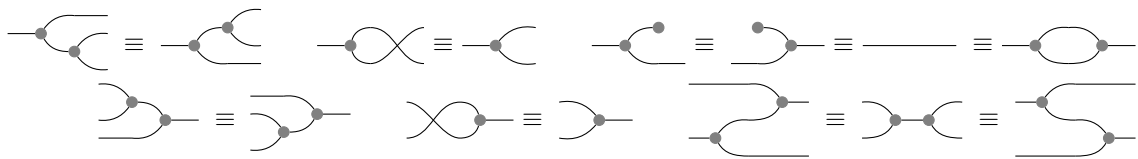

\section{Background: Bialgebras and GSOS Specifications}\label{sec:background}
We recall the fundamentals of (co)algebras and monads in Appendix~\ref{app:preliminaries}.
We refer to Appendix~\ref{app:exampleautomata} for a simple example showcasing the notions recalled below.

\smallskip
\noindent \textbf{Distributive laws and bialgebras.} 
%
A \emph{distributive law} of a monad $(\mathcal{T}, \eta, \mu)$ over an endofunctor $\funF$ is a natural transformation $\lambda\colon \mathcal{T}\funF \To \funF\mathcal{T}$ s.t.\ $\lambda \circ \eta_\funF = \funF \eta$ and $\lambda \circ \mu_\funF = \funF\mu \circ \lambda_\mathcal{T} \circ \mathcal{T}\lambda$. A \emph{$\lambda$-bialgebra} is a triple $(X,\alpha,\beta)$ s.t.\  $(X,\alpha)$ is an Eilenberg-Moore algebra for $\mathcal{T}$, $(X,\beta)$ is a $\funF$-coalgebra and $\funF\alpha\circ \lambda_X \circ \mathcal{T}\beta = \beta \circ \alpha$. Bialgebra morphisms are both $\mathcal{T}$-algebra and $\funF$-coalgebra morphisms. 

Given a coalgebra $\beta\colon X \to \funF\mathcal{T}X$ for a monad $(\mathcal{T}, \eta, \mu)$ and a functor $\funF$, if there exists a distributive law $\lambda\colon \mathcal{T}\funF \To \funF\mathcal{T}$, one can form a coalgebra $\beta^{\sharp}\colon \mathcal{T}X \to \funF\mathcal{T}X $ defined as $\mathcal{T}X \tr{\mathcal{T}\beta}\mathcal{T}\funF\mathcal{T}X \tr{\lambda_{\mathcal{T}X}} \funF \mathcal{T}\mathcal{T}X \tr{\funF\mu}\funF\mathcal{T}X$. Most importantly, $(\mathcal{T}X,\mu, \beta^{\sharp})$ is a $\lambda$-bialgebra. 

\noindent \textbf{Free monads.} We recall the construction of the monad $\free{\funF} \colon \catC \to \catC$ \emph{freely generated} by a functor $\funF\colon \catC \to \catC$. Assume that $\catC$ has coproducts and that, for all objects $X$ of $\catC$, there exists an initial $X+\funF$-algebra that we denote as
$X + \funF(\free{\funF} X) \tr{[\eta_X,\kappa_X]} \free{\funF}X$. It is easy to check that the assignment $X \mapsto \free{\funF}X$ induces a functor $\free{\funF} \colon \catC \to \catC$.
The map $\eta_X\colon X \to \free{\funF}X$ gives rise to the unit of the monad; the multiplication $\mu_X \colon \free{\funF}\free{\funF} X \to \free{\funF} X$ is the unique algebra morphism from the initial $ \free{\funF}X+\funF$-algebra to the algebra $\free{\funF}X+\funF(\free{\funF}X) \tr{[id,\kappa_X]} \free{\funF}X$. 

\smallskip
\noindent \textbf{GSOS specifications.} An \emph{abstract GSOS specification} is a natural transformation $\lambda\colon \mathcal{S} \funF \To \funF\mathcal{S}^\dagger$, where $\funF$ is a functor representing the coalgebraic behaviour, $\mathcal{S}$ is a functor representing the syntax. It is important to recall the following fact. 
%

\begin{proposition}[\!\!\cite{LPW-distrlawsmoncom}] \label{prop:GSOSToDistrLaw} Any GSOS spec.\ $\lambda\colon \mathcal{S} \funF \To \funF\mathcal{S}^\dagger$ yields a distrib.\ law $\lambda^{\dagger}\colon \mathcal{S}^{\dagger} \funF \To \funF\mathcal{S}^\dagger$.
\end{proposition}

%

\noindent\textbf{Coproduct of GSOS specifications.} Suppose we have two functors $\mathcal{S}_1,\mathcal{S}_2 \colon \catC \to \catC$ capturing two syntaxes, a functor $\funF\colon \catC \to \catC$ for the coalgebraic behaviour, and two GSOS specifications $\lambda_1\colon \mathcal{S}_1 \funF \To \funF\mathcal{S}_1^\dagger$ and $\lambda_2\colon \mathcal{S}_2 \funF \To \funF\mathcal{S}_2^\dagger$. Then we can construct a GSOS specification 
$\lambda_1\cdot \lambda_2 \colon (\mathcal{S}_1+\mathcal{S}_2) \funF \To  \funF  (\mathcal{S}_1+\mathcal{S}_2)^{\dagger}$. The details are in Appendix~\S~\ref{app:preliminaries}.

\smallskip

\noindent \textbf{Quotients of monads and distributive laws.} Given the correspondence between finitary monads and algebraic theories~\cite{hyland2007category}, it natural to consider \emph{quotients} of monads by additional equations. 
Following~\cite{DBLP:conf/concur/BuscemiM02,BHKR13,RotPhDThesis}, for a monad $\funT$ on a category $\catC$, \emph{$\funT$-equations} can be defined as a tuple $\Eq = (\ArEq, l , r)$ consisting of a functor $\ArEq \colon \catC \to \catC$ and natural transformations $l, r \colon \ArEq \To \funT$. The intuition is that $\ArEq$ acts as the variables of each equation, whose left- and right-hand sides are $l$ and $r$, respectively. Assuming mild conditions that generalise the properties of $\SET$ (see \cite[Ass. 7.1.2]{RotPhDThesis}), one constructs the \emph{quotient} of $\funT$ by $\funT$-equations.  The conditions hold in our setting: categories of presheaves over a discrete index category.

\begin{proposition}[\emph{cf.} \cite{RotPhDThesis}] \label{prop:monadquotient}
If $\catC = \SET^{\catD}$ for discrete $\catD$, $\funT$-equations $\Eq$ yield a monad $\quot{\funT}{\Eq} \colon \catC \to \catC$ 
with algebras precisely $\funT$-algebras $\funT A \tr{\alpha} A$ that satisfy $\Eq$, in the sense that $\alpha \circ l_A = \alpha \circ r_A$. Moreover, there exists a monad morphism $q^\Eq \colon \funT \to \quot{\funT}{\Eq}$ with epi components.
\end{proposition}

One may also quotient distributive laws, provided these are compatible with the new equations. Fix an endofunctor $\funF$ and a monad $\funT$ on $\SET^{\catD}$, together with $\funT$-equations $\Eq = (\ArEq, l, r)$. We say that a distributive law $\lambda \colon \funT \funF \Rightarrow \funF \funT$ \emph{preserves equations $\Eq$} if, for all $A \in  \catC$, the following diagram commutes: 
$\xymatrix@L=5pt{\ArEq \funF A \ar@<-.4ex>[r]^-{l_{\funF A}} \ar@<.4ex>[r]_-{r_{\funF A}} & \funT \funF A \ar[r]^-{\lambda_A} & \funF \funT A \ar[r]^-{\funF q^\Eq_{A}} & \funF \quot{\funT}{\Eq} A}$.

\begin{proposition}[\emph{cf.} \cite{RotPhDThesis}]  \label{prop:distrlawquotient} If $\lambda \colon \funT \funF \to \funF \funT$ preserves equations $\Eq$ then there exists a (unique) distributive law $\lambda_{/\Eq} \colon \quot{\funT}{\Eq} \funF \Rightarrow \funF \quot{\funT}{\Eq}$ such that $\lambda_{/\Eq} \circ q^\Eq \funF = \funF q^\Eq \circ \lambda$.
\end{proposition}

\section{Diagrammatic Syntax as Monads}\label{sec:hierarchy}

\subsection{The Category of Signatures}
Syntax and semantics of string diagrams will be specified in the category $\SIG \df \Span{\SET}(\N,\N)$, where objects are spans $\N \tl{} \Sigma \tr{} \N$ in $\SET$ and arrows are span morphisms: given objects $\N \tl{s} X \tr{t} \N$ and $\N \tl{s'} \Sigma' \tr{t'} \N$, an arrow is a function $f \colon \Sigma \to \Sigma'$ such that $t' \circ f = t$ and $s' \circ f = s$. We think of an object of $\SIG$ as a \emph{signature}, i.e.\ a set of symbols $\Sigma$ equipped with arity and coarity functions $a, c\colon \Sigma \to \N$. We write $\Sigma(n,m)$ for the set $\{d\in \Sigma \;|\; \langle a,c \rangle (d)=(n,m)\}$ of operations with arity $n$ and coarity $m$. Note that we allow coarities different from $1$: this is because string diagrams express \emph{monoidal} algebraic theories, not merely \emph{cartesian} ones.

Since the objects in $\SIG$ are spans with identical domain and codomain, we will often write $\Sigma$ for the entire span $\N \tl{a} \Sigma \tr{c} \N$. In particular, $\N$ means the identity span $\N \tl{id} \N \tr{id} \N$.

\begin{example}\label{ex:signature}
Recall the language $\Circ{\Rig}$ from \S~\ref{sec:motivating}. Line \eqref{eq:syntax1} of its syntax together with the first two lines of the sorting discipline in Fig.~\ref{sort} define a signature $\Sigma$:  every axiom $\typeJudgment{}{d}{\sort{n}{m}}$ gives the symbol $d$ arity $n$ and coarity $m$. For instance, $\Sigma(1,2)= \{\Bcomult,\, \Wcomult\}$.
\end{example}

For computing (co)limits, it is useful to observe that $\SIG$ is isomorphic to the presheaf category $\SET^{\N \times \N}$, where $\N \times \N$ is the discrete category with objects pairs $(n,m)\in \N \times \N$. 


\subsection{Functors on Signatures}\label{sec:functors}

We turn to (co)algebras of endofunctors  $\funF\colon \SIG \to \SIG$
generated by the following grammar:
 $$\funF \quad ::= \quad Id \ \mid \  \Sigma \ \mid \  \N \ \mid \    \funF\poi \funF \ \mid \  \funF \tns \funF \ \mid  \  \funF+ \funF \ \mid \   \funF\times \funF \ \mid \  \overline{\funG}$$
where $\funG$ ranges over functors $\funG \: \colon \SET \to \SET$ and $\Sigma$ is a span $\N \tl{}\Sigma \tr{} \N$. 
In more detail:
\begin{itemize}
\item $Id\colon \SIG \to \SIG$ is the identity functor. 
\item 
$\Sigma\colon \SIG \to \SIG$ is the constant functor mapping every object to $\N \tl{}\Sigma \tr{} \N$ and every arrow to $id_{\Sigma}$; an important special case is $\N \colon \SIG \to \SIG$ the constant functor to $\N \tl{id} \N \tr{id} \N$.
\item $(\cdot) \poi (\cdot) \colon \SIG^2 \to \SIG$ is \emph{sequential composition} for signatures. 
On objects, 
$\Sigma_1 \poi \Sigma_2$
is  $$\N \tl{s_1 \circ \pi_1}\{(d_1,d_2)\in \Sigma_1\times \Sigma_2 \;|\; t_1(d_1)=s_2(d_2)\} \tr{t_2 \circ \pi_2} \N \text{.}$$
Since the above is a $\SET$-pullback, the action on arrows is inducted by the universal property. 
Note that, up to iso, $(\cdot) \poi (\cdot) \colon \SIG^2 \to \SIG$ is associative with unit $\N \colon \SIG \to \SIG$. 
\item $(\cdot) \tns (\cdot) \colon \SIG^2 \to \SIG$ is \emph{parallel composition} for signatures, with
$\Sigma_1 \tns \Sigma_2$ given by:
$$ \N \tl{+\circ (s_1\times s_2)} \Sigma_1\times \Sigma_2 \tr{+ \circ(t_1 \times t_2)} \N $$
where $+\colon \N \times \N \to \N$ is usual $\N$-addition. Again $(\cdot) \tns (\cdot) \colon \SIG^2 \to \SIG$ associates up to iso.

\item For the remaining functors, we use the fact that $\SIG\cong\SET^{\N \times \N}$, which guarantees (co)completeness, with limits and colimits constructed pointwise in $\SET$. Thus, 
for spans 
$\Sigma_1$ and $\Sigma_2$, 
their coproduct is $\N \tl{[s_1,s_2]} \Sigma_1+\Sigma_2 \tr{[t_1,t_2]} \N$ 
and the carrier of the product is  $\{(d_1,d_2) \;|\; s_1(d_1)=s_2(d_2) \text{ and } t_1(d_1)=t_2(d_2) \}$, with the two obvious morphisms to $\N$. 
\item The isomorphism $\SIG \cong \SET^{\N \times \N}$ also yields the extension of an arbitrary endofunctor $\funG\colon \SET \to \SET$ to a functor $\bar{\funG} \colon\SIG \to \SIG$ defined by post-composition with $\funG$, that is $\bar{\funG}(\Sigma)=\funG\circ \Sigma$ for all $\Sigma\colon \N \times \N \to \SET$. In particular, we shall often use the functor $\overline{\Pow}$ obtained by post-composition with the $\kappa$-bounded powerset functor $\Pow\colon \SET \to \SET$.\footnote{Boundedness is needed to ensure the existence of a final coalgebra, see \S~\ref{sec:comp-prop}. In our leading example $\Circ{\Rig}$, $\kappa$ can be taken to be the cardinality of the semiring $\Rig$.}
\end{itemize}

Next we use these endofunctors to construct monads that capture the two-dimensional algebraic structure of string diagrams. In \S~\ref{sec:prop-monad} we construct the monad encoding the symmetric monoidal structure of props and in \S~\ref{sec:cw-monad} we construct the monad for the Frobenius structure of Carboni-Walters props. Later, in \S~\ref{sec:SOS}, we shall use these monads to define compositional bialgebraic semantics for string diagrams of each of these categorical structures.



\subsection{The Prop Monad}
\label{sec:prop-monad}

Here we
define a monad on $\SIG$ with algebras precisely props: 
symmetric strict monoidal categories with objects the natural numbers, where the monoidal product on objects is addition. Together with identity-on-objects symmetric monoidal functors they form a category $\PROP$. 
The first step is to encapsulate the operations of props as a $\SIG$-endofunctor. 
\begin{equation}\label{eq:signatureProps}
\SCOMPSYM  \df  (Id\poi Id) + \IdFunc + (Id\tns Id)  +\epsilon +  \sigma \colon \SIG \to \SIG.
\end{equation}
In the type of $\SCOMPSYM$, $Id\poi Id \colon \SIG \to \SIG$ is sequential composition and $\IdFunc$ the identity arrow on object $1$, i.e.\ the constant functor to $\N \tl{h} \{ id_1\} \tr{h} \N$, with $h \colon \id_1 \mapsto 1$. Similarly, $Id\tns Id$ is the monoidal product with unit $\epsilon$, i.e.\ the constant functor to $\N \tl{q} \{ 0 \} \tr{q} \N$, with $q \colon 0 \mapsto 0$.  Finally, $\sigma$ is the basic symmetry: the constant functor to $\N \tl{f} \{ \sigma_{1,1}\} \tr{f} \N$, with $f \colon  \sigma_{1,1} \mapsto 2$. 

The free monad $\SCOMPSYM^{\dagger}$ on $\SCOMPSYM$ is the functor mapping a span $\Sigma$ to the span of $\Sigma$-terms obtained by sequential and parallel composition, together with symmetries and identities ---with the identity $id_n$ defined by parallel composition of $n$ copies of $id_1$. 

Algebras for this monad are spans $\Sigma$ together with span morphisms $identity\colon \iota \to \Sigma$, $composition \colon \Sigma \poi \Sigma \to \Sigma$, $parallel \colon \Sigma \tns \Sigma \to \Sigma$, $unit\colon \epsilon\to \Sigma$, and $swap\colon \sigma\to \Sigma$. This information \emph{almost} defines a 
prop $\catC_\Sigma$: the carrier $\Sigma$ of the span is the set of arrows of $\catC_\Sigma$, containing special arrows $id_n$ and $\sigma_{n,m}$ for identities and symmetries, $compose$ assigns to every pairs of composable arrows their composition, and $\tns$ assigns to every pair of arrows their monoidal product. The missing data  is the usual equations~\eqref{eq:equationCategories}-\eqref{eq:equationPROPs3} of symmetric monoidal categories. Thus, in order to obtain props as algebras, we quotient the monad $\SCOMPSYM^{\dagger}$ by those equation, expressed abstractly as a triple $\EqSM = (\ArEqSM, l, r)$, as described in \S~\ref{sec:background}. The functor $\ArEqSM \colon \SIG \to \SIG$, defined below, 
has summands following the order \eqref{eq:equationPROPs}-\eqref{eq:equationPROPs3}:
\begin{equation}\label{eq:typeEquations}
(Id \tns Id \tns Id) + Id + Id + \sigma + ((Id \poi\!\! Id) \tns ( Id \poi\!\! Id)) + (Id \poi \!\! Id \poi\!\! Id) + Id + Id + Id^{+1} + Id^{+1} \end{equation}
Here, $Id^{+1}$ is the functor adding $1$ to the arity/coarity of each element of a given span $\N \tl{a} \Sigma \tr{c} \N$. We also need natural transformations $l,r \colon \ArEqSM \to \SCOMPSYM^{\dagger}$ that define the left- and right-hand side of each equation. For instance, for fixed $\Sigma \in \SIG$ and $(n,m) \in \N \times \N$:
\begin{itemize}
\item an element of $\Sigma \poi \Sigma \poi \Sigma$ (sixth summand of \eqref{eq:typeEquations}) is a tuple $(f,g,h)$ of $\Sigma$-elements, where $f$ is of type $(n,w)$, $g$ of type $(w,v)$, and $h$ of type $(v,m)$, for arbitrary $w,v \in \N$. We let $l_{\Sigma}$ map $(f,g,h)$ to the term $(f \poi g) \poi h$ of type $(n,m)$ in $\SCOMPSYM^{\dagger}(\Sigma)$, and $r_{\Sigma}$ to the term $f \poi (g \poi h)$. Thus this component gives the second equation in~\eqref{eq:equationPROPs2} (associativity).
\item the seventh summand $Id$ in \eqref{eq:typeEquations} yields a $\Sigma$-term $f$, which $l_{\Sigma} \colon \Sigma \to \SCOMPSYM^{\dagger}(\Sigma)$ maps to $f \poi \id_m$ and $r_{\sigma}\colon \Sigma \to \SCOMPSYM^{\dagger}(\Sigma)$ maps to $f$, thus yielding the final equation in~\eqref{eq:equationPROPs2}.
\item an element in $\Sigma^{+1}$ (last summand of \eqref{eq:typeEquations}) of type $(n+1,m+1)$ is a $\Sigma$-term $g$ of type $(n,m)$, which is mapped by $l_{\Sigma}$ to $(\sigma_{n,1} \poi ( id_1 \tns g))$ and by $r_{\sigma}$ to $(g \tns id_1) \poi \sigma_{m,1}$, both elements of $\SCOMPSYM^{\dagger}(\Sigma)$ of type $(n+1,m+1)$, thus giving the final equation in~\eqref{eq:equationPROPs3}.
\end{itemize}
The remainder of the definition of $l,r \colon \ArEqSM \to \SCOMPSYM^{\dagger}$, handles the remaining equations in \eqref{eq:equationPROPs}-\eqref{eq:equationPROPs3}, and should be clear from the above.
Now, using Proposition~\ref{prop:monadquotient}, we quotient the monad $\SCOMPSYM^{\dagger}$ by $(\ArEqSM, l, r)$, obtaining a monad that we call $\SPROP$. We can then conclude by construction that the Eilenberg-Moore category $\EM{\SPROP}$ for the monad $\SPROP$ (with objects the $\SPROP$-algebras, and arrows the $\SPROP$-algebra homomorphisms) is precisely $\PROP$.
\begin{proposition}\label{prop:algebrasprops}
$\EM{\SPROP}\cong \PROP$.
\end{proposition}

\begin{example}\label{ex:freePROP}
The monad $\SPROP$ takes $\Sigma$ to the prop 
freely generated by $\Sigma$. Taking $\Sigma$ as in Example \ref{ex:signature}, one obtains $\SPROP(\Sigma)$
with arrows $n \to m$ string diagrams of $\Circ{\Rig}$ of sort $\sort{n}{m}$.
\end{example}

\subsection{The Carboni-Walters Monad}
\label{sec:cw-monad}

The treatment we gave to props may be applied to other categorical structures. For space reasons, we 
only consider one additional such structure: \emph{Carboni-Walters} (CW) props, also called `hypergraph categories' \cite{Fong-hypergraph}. 
Here each object $n$ carries a distinguished special Frobenius bimonoid compatible with the monoidal product: it can be defined recursively using parallel compositions of the Frobenius structure on the generating object $1$. 
\begin{definition}\label{def:cw-prop}
  A \emph{CW prop} is a prop with morphisms $\Gcomult\from 1\to 2$, $\Gcounit\from 1\to 0$, $\Gmult\from 2\to 1$, $\Gunit\from 0\to 1$ satisfying the equations of special Frobenius bimonoids (Fig.~\ref{fig:Frob}).
\end{definition}

CW props with prop morphisms preserving the Frobenius bimonoid form a subcategory $\CWP$ of $\PROP$.
We can now extend the prop monad of \S~\ref{sec:prop-monad} to obtain a monad with algebras CW props.
The signature is that of a prop with the additional Frobenius structure. Let $\Gcomult\from \SIG\to\SIG$ be the functor constant at $\N \xleftarrow{s} \{\Gcomult\}\xrightarrow{t} \N$ with $s(\Gcomult) = 1$ and $t(\Gcomult) = 2$. Similarly, we introduce the constant functors $\Gcounit\from \SIG\to\SIG$, $\Gmult\from \SIG\to\SIG$ and $\Gunit\from \SIG\to\SIG$ for the other generators. Let $\SCOMPFR \df \SPROP + \Gcomult+\Gcounit+\Gmult+\Gunit$.

We now need to quotient $\SCOMPFR$ by the defining equations of special Frobenius bimonoids (Fig.~\ref{fig:Frob}). We omit the detailed encoding of these equations as a triple $\EqCW = (\ArEqCW , \lCW, \rCW)$ since it presents no conceptual difficulty. Let $\SCW$ be the quotient of $\SCOMPFR$ by these equations. As for props, we obtain $\EM{\SCW}\cong \CWP$ by construction.

\section{Bialgebraic Semantics for String Diagrams}\label{sec:SOS}

Now that we have established monads for our categorical structures of interest, we study coalgebras that capture behaviour for string diagrams in these categories, and distributive laws that yield the desired bialgebraic semantics. We  fix our `behaviour' functor to $$\Fbeh \df \overline{\Pow}(L\poi Id \poi L)\colon \SIG \to \SIG$$
where $L \colon \SIG \to \SIG$ is the \emph{label functor} constant at the span $\N \tl{|\cdot |}A^*\tr{|\cdot |} \N$, with $A^*$ the set of words on some set of labels $A$. The map $|\cdot | \colon A^* \to \N$ takes $w\in A^*$ to its length $|w|\in \N$. An $\Fbeh$-coalgebra is a span morphism $\Sigma \to \overline{\Pow}(L\poi \Sigma \poi L)$; a function that takes $f\in \Sigma({n,m})$ to a set of transitions $(v,g,w)$ with the appropriate sorts, i.e.\ $g\in \Sigma(n,m)$, $|v|=n$ and $|w|=m$.

The data of an $\Fbeh$-coalgebra $\beta \colon \Sigma \to \overline{\Pow}(L\poi \Sigma \poi L)$ is that of a transition relation. For instance, fix labels $A = \{a,b\}$ and let $x,y\in\Sigma(1,2)$ and $z\in\Sigma(1,1)$; suppose also that $\beta$ maps $x$ to $\{(b \poi y \poi ab), (a \poi x \poi aa) \}$, $y$ to $\emptyset$ and $z$ to $\{(b \poi z \poi a)\}$. Then $\beta$ can be written:
\begin{equation}\label{eq:exFBehcoalgebra}
x \dtrans{\,b\,}{a \labelSep b} y\quad\quad
x \dtrans{\,a\,}{a \labelSep a} x \quad\quad z \dtrans{\,b\,}{\,a\,} z
\end{equation}

\begin{example}\label{ex:coalgsig}
In our main example, Fig.~\ref{sos1} defines a coalgebra $\beta \colon \Sigma \to \overline{\Pow}(L\poi \Sigma \poi L)$ where $\Sigma$ is the signature from Example~\ref{ex:signature} and the set of labels is $\Rig$. For instance $\beta(\Bcounit)=\{(k,\Bcounit, \varepsilon) \mid k\in \Rig\}$. Note the $\kappa$ bounding $\Pow$ is the cardinality of $\Rig$.
\end{example}

In the sequel we shall construct distributive laws between the above behaviour functor and monads encoding the various categorical structures defined in the previous section.

\subsection{Bialgebraic Semantics for Props}\label{sec:comp-prop}
The modularity of $\SPROP$ can be exploited to define a distributive law of the $\SPROP$ over $\Fbeh$. Recall from~\S~\ref{sec:prop-monad} that $\SPROP$ is a quotient of $\SCOMPSYM^{\dagger}$. We start by letting $\Fbeh = \overline{\Pow}(L\poi Id \poi L)$ interact with the individual summands of $\SCOMPSYM$ (see \eqref{eq:signatureProps}), corresponding to the operations of props. This amounts to defining GSOS specifications: 
\begin{align*}
&\lambdaseqcomp \colon  \overline{\Pow}(L\poi Id \poi L) \poi  \overline{\Pow}(L\poi Id \poi L)  \To  \overline{\Pow}(L\poi (Id\poi Id)^{\dagger} \poi L) 
& \text{(sequential composition)} \\ 
&\lambdaid \colon \iota \To \overline{\Pow}(L\poi \iota^{\dagger} \poi L) 
& \text{(identity)} \\
&\lambdaparcomp \colon  \overline{\Pow}(L\poi Id \poi L)  \tns  \overline{\Pow}(L\poi Id \poi L)\To  \overline{\Pow}(L\poi (Id \tns Id)^{\dagger} \poi L) 
&\text{(monoidal product)}  \\
&\lambdazero\colon  \epsilon \To  \overline{\Pow}(L\poi \epsilon^{\dagger} \poi L) 
&\text{(product unit)} \\
&\lambdasym \colon  \sigma \To  \overline{\Pow}(L\poi \sigma^{\dagger} \poi L)
&\text{(symmetry)}
\end{align*}
Definitions of these maps are succinctly given via derivation rules, see Fig.~\ref{sos2}.

We explain this in detail for $\lambdaseqcomp$, the others are similar. Given $\Sigma \in \SIG$, an element of type $(n,m)$  in the domain $ \overline{\Pow}(L\poi \Sigma \poi L) \poi  \overline{\Pow}(L\poi \Sigma \poi L)$ 
is a pair $(A,B)$, where, for some $z \in \N$,
$$A \text{ is a set of triples }(\bm{a},c',\bm{b}) \in L(n,n) \times \Sigma(n,z) \times L(z,z),\text{ and}$$
$$B \text{ is a set of triples }(\bm{b},d',\bm{c}) \in L(z,z) \times \Sigma(z,m) \times L(m,m).$$
 Then $\lambdaseqcomp_\Sigma(A,B) \df \{(\bm{a}, c'\poi d', \bm{c}) \;|\; (\bm{a},c',\bm{b})\in A, \; (\bm{b},d',\bm{c})\in B \}$. Following the convention~\eqref{eq:exFBehcoalgebra}, we can write this data as: $\boxed{\dtrans {\bm{a}}{\bm{c}} c' \poi d' } \in  \lambdaseqcomp_\Sigma(A,B)$ if $\boxed{\dtrans{\bm{a}}{\bm{b}} c' } \in A$ and $\boxed{\dtrans{\bm{b}}{\bm{c}} d'}\in B$.
This leads us to the more compact version of $\lambdaseqcomp$ as the transition rule in Fig.\ref{sos2}.
Next, take the coproduct of GSOS specifications $\lambdaseqcomp$, $\lambdaid$, $\lambdaparcomp$, $\lambdazero$ and $\lambdasym$ (see Appendix~\ref{app:preliminaries}) to obtain ${\lambda} \colon \SCOMPSYM \Fbeh \To \Fbeh \SCOMPSYM^{\dagger}$. 
By Proposition~\ref{prop:GSOSToDistrLaw}, this yields distributive law 
$\lambda^{\dagger} \colon \SCOMPSYM^{\dagger} \Fbeh \To \Fbeh \SCOMPSYM^{\dagger}$.

The last step is to upgrade $\lambda^{\dagger}$ to a distributive law $\lambdaquot{SMC}$ over the quotient $\SPROP$ of $\SCOMPSYM^{\dagger}$ by the equations~\eqref{eq:equationPROPs}-\eqref{eq:equationPROPs3} of SMCs. By Proposition~\ref{prop:distrlawquotient}, this is well-defined if $\lambda^{\dagger}$ preserves $\EqSM$. 
We show compatibility with associativity of sequential composition---the other equations can be verified similarly. This amounts to checking that if $\lambda^{\dagger}$ allows the derivation for $s_1 \poi (s_2 \poi s_3)$ as below left, then there exists a derivation for $(s_1 \poi s_2) \poi s_3$ as on the right, and vice-versa.
\begin{equation*}
\derivationRule{
s_1 \dtrans{u}{v} s_1
\quad
\derivationRule{
s_2\dtrans{v}{w} s_2 \quad s_3\dtrans{w}{x} s_3
}{s_2 \poi s_3\dtrans{v}{x}s_2 \poi s_3 }{}
 }
{s_1 \poi (s_2 \poi s_3) \dtrans {u}{x} s_1 \poi (s_2 \poi s_3)}{}
\qquad \quad \qquad
\derivationRule{
\derivationRule{
s_1\dtrans{u}{v} s_1 \quad s_2\dtrans{v}{w} s_2
}{s_1 \poi s_2\dtrans{u}{w} s_1 \poi s_2 }{}
\quad s_3 \dtrans{w}{x} s_3}
{(s_1 \poi s_2) \poi s_3 \dtrans {u}{x} (s_1 \poi s_2) \poi s_3}{}
\end{equation*}
By Proposition \ref{prop:distrlawquotient}, we can therefore upgrade $\lambda^{\dagger}$ to a distributive law $\lambdaquot{SM} \colon \SPROP \Fbeh \To \Fbeh \SPROP$.

We are now ready to construct the compositional semantics as a morphism into the final coalgebra. One starts with a coalgebra $\beta \colon \Sigma \to \Fbeh (\SPROP(\Sigma))$ that describes the behaviour of $\Sigma$-operations, assigning to each a set of transitions, as in \eqref{eq:exFBehcoalgebra}. The difference with \eqref{eq:exFBehcoalgebra} is that, because $\Fbeh$ is applied to $\SPROP(\Sigma)$ instead of just $\Sigma$, the right-hand side of each transition contains not just a $\Sigma$-operation, but a \emph{string diagram}: a $\Sigma$-term modulo the laws of SMCs. 

As recalled in \S~\ref{sec:background}, using the distributive law $\lambdaquot{SM}$ we can lift $\beta \colon \Sigma \to \Fbeh (\SPROP(\Sigma))$ to a $\lambdaquot{SM}$-bialgebra, $\beta^{\sharp}\from\SPROP(\Sigma) \to \Fbeh (\SPROP(\Sigma))$. Since this is a $\Fbeh$-coalgebra, the final $\Fbeh$-coalgebra $\Omega$ (the existence of which is shown in Appendix~\ref{sec:final}) yields a semantics $\sem{\cdot}_{\beta}$ as below. 
The operational semantics of a string diagram $c$ is $\beta^{\sharp}(c)$, obtained from \emph{(i)} transitions for $\Sigma$-operations given by $\beta$ and \emph{(ii)} the derivation rules (Fig.~\ref{sos2}) of $\lambdaquot{SM}$. Instead, $\sem{c}_{\beta}$ is the observable behaviour:
intuitively, its transition systems modulo bisimilarity.


\begin{minipage}{0.4\textwidth}
\[
\xymatrix@C=40pt{
\SPROP(\Sigma) \ar@{-->}[r]^-{\sem{\cdot}_{\beta}} \ar[d]^{\beta^{\sharp}} & \Omega \ar[d]\\
\Fbeh(\SPROP(\Sigma)) \ar[r]_-{\Fbeh(\sem{\cdot}_{\beta})} & \Fbeh(\Omega)
}\]
\end{minipage}
\begin{minipage}{0.55\textwidth}
The bialgebraic semantics framework ensures that $\SPROP(\Sigma)$ and $\Omega$ are $
\SPROP$-algebras,
which by Proposition \ref{prop:algebrasprops} are props. This means that the final coalgebra $\Omega$ is a prop and that $\sem{\cdot}_{\beta}$ is a prop morphism, preserving identities, symmetries and guaranteeing\end{minipage}
 compositionality:
 $\sem{s \poi t}_{\beta} = \sem{s}_{\beta} \poi \sem{t}_{\beta}$ and $\sem{s \tns t}_{\beta} = \sem{s}_{\beta} \tns \sem{t}_{\beta}$. 


\begin{example}\label{ex:finalsemleadingex}
Coming back to our running example,
in Example~\ref{ex:coalgsig} we showed that rules in Fig.~\ref{sos1} induce a coalgebra of type $ \Sigma \to \Fbeh(\Sigma)$. Since each operation in $\Sigma$ is itself a string diagram (formally, via the unit $\eta_{\Sigma} \colon \Sigma \to \SPROP(\Sigma)$), the same rules induce a coalgebra $\beta \from \Sigma \to \Fbeh\SPROP(\Sigma)$, which has the type required for the above construction. The resulting 
 coalgebra $\beta^\sharp \colon \SPROP(\Sigma) \to \Fbeh\SPROP(\Sigma)$ assigns to each diagram of 
$\SPROP(\Sigma)$  the set of transitions specified by the combined operational semantics of Figs.~\ref{sos1} and~\ref{sos2}.
The preceding discussion implies that, when e.g.\ $\Rig=\N$, bisimilarity for the Petri nets of~\cite{BonchiHPSZ19} is a congruence.
\end{example}

\subsection{Bialgebraic Semantics for Carboni-Walters Props}
\label{sec:comp-cw}
In this section we shall see two different ways of extending the GSOS specification of \S~\ref{sec:comp-prop} for CW props (see \S~\ref{sec:cw-monad}). They correspond to the operational semantics of the black and white (co)monoids as given in Fig.~\ref{sos1}. In the next section, we will see that these two different extensions give rise to two classic forms of synchronisation: \`{a} la Hoare and \`{a} la Milner.

\noindent\textbf{Black distributive law.}
The first interprets the operations of the Frobenius structure as label synchronisation: from the black node derivations on the left of Fig.~\ref{sos1} we get GSOS specifications given by natural transformations $\Bcomult\Rightarrow \funF( \Bcomult^{\dagger})$, $\Bcounit\Rightarrow \funF( \Bcounit^{\dagger})$, $\Bmult\Rightarrow  \funF(\Bmult^{\dagger} )$, and $\Bunit\Rightarrow  \funF( \Bunit^{\dagger})$. Recall that, here, we use the diagrams to denote their associated functors $\SIG\to\SIG$. 
By taking the coproduct of these and $\lambda$, the GSOS specification for props from \S~\ref{sec:comp-prop}, we obtain a specification $\lambda_{\bullet}$ for $\SCOMPFR$. It is straightforward to verify that $\lambda_{\bullet}^{\dagger}\from \SCOMPFR^\dagger\Fbeh\To\Fbeh\SCOMPFR^\dagger$
preserves the equations of special Frobenius bimonoids (Fig.~\ref{fig:Frob}), yielding  a distributive law $\lambda_{\bullet/\scriptscriptstyle{\mathsf{CW}}}^{\dagger}\from \SCW^\dagger\Fbeh\To\Fbeh\SCW^\dagger$.
As before, with $\lambda_{\bullet/\scriptscriptstyle{\mathsf{CW}}}^{\dagger}$ we obtain a bialgebra $\beta_{\bullet}^\sharp \colon \SCW(\Sigma) \to \Fbeh\SCW(\Sigma)$ from any coalgebra $\beta \colon \Sigma \to \Fbeh\SCW(\Sigma)$.

\noindent\textbf{White distributive law.}
When the set of labels $A$ is an \emph{Abelian group}, it is possible to give a different GSOS specifications for the Frobenius structure, capturing the group operation of $A$: from the white node derivations on the right of Fig.~\ref{sos1} we get GSOS specifications $\Wcomult\Rightarrow  \funF (\Wcomult^{\dagger})$, $\Wcounit\Rightarrow  \funF( \Wcounit^{\dagger})$,  $\Wmult\Rightarrow \funF( \Wmult^{\dagger})$, and $\Wunit\Rightarrow  \funF( \Wunit^{\dagger} )$.
Using a now familiar procedure we obtain a GSOS specifications $\lambda_{\circ}$ for $\SCOMPFR$. The group structure on $A$ guarantees~\cite{pavlovic2009quantum} that $\lambda_{\circ}^{\dagger}\from \SCOMPFR^\dagger\Fbeh\To\Fbeh\SCOMPFR^\dagger$
preserves the equations of special Frobenius bimonoids (Fig.~\ref{fig:Frob}). Therefore we get a distributive law $\lambda_{\circ/\scriptscriptstyle{\mathsf{CW}}}^{\dagger}\from \SCW^\dagger\Fbeh\To\Fbeh\SCW^\dagger$.

\begin{remark}\label{rmk:comp-Lawvere}
Given the results in this section, one could ask if bialgebraic semantics works for \emph{any} categorical structure. A notable case in which it fails is that of Lawvere theories \cite{hyland2007category}. These can be seen as props with a \emph{natural} comonoid structure on each object \cite{BonchiSZ18}. One may define a monad for Lawvere theories following the same recipe as above. However, it turns out that this monad is incompatible with the GSOS specification for the comonoid given in Fig.~\ref{sos1}. To see the problem consider a term $d$ that can perform two transitions nondeterministically: $d \dtrans{\epsilon}{a} d$ and  $d \dtrans{\epsilon}{b} d$. The naturality of the comonoid forces $d\poi \Bcomult\approx d\tns d$ but $d\poi\Bcomult$ can only perform the $a\, a$ and $b\, b$ transitions while $d\tns d$ can also perform $a\, b$ or $b\, a$. Thus the specification would not be compositional. For more details, see Appendix~\ref{sec:Lawvere-fail}.
\end{remark}


\newcommand{\Names}{\mathcal{N}}
\newcommand{\Var}{\mathcal{V}}
\newcommand{\Proc}{\mathbb{P}}
\newcommand{\Alphabet}{al}
\newcommand{\perm}{\sigma}
\newcommand{\CW}{\mathsf{CW}}
\newcommand{\zerodv}{\bullet}

\newcommand{\Bdtrans}[2]{\dtrans{#2}{#1}_b}
\newcommand{\Wdtrans}[2]{\dtrans{#2}{#1}_w}

\newcommand{\Htr}[1]{\tr{#1}_H}
\newcommand{\Mtr}[1]{\tr{#1}_M}

\newcommand{\enc}[1]{\langle\! \langle #1 \rangle \! \rangle}

\section{Black and White Frobenius as Hoare and Milner Synchronisation}\label{BWFrob}
%
The role of this section is twofold: on the one hand we demonstrate how classical process calculus syntax benefits from a string diagrammatic treatment; on the other we draw attention towards a surprising observation, namely that the black and white Frobenius structures discussed previously provide the synchronisation mechanism of, respectively, CSP and CCS.

\subsection{Syntax}
We consider a minimal process calculus for simplicity. Assume a countable set $\Names$ of \emph{names}, $a_1$, $a_2$, \dots and a set $\Var$ of \emph{process variables}, $\procvar{f}$, $\procvar{g}$, \dots, equipped with a function $ar\colon \Var \to \N$ that assigns the set of names that the process may use: $ar(\procvar{f})=n$ means that the process $\procvar{f}$ uses only names $\{a_1, \dots, a_n\}$. This is Hoare's~\cite{Hoarebook} notion of \emph{alphabet} for process variables.

Roughly speaking, in a string diagram, dangling wires perform the job of variables.
To ease the translation of terms to diagrams, we include permutations of names in the syntax, hereafter denoted by $\perm$. For a permutation $\perm \colon \Names \to \Names$, its support is the set $supp(\perm)=\{a_i \mid a_i\neq \perm(a_i)\}$; $\perm$ is \emph{finitely supported} if $supp(\perm)$ is finite. For each finitely supported permutation $\perm$ its \emph{degree} is defined as the greatest $i\in \N$ such that $a_i\in supp(\perm)$.

The set of processes is defined recursively as follows
$$P:= P|P, \;\; \nu a_i(P),\;\;\procvar{f}, \;\;  P\perm$$
where $a_i\in \Names$, $\procvar{f}\in \Var$ and $\perm$ is a finitely supported permutation of names. The symbol $|$ stands for the parallel composition of processes. The symbol $\nu a_i$ stands for the restriction, or hiding, of the name $a_i$. Observe that there are no primitives for prefixes, non-deterministic choice or recursion: these will appear in the declaration of process variables which we will describe in \S~\ref{sec:pcsemantics}. The idea here is to separate the behaviour, specified in the declaration of process variables, and the communication topology of the network, given by the syntax above. The notion of alphabet can be defined for all processes as follows:
$$
\Alphabet(P|Q) \! = \! \Alphabet(P)\cup \Alphabet(Q) \quad
\Alphabet(\nu a_i(P)) \! = \! \Alphabet(P)\setminus \{a_i\} \quad
\Alphabet(\procvar{f}) \! = \! \{a_1, \dots, a_{ar(\procvar{f})}\}  \quad
\Alphabet(P\perm) \! = \! \perm[\Alphabet(P)]
$$

\smallskip\noindent\textbf{From one-dimensional to two-dimensional syntax.}
We use a typing discipline to guide the translation of terms to string diagrams:
\begin{equation}\label{eq:typingdiscipline}
	\derivationRule{n\vdash P \quad n\vdash Q}{n\vdash P|Q}{} \quad
\derivationRule{n+1\vdash P }{n\vdash \nu a_{n+1}(P)}{} \quad
\derivationRule{ar(\procvar{f})=n}{n\vdash \procvar{f}}{} \quad
\derivationRule{n\vdash P\quad degree(\perm)\leq n}{n\vdash P\perm}{} \quad
\derivationRule{n\vdash P}{n+1\vdash P}{}
\end{equation}
 The meaning of the types is explained by the following lemma, easily proven by induction.

\begin{lemma}
If $n \vdash P$ then $\Alphabet(P)\subseteq \{a_1, \dots a_n\}$.
\end{lemma}

We will translate processes to the CW prop freely generated
from $\Sigma=\{\procvar{f}\colon(n,0) \mid \procvar{f}\in \Var \text{ and } ar(\procvar{f})=n\}$; in
particular a typed process $n\vdash P$ results in a string diagram of $\SCW(\Sigma)(n,0)$. The translation $\enc{\cdot}$ is defined recursively on typed terms as follows:


\begin{gather*}
\enc{n\vdash P|Q}=\parProc{black-faded}{n}{P}{Q} \qquad
\enc{n \vdash \nu a_{n+1}(P)}=\hideProc{black-faded}{n}{P} \\
\enc{n \vdash \procvar{f}}=\diagState{n}{\procvar{f}} \qquad
\enc{n \vdash P\perm}= \permProc{n}{\sigma}{P} \qquad \enc{n+1 \vdash P}=\newvarProc{black-faded}{n}{P} 
\end{gather*}
where for $\perm$ with $degree(\perm)<n$, $\overline{\perm}\colon n \to n$ is the obvious corresponding arrow in $\SCW(\Sigma)$.

\medskip
\noindent\begin{minipage}{.63\textwidth}
\begin{example}\label{ex:pcfirst}
Let $\Var=\{\procvar{f},\procvar{g}\}$ with $ar(\procvar{f})=1$ and $ar(\procvar{g})=2$. Let $[a_2/a_1]\colon \Names \to \Names$ be the permutation swapping $a_1$ and $a_2$. One can easily check that $1\vdash \nu a_2 (\procvar{f}[a_2/a_1]\,|\, \procvar{g})$. Then $\enc{1\vdash \nu a_2 (\procvar{f}[a_2/a_1]\,|\, \procvar{g})}$ is as on the right.
\end{example}
\end{minipage}
\begin{minipage}{.37\textwidth}
\vspace{-.3cm} \	\[
  
\tikzset{x=1em, y=2.1ex}
\begin{tikzpicture}[show background rectangle, {background rectangle/.style}={fill=backcolour}]
	\begin{pgfonlayer}{nodelayer}
		\node [style=none] (0) at (-4.25, 0.75) {};
		\node [style=black-faded] (1) at (-2.5, 0.75) {};
		\node [style=none] (2) at (-1, -0.5) {};
		\node [style=none] (3) at (-1, 1.5) {};
		\node [style=black-faded] (4) at (-3.5, -0.25) {};
		\node [style=none] (5) at (-1, -1) {};
		\node [style=black-faded] (6) at (-2.5, -0.25) {};
		\node [style=none] (7) at (-1, 0.75) {};
		\node [style=none] (8) at (0.75, 1.5) {};
		\node [style=black-faded] (9) at (0.75, 0.75) {};
		\node [style=none] (10) at (1.75, 1.5) {};
		\node [style=none] (11) at (1.75, -0.5) {};
		\node [style=none] (12) at (1.75, -1) {};
		\node [style=reg] (13) at (2, 1.5) {\normalsize $\procvar{f}$};
		\node [style=reg] (14) at (2, -0.75) {\normalsize $\procvar{g}$};
	\end{pgfonlayer}
	\begin{pgfonlayer}{edgelayer}
		\draw [in=75, out=180, looseness=1.00] (3.center) to (1);
		\draw (1) to (0.center);
		\draw [in=-75, out=180, looseness=1.00] (2.center) to (1);
		\draw [in=75, out=180, looseness=1.00] (7.center) to (6);
		\draw (6) to (4);
		\draw [in=-75, out=180, looseness=1.00] (5.center) to (6);
		\draw [in=180, out=0, looseness=1.00] (3.center) to (9);
		\draw [in=180, out=0, looseness=1.25] (7.center) to (8.center);
		\draw (8.center) to (10.center);
		\draw [in=180, out=0, looseness=1.00] (2.center) to (11.center);
		\draw [in=180, out=0, looseness=1.00] (5.center) to (12.center);
	\end{pgfonlayer}
\end{tikzpicture}}
\tikzset{x=1em, y=1.5ex}
 \stcong  
\tikzset{x=1em, y=2.1ex}
\begin{tikzpicture}[show background rectangle, {background rectangle/.style}={fill=backcolour}]
	\begin{pgfonlayer}{nodelayer}
		\node [style=reg] (0) at (2, 1) {\normalsize $\procvar{f}$};
		\node [style=none] (1) at (1, 1) {};
		\node [style=black-faded] (2) at (-0.5, -0) {};
		\node [style=black-faded] (3) at (-1.5, -0) {};
		\node [style=reg] (4) at (2, -0.75) {\normalsize $\procvar{g}$};
		\node [style=none] (5) at (1, -1) {};
		\node [style=none] (6) at (-2, 1) {};
		\node [style=none] (7) at (1.25, -0.5) {};
		\node [style=none] (8) at (1.25, -1) {};
		\node [style=none] (9) at (1.75, 1) {};
		\node [style=none] (10) at (1.75, -0.5) {};
		\node [style=none] (11) at (1.75, -1) {};
		\node [style=none] (12) at (-0.75, 1) {};
	\end{pgfonlayer}
	\begin{pgfonlayer}{edgelayer}
		\draw [in=75, out=180, looseness=1.00] (1.center) to (2.center);
		\draw (2) to (3.center);
		\draw [in=-75, out=180, looseness=1.00] (5.center) to (2);
		\draw [in=180, out=0, looseness=1.00] (5.center) to (8.center);
		\draw (1.center) to (9.center);
		\draw (8.center) to (11.center);
		\draw (7.center) to (10.center);
		\draw (6.center) to (12.center);
		\draw [in=180, out=0, looseness=1.00] (12.center) to (7.center);
	\end{pgfonlayer}
\end{tikzpicture}}
\tikzset{x=1em, y=1.5ex}

\]
\end{minipage}
\medskip

%
%
%

\subsection{Semantics}\label{sec:pcsemantics}
In order to give semantics to the calculus, we assume a set $\mathcal{A}$ of actions, $\alpha$, $\beta$, \dots. Since, we will consider different sets of actions (for Hoare and Milner synchronisation), we assume them to be functions of type $\Names \to M$ for some monoid $(M,+,0)$. The support of an action $\alpha$ is the set $\{a_i \mid \alpha(a_i)\neq 0\}$. The alphabet of $\alpha$, written $\Alphabet(\alpha)$ is identified with its support.

For Hoare synchronisation, the monoid $M$ is $(2,\cup,0)$, while for Milner it is $(\mathbb{Z},+,0)$. In both cases, we will write $a_i$ for the function mapping the name $a_i$ to $1$ and all the others to $0$. For Milner synchronisation, write $\overline{a_i}$ for the function mapping $a_i$ to $-1$.

To give semantics to processes, we need a \emph{process declaration} for each $\procvar{f} \in \Var$. That is, an expression $\procvar{f} \syntaxdf \sum_{i \in I}\alpha_i.P_i$, for some finite set $I$, $\alpha_i \in \mathcal{A}$ and processes $P_i$ such that 
\begin{equation}\label{eq:gooddecl}
\{a_1, \dots a_{ar(\procvar{f})}\} \subseteq \bigcup_{i\in I}\Alphabet(\alpha_i) \cup \bigcup_{i\in I}\Alphabet(P_i)\end{equation}
The basic behaviour of process declarations is captured by the three rules below.
\begin{equation}\label{eq:opgeneral}\derivationRule{}{\procvar{f} \tr{0}\procvar{f}}{} \qquad \derivationRule{\procvar{f} \syntaxdf \sum_{i \in I}\alpha_i.P_i}{\procvar{f} \tr{\alpha_i}P_i}{} \qquad \derivationRule{P \tr{\alpha}P'}{P\perm \tr{\alpha\circ \perm} P' \perm}{}\end{equation}

\begin{example}\label{ex:pcseconf}
Recall $\procvar{f}$ and $\procvar{g}$ from Example \ref{ex:pcfirst}. Now assume declarations $\procvar{f} \syntaxdf  a_1.\nu a_2(\procvar{f}[a_2/a_1] \,| \, \procvar{g})$ and $\procvar{g} \syntaxdf a_1.\procvar{g}+a_2.\procvar{g}$. Observe that they respect \eqref{eq:gooddecl}. We have that $\procvar{g}\tr{a_1}\procvar{g}$ and $\procvar{g}\tr{a_2}\procvar{g}$ while $\procvar{f}\tr{a_1}\nu a_2 (\procvar{f}[a_2/a_1] \, |\, \procvar{g})$. Similarly $\procvar{f} [a_2/a_1]\tr{a_2}(\nu a_2 (\procvar{f}[a_2/a_1] \, |\, \procvar{g}))[a_2/a_1]$.
\end{example}
To define the semantics of parallel and restriction, we need to distinguish between the Hoare and Milner  
synchronisation patterns.

\smallskip
\noindent\textbf{Hoare synchronisation.\label{sec:Hoare}} 
Here actions are 
functions $\alpha\colon \Names \to 2$, 
which can equivalently be thought of as subsets of $\Names$.
The synchronisation mechanism presented below is analogous to the one used in CSP~\cite{Hoarebook}.
The main difference is the level of concurrency: the classical semantics~\cite{Hoarebook} is purely interleaving, while for us it is a step semantics. Essentially, in $P|Q$, the processes $P$ and $Q$ may evolve independently on the non-shared names, i.e. the evolution of two or more processes may happen at the same time. It is for this reason that our actions are sets of names.
The operational semantics of parallel and restriction is given by rules
\begin{equation}\label{eq:Hoare}\derivationRule{P\tr{\alpha}P' \qquad Q\tr{\beta} Q' \quad \alpha\cap \Alphabet(Q)=\beta \cap \Alphabet(P)}{P|Q \tr{\alpha \cup \beta} P'|Q'}{} \qquad \derivationRule{P \tr{\alpha} P'}{\nu a_i(P) \tr{\alpha \setminus \{a_i\}} \nu a_i(P') }{}{}\end{equation}
We write $\Htr{\alpha}$ for the transition systems 
generated by the rules~\eqref{eq:opgeneral}, \eqref{eq:Hoare}.
By a simple inductive argument, using \eqref{eq:gooddecl} as base case, we see that for all processes $P$, if $P \tr{\alpha} P'$ then $\alpha \subseteq \Alphabet(P)$.
The rule for parallel, therefore, ensures that $P$ and $Q$ synchronise over all of their shared names. 
%
The rule for restriction hides $a_i$ from the environment. For instance, if $\alpha=\{a_i\}$, then  $\nu a_i(P)\tr{\emptyset}\nu a_i(P')$. If $\alpha=\{a_j\}$ with $a_j\neq a_i$, then $\nu a_i(P)\tr{\{a_j\}}\nu a_i(P')$.

\begin{example}
Recall $\procvar{f}$ and $\procvar{g}$ from Example \ref{ex:pcseconf}. We have that $\procvar{f}\Htr{a_1}\nu a_2 (\procvar{f}[a_2/a_1] \, |\, \procvar{g})$.
From $\nu a_2 (\procvar{f}[a_2/a_1] \, |\, \procvar{g})$, there are two possibilities: either $\procvar{f}[a_2/a_1]$ and $\procvar{g}$ synchronise on $a_2$, and in this case we have $\nu a_2 (\procvar{f}[a_2/a_1] \, |\, \procvar{g}) \tr{\emptyset}$, or $\procvar{g}$ proceeds without synchronising on $a_1$, therefore $\nu a_2 (\procvar{f}[a_2/a_1] \, |\, \procvar{g}) \Htr{\{a_1\}}$ since  $a_1$ belongs to $\Alphabet(\procvar{g})$ and not to $\Alphabet(\procvar{f}[a_2/a_1])$. 
\end{example}

\smallskip
\noindent\textbf{Milner synchronisation.} We take $\mathcal{A}=\mathbb{Z}^\Names$. 
Sum of functions, denoted by $+$, is defined pointwise and we 
write $0$ for its unit, the constant $0$ function. 
\begin{equation}\label{eq:Milner}
\derivationRule{P\tr{\alpha}P' \qquad Q\tr{\beta} Q'}{P|Q \tr{\alpha+\beta} P'|Q'}{} \qquad \derivationRule{P \tr{\alpha} P'}{\nu a_i(P) \tr{\alpha} \nu a_i(P') }{\alpha(a_i)=0}
\end{equation}
We write $\Mtr{\alpha}$ for the transition system generated by the rules \eqref{eq:opgeneral}, \eqref{eq:Milner}.

Functions in $\mathbb{Z}^\Names$ to represent concurrent occurrences of CCS send and receive actions. A single CCS action $a$ is the function mapping $a$ to $1$ and all other names to $0$. Similarly, the action $\bar{a}$ maps $a$ to $-1$ and the other names to $0$. The silent action $\tau$ is the function $0$.
With this in mind, it is easy to see that, similarly to CCS, the rightmost rule forbids $\nu a_i(P) \tr{\alpha} \nu a_i(P')$ whenever $\alpha=a_i$ or $\alpha=\bar{a_i}$. CCS-like synchronisation is obtained by the leftmost rule: when $\alpha=a_i$ and $\beta=\bar{a_i}$, one has that $P|Q \tr{0}P'|Q'$.

A simple inductive argument confirms that $P\tr{0}P$ for any process $P$. Then, by the leftmost rule again, one has that whenever $Q\tr{\beta}Q'$, then $P|Q\tr{\beta}P|Q'$. Note, however, that as in~\S~\ref{sec:Hoare}, while our synchronisation mechanism is essentially Milner's CSS handshake, our semantics is not interleaving and allows for step concurrency.  It is worth remarking that the operational rules in \eqref{eq:Milner} have already been studied by Milner in its work on SCCS~\cite{milner1983calculi}.

\smallskip
\noindent\textbf{Semantic correspondence.} For an action $\alpha\colon \Names \to M$ with $\Alphabet(\alpha)\subseteq \{a_1, \dots a_n\}$, we write $n\vdash \alpha$ for the restriction $\{a_1,\dots, a_n\} \to M$.
Define coalgebras $\beta_b, \beta_w \colon \Sigma \to \overline{\Pow}(L\poi \SCW(\Sigma) \poi L)$ for each $\procvar{f}\in \Sigma_{n,0}$ where 
$\procvar{f} \syntaxdf \sum_{i \in I}\alpha_i.P_i$ as
$$\beta_b(\procvar{f})=\beta_w(\procvar{f})=\left\{\big( n\vdash \alpha_i, \enc{P_i}, \zerodv  \big) \mid i\in I\} \cup \{\big(n\vdash 0 , \procvar{f}, \zerodv \big)\right\}\text{.}$$
For both $\beta_b$ and $\beta_w$, $L$ is the span $\N \tl{|\cdot |}A^*\tr{|\cdot |} \N$, but $A=2$ for $\beta_b$ and $A=\mathbb{Z}$ for $\beta_w$.

Via the distributive law  (\S~\ref{sec:comp-cw}) for the black Frobenius, we obtain the coalgebra $\beta_b^{\sharp} \colon \SCW(\Sigma) \to \overline{\Pow}(L\poi \SCW(\Sigma) \poi L)$. Via the white Frobenius, we obtain 
$\beta_w^{\sharp} \colon \SCW(\Sigma) \to \overline{\Pow}(L\poi \SCW(\Sigma) \poi L)$. We write $c \Bdtrans{\alpha}{\beta}d$ for $(\alpha,\beta, d) \in \beta_b^\sharp (c)$ and $c \Wdtrans{\alpha}{\beta}d$ for  $(\alpha,\beta, d) \in \beta_w^\sharp (c)$.
The correspondence can now be stated formally.

\begin{theorem}\label{thmHMbw} Let $n \vdash P$ and $n \vdash \alpha$ such that $\Alphabet(\alpha)\subseteq \Alphabet(P)$.
\begin{itemize}
\item {\bf Hoare is black.} If $P\Htr{\alpha} P'$ then $\semProc{n}{P}\Bdtrans{\zerodv}{n\vdash \alpha}\semProc{n}{P'}$.
Vice versa, if 
$\semProc{n}{P}\Bdtrans{\zerodv}{n\vdash \alpha}\!\diagState{n}{d}$
then there is $n\vdash P'$ s.t. $P\!\Htr{\alpha}\! P'$ and $\semProc{n}{P'}\!=\!\diagState{n}{d}$. 
\item{\bf Milner is white.} If $P\Mtr{\alpha} P'$ then $\semProc{n}{P}\Wdtrans{\zerodv}{n\vdash \alpha}\semProc{n}{P'}$. Vice versa, if $\semProc{n}{P} \!\Wdtrans{\zerodv}{n\vdash \alpha}\diagState{n}{d}$
then there is $n\vdash P'$ s.t. $P\!\Mtr{\alpha}\! P'$ and $\semProc{n}{P'}\!=\!\diagState{n}{d}$. 
\end{itemize}
\end{theorem}

\begin{example}
We illustrate the semantic correspondence by returning to Example~\ref{ex:pcseconf}. Diagrammatically, it yields the following transitions:
\[
\tikzset{x=1em, y=2.1ex}
\begin{tikzpicture}[show background rectangle, {background rectangle/.style}={fill=backcolour}]
	\begin{pgfonlayer}{nodelayer}
		\node [style=none] (0) at (1, 1) {};
		\node [style=black] (1) at (-0.5, -0) {};
		\node [style=black] (2) at (-1.5, -0) {};
		\node [style=none] (3) at (1, -1) {};
		\node [style=none] (4) at (-2, 1) {};
		\node [style=none] (5) at (1.25, -0.5) {};
		\node [style=none] (6) at (1.25, -1) {};
		\node [style=none] (7) at (1.75, 1) {};
		\node [style=none] (8) at (1.75, -0.5) {};
		\node [style=none] (9) at (1.75, -1) {};
		\node [style=none] (10) at (-0.75, 1) {};
		\node [style=reg] (11) at (2, 1) {\normalsize $\procvar{f}$};
		\node [style=reg] (12) at (2, -0.75) {\normalsize $\procvar{g}$};
	\end{pgfonlayer}
	\begin{pgfonlayer}{edgelayer}
		\draw [in=75, out=180, looseness=1.00] (0.center) to (1.center);
		\draw (1) to (2.center);
		\draw [in=-75, out=180, looseness=1.00] (3.center) to (1);
		\draw [in=180, out=0, looseness=1.00] (3.center) to (6.center);
		\draw (0.center) to (7.center);
		\draw (6.center) to (9.center);
		\draw (5.center) to (8.center);
		\draw (4.center) to (10.center);
		\draw [in=180, out=0, looseness=1.00] (10.center) to (5.center);
	\end{pgfonlayer}
\end{tikzpicture}}
\tikzset{x=1em, y=1.5ex}
\Bdtrans{\zerodv}{1} 
\tikzset{x=1em, y=2.1ex}
\begin{tikzpicture}[show background rectangle, {background rectangle/.style}={fill=backcolour}]
	\begin{pgfonlayer}{nodelayer}
		\node [style=none] (0) at (1, 1) {};
		\node [style=black] (1) at (-0.5, -0) {};
		\node [style=black] (2) at (-1.5, -0) {};
		\node [style=none] (3) at (1, -1) {};
		\node [style=none] (4) at (-2, 1) {};
		\node [style=none] (5) at (1.25, -0.5) {};
		\node [style=none] (6) at (1.25, -1) {};
		\node [style=none] (7) at (1.75, 1) {};
		\node [style=none] (8) at (1.75, -0.5) {};
		\node [style=none] (9) at (1.75, -1) {};
		\node [style=none] (10) at (-0.75, 1) {};
		\node [style=reg] (11) at (2, 1) {\normalsize $\procvar{f}$};
		\node [style=reg] (12) at (2, -0.75) {\normalsize $\procvar{g}$};
	\end{pgfonlayer}
	\begin{pgfonlayer}{edgelayer}
		\draw [in=75, out=180, looseness=1.00] (0.center) to (1.center);
		\draw (1) to (2.center);
		\draw [in=-75, out=180, looseness=1.00] (3.center) to (1);
		\draw [in=180, out=0, looseness=1.00] (3.center) to (6.center);
		\draw (0.center) to (7.center);
		\draw (6.center) to (9.center);
		\draw (5.center) to (8.center);
		\draw (4.center) to (10.center);
		\draw [in=180, out=0, looseness=1.00] (10.center) to (5.center);
	\end{pgfonlayer}
\end{tikzpicture}}
\tikzset{x=1em, y=1.5ex}
\quad \text{ and } \quad 
\tikzset{x=1em, y=2.1ex}
\begin{tikzpicture}[show background rectangle, {background rectangle/.style}={fill=backcolour}]
	\begin{pgfonlayer}{nodelayer}
		\node [style=none] (0) at (1, 1) {};
		\node [style=black] (1) at (-0.5, -0) {};
		\node [style=black] (2) at (-1.5, -0) {};
		\node [style=none] (3) at (1, -1) {};
		\node [style=none] (4) at (-2, 1) {};
		\node [style=none] (5) at (1.25, -0.5) {};
		\node [style=none] (6) at (1.25, -1) {};
		\node [style=none] (7) at (1.75, 1) {};
		\node [style=none] (8) at (1.75, -0.5) {};
		\node [style=none] (9) at (1.75, -1) {};
		\node [style=none] (10) at (-0.75, 1) {};
		\node [style=reg] (11) at (2, 1) {\normalsize $\procvar{f}$};
		\node [style=reg] (12) at (2, -0.75) {\normalsize $\procvar{g}$};
	\end{pgfonlayer}
	\begin{pgfonlayer}{edgelayer}
		\draw [in=75, out=180, looseness=1.00] (0.center) to (1.center);
		\draw (1) to (2.center);
		\draw [in=-75, out=180, looseness=1.00] (3.center) to (1);
		\draw [in=180, out=0, looseness=1.00] (3.center) to (6.center);
		\draw (0.center) to (7.center);
		\draw (6.center) to (9.center);
		\draw (5.center) to (8.center);
		\draw (4.center) to (10.center);
		\draw [in=180, out=0, looseness=1.00] (10.center) to (5.center);
	\end{pgfonlayer}
\end{tikzpicture}}
\tikzset{x=1em, y=1.5ex}
\Bdtrans{\zerodv}{0} 
\tikzset{x=1em, y=2.1ex}
\begin{tikzpicture}[show background rectangle, {background rectangle/.style}={fill=backcolour}]
	\begin{pgfonlayer}{nodelayer}
		\node [style=none] (0) at (2, 1.5) {};
		\node [style=black] (1) at (-0.5, -0) {};
		\node [style=black] (2) at (-1.5, -0) {};
		\node [style=none] (3) at (1, -1) {};
		\node [style=none] (4) at (-2, 1) {};
		\node [style=none] (5) at (1.25, -0.5) {};
		\node [style=none] (6) at (1.25, -1) {};
		\node [style=none] (7) at (1.75, -0.5) {};
		\node [style=none] (8) at (1.75, -1) {};
		\node [style=none] (9) at (-0.75, 1) {};
		\node [style=none] (10) at (4.75, -0.5) {};
		\node [style=none] (11) at (4.25, -0.5) {};
		\node [style=none] (12) at (4.75, 1.5) {};
		\node [style=none] (13) at (4.5, -0) {};
		\node [style=none] (14) at (2.25, 1.5) {};
		\node [style=black] (15) at (1.5, 0.5) {};
		\node [style=none] (16) at (4.75, -0) {};
		\node [style=black] (17) at (2.5, 0.5) {};
		\node [style=none] (18) at (4, -0.5) {};
		\node [style=none] (19) at (4, 1.5) {};
		\node [style=reg] (20) at (5, 1.5) {\normalsize $\procvar{f}$};
		\node [style=reg] (21) at (5, -0.25) {\normalsize $\procvar{g}$};
		\node [style=reg] (22) at (2, -0.75) {\normalsize $\procvar{g}$};
	\end{pgfonlayer}
	\begin{pgfonlayer}{edgelayer}
		\draw [in=75, out=180, looseness=1.00] (0.center) to (1.center);
		\draw (1) to (2.center);
		\draw [in=-75, out=180, looseness=1.00] (3.center) to (1);
		\draw [in=180, out=0, looseness=1.00] (3.center) to (6.center);
		\draw (6.center) to (8.center);
		\draw (5.center) to (7.center);
		\draw (4.center) to (9.center);
		\draw [in=180, out=0, looseness=1.00] (9.center) to (5.center);
		\draw [in=75, out=180, looseness=1.00] (19.center) to (17);
		\draw (17) to (15);
		\draw [in=-75, out=180, looseness=1.00] (18.center) to (17);
		\draw [in=180, out=0, looseness=1.00] (18.center) to (11.center);
		\draw (19.center) to (12.center);
		\draw (11.center) to (10.center);
		\draw (13.center) to (16.center);
		\draw [in=180, out=0, looseness=1.00] (14.center) to (13.center);
		\draw [in=180, out=0, looseness=1.25] (0.center) to (14.center);
	\end{pgfonlayer}
\end{tikzpicture}}
\tikzset{x=1em, y=1.5ex}
\Bdtrans{\zerodv}{0}\dots\]
\end{example}

\section{Related and Future Work}\label{sec:conclusion}
The terminology \emph{Hoare and Milner synchronisation} is used in Synchronised Hyperedge Replacement (SHR)~\cite{degano1987model,DBLP:journals/entcs/LaneseM06}. 
Our work is closely related to SHR: indeed, the prop
$\SCW(\Sigma)$ has arrows open hypergraphs, where hyperedges are labeled with elements of $\Sigma$~\cite{BonchiGKSZ16}. To define a coalgebra $\beta \colon \Sigma \to \Fbeh\SCW(\Sigma)$ is to specify a transition system for each label in $\Sigma$. Then, constructing the coalgebra $\beta^{\sharp} \colon \SCW(\Sigma) \to \Fbeh\SCW(\Sigma)$ from a distributive law amounts to giving a transition system to all hypergraphs according to some synchronisation policy (e.g.\ \`{a} la Hoare or \`{a} la Milner).
%
%
%
SHR systems equipped with Hoare and Milner synchronisation are therefore instances of our approach. A major difference 
is our focus on the algebraic aspects: e.g.\ since string diagrams can be regarded as syntax as well as  combinatorial entities, their syntactic nature allows for the bialgebraic approach, and simple inductive proofs. 
The operational rules in Figure~\ref{sos2} are also those of tile systems~\cite{Gadducci2000}. However, in the context of tiles, transitions are arrows of the vertical \emph{category}: this forces every state to perform at least one identity transition. For example, it is not possible to consider empty sets of transitions, which can be a useful feature in the string diagrammatic approach, see~\cite{BonchiLICS2019}.

Amongst the many other related models,
 it is worth mentioning bigraphs~\cite{milner2009space}. While also graphical, bigraphs can be nested hierarchically, a capability that we have not considered. Moreover, the behaviour functor $\Fbeh$ in~\S~\ref{sec:SOS} forces the labels and the arriving states to have the same sort as the starting states. Therefore, fundamental mobility mechanisms such as scope-extrusion cannot immediately be addressed within our framework. We are confident, however, that 
 the solid algebraic foundation we have laid here for the operational semantics of two-dimensional syntax will be needed to shed light on such concepts as hierarchical composition and mobility.
Some ideas may come from~\cite{brunihierarchical}.

\medskip
\noindent\textbf{Acknowledgements.} RP and FZ acknowledge support from EPSRC grant EP/R020604/1.



\bibliography{catBib3}

\appendix

\section{Additional Background on Bialgebraic Semantics}
\label{app:preliminaries}

\textbf{(Co)algebras for endofunctors.} Given a category $\catC$ and a functor $\funF\colon \catC \to \catC$, an \emph{$\funF$-algebra} is a pair $(X,\alpha)$ for an object $X$  and an arrow $\alpha\colon \funF X \to X$  in $\catC$. An $\funF$-algebra morphism $f\colon (X,\alpha) \to (X',\alpha')$ is an arrow $f \colon X \to X'$ such that $f\circ \alpha=\alpha' \circ \funF f$. The category of $\funF$-algebras and their morphisms is denoted by $\Alg{\funF}$.

Coalgebras are defined dually: an \emph{$\funF$-coalgebra} is a pair $(X,\beta)$ for an object $X$  and an arrow $\alpha\colon X \to \funF X$  in $\catC$; an $\funF$-coalgebra morphism $f\colon (X,\beta) \to (X',\beta')$ is an arrow $f \colon X \to X'$ such that $ \alpha' \circ f= \funF f \circ \alpha$. The category of $\funF$-coalgebras and their morphisms is denoted by $\Coalg{\funF}$. 

When  $\catC$ is $\SET$, coalgebras can be thought as state-based systems: an \emph{$\funF$-coalgebra} $(X,\beta)$ consists of a set of states $X$ and a ''transition'' function $\beta\colon X \to \funF X$. The functor $\funF$ define the type of the transitions: for instance, coalgebras for the functor $\funF X=2 \times X^A$ are deterministic automata (see Appendix~\ref{app:exampleautomata}). When a final object $(\Omega,\omega)$ exists in $\Coalg{\funF}$, this can be thought as a universe of all possible $\funF$-behaviours. The unique $\funF$-coalgebra morphism  $\sem{\cdot} \colon (X,\beta) \to (\Omega,\omega)$
is a function mapping every state of $X$ into its behaviour.

\medskip

\noindent \textbf{Monads.} An endofunctor $\funT$ forms a \emph{monad} when it is equipped with natural transformations $\eta \colon Id \to \funT$ and $\mu \colon \funT \funT \to \funT$ such that $\mu \circ \eta T = Id_{\funT} = \mu \circ T \eta$ and $\mu \circ T \mu = \mu \circ \mu T$. An \emph{Eilenberg-Moore algebra for a monad $(\mathcal{T} , \eta, \mu)$} is a $\funT$-algebra $\alpha \colon \funT X \to X$ such that 
$\alpha \circ \eta_X =id_X$ and $g \circ \mu_X=g \circ \funT g$. Morphisms are simply morphisms of $\funT$-algebras. The category of 
Eilenberg-Moore algebras for $\mathcal{T}$ and their morphisms is denoted by $\EM{\mathcal{T}}$. For the sake of brevity, when $\funT$ is a monad, $\funT$-algebra will not mean algebra for the endofunctor $\funT$, but Eilenberg-Moore algebra for the monad $\funT$.

A \emph{monad morphism}  from a monad $(\funT,\eta,\mu)$ to a monad $(\funT',\eta',\mu')$ on the same category $\catC$ is a natural transformation $\kappa \colon \funT \to \funT'$ such that $\kappa \circ \eta = \eta'$ and $\kappa \circ \mu = \mu' \circ \kappa \kappa$.

A recurrent monad in our work is the powerset monad $\Pow$ bounded by a cardinal $\kappa$. It maps a set $X$ to the set $\Pow X = \{U \mid U \subseteq X, \,\, U \textrm{ has cardinality at most }\kappa\}$ and a function $f\colon X\to Y$ to $\Pow f\colon \Pow X\to \Pow Y$, $\Pow f (U) = \{f(u) \mid u\in U\}$. The unit $\eta$ of $\Pow$ is given by singleton, i.e., $\eta(x) = \{x\}$ and the multiplication $\mu$ is given by union, i.e., $\mu(S) = \bigcup_{U \in S} U$ for $S \in \Pow\Pow X$. 

For instance, with $\kappa = \omega$, $\Powf$ is the \emph{finite} powerset monad, mapping a set to its finite subsets.

\smallskip

\noindent \textbf{Coproduct of GSOS specifications.} Suppose to have two functors $\mathcal{S}_1,\mathcal{S}_2 \colon \catC \to \catC$ modelling two syntaxes and a functor $\funF\colon \catC \to \catC$ modelling the coalgebraic behaviour, such that there are two GSOS specifications
$$\lambda_1\colon \mathcal{S}_1 \funF \To \funF\mathcal{S}_1^\dagger \; \text{ and }\; \lambda_2\colon \mathcal{S}_2 \funF \To \funF\mathcal{S}_2^\dagger\text{.}$$
Then we can construct a GSOS specification 
$$\lambda_1\cdot \lambda_2 \colon (\mathcal{S}_1+\mathcal{S}_2) \funF \To  \funF  (\mathcal{S}_1+\mathcal{S}_2)^{\dagger}$$
as follows
$$\xymatrix{
& \mathcal{S}_1 \funF \ar[d]_{\gamma_1} \ar[r]^{\lambda_1} & \funF \mathcal{S}_1^{\dagger} \ar[d]^{\funF \iota_1} \\
(\mathcal{S}_1+\mathcal{S}_2)\funF  \ar@{}[r]|{\cong}&  \mathcal{S}_1 \funF  +  \mathcal{S}_2 \funF \ar@{-->}[r]& \funF (\mathcal{S}_1+\mathcal{S}_2)^{\dagger} \\
& \mathcal{S}_2 \funF \ar[u]^{\gamma_1}  \ar[r]_{\lambda_2} & \funF \mathcal{S}_2^{\dagger} \ar[u]_{\funF \iota_1}
}
$$
In the above diagram, the dashed map is given by universal property of the coproduct $\mathcal{S}_1 \funF  +  \mathcal{S}_2 \funF$. The definitions of $\gamma_1$ and $\gamma_2$ are as follows. For $X \in \catC$, $\gamma_1(X)$ is the unique $X + \mathcal{S}_1$-algebra map from the initial $X + \mathcal{S}_1$-algebra $\mathcal{S}_1^{\dagger}(X)$ to $(\mathcal{S}_1+\mathcal{S}_2)^{\dagger}(X)$. Indeed, since $(\mathcal{S}_1+\mathcal{S}_2)^{\dagger}(X)$ is the initial $X+\mathcal{S}_1+\mathcal{S}_2$-algebra, it is in particular a $X + \mathcal{S}_1$-algebra by precomposition with the coproduct map $X + \mathcal{S}_1 \tr{} X+\mathcal{S}_1+\mathcal{S}_2$. The definition of $\gamma_2(X)$ is analogous, using the fact that $(\mathcal{S}_1+\mathcal{S}_2)^{\dagger}(X)$ is also a $X + \mathcal{S}_2$-algebra.

\section{An Example of Bialgebraic Semantics: Non-Deterministic Automata}\label{app:exampleautomata}

In this appendix, we illustrate bialgebraic semantics on a well-known case study, namely non-deterministic automata~\cite{SilvaBBR10}. The intention is to provide in full a basic example that may serve as a roadmap for the approach to string diagrams proposed in the main text. 

\smallskip
\noindent{\bf{Deterministic automata as coalgebras.}} 
Let $2$ be the set $\{0,1\}$ and $A$ an alphabet of symbols.
A deterministic automata (DA) is a pair  $(S, \langle o,t\rangle)$ where $S$ is a set of state and $\langle o,t \rangle \colon S \to 2 \times S^A$ consists of the \emph{output function} $o\colon S \to 2$, defining whether a state $x\in S$ is accepting ($o(x)=1$) or not ($o(x)=0$), and the \emph{transition function} $t\colon S \to S^A$ mapping each state $x\in S$ and each $a\in A$ the successor state $t(x)(a)$.

DAs are in one to one correspondence with coalgebras for the functor $\funF\colon \SET \to \SET$ defined as $\funF(X)=2\times X^A$. The set of all languages over the alphabet $A$, hereafter denoted by $2^{A^*}$, carries a final coalgebra. For each deterministic automaton $(S, \langle o,t\rangle)$, there is a unique coalgebra homomorphism $\sem{\cdot} \colon S \to 2^{A^*}$ assigning to each state in $x$ the language that it accepts (defined for all words $w\in A^*$ as $\sem{x}(\epsilon)=o(x)$ and $\sem{x}(aw)=\sem{t(x)(a)}(w)$).

\smallskip
\noindent
{\bf{Non deterministic automata give rise to bialgebras.}} 
A non deterministic automata (NDA) is a pair  $(S, \langle o,t\rangle)$ where $S$ and $o$ are like for deterministic automata, but the transition function $t$  has now type $S \to \Powf{S}^A$, where $\Powf$ is the finite powerset functor, see Appendix \ref{app:preliminaries}. $t$ maps each state $x\in S$ and each $a\in A$ to a (finite) set of possible successor states $t(x)(a)$.

Therefore NDA are coalgebras for the functor $\funF \Powf$ where $\funF$ is the functor for deterministic automata explained above. Traditionally the semantics of NDA is also defined in terms of languages, but the set $2^{A^*}$ does not carry a final coalgebra for $\funF \Powf$. The solution is to view NDAs as DAs, by a construction that in automata theory is usually called determinisation (or powerset construction). Categorically, this amounts to transform an NDA $\beta \colon S \to 2\times \Powf(S)^A$ into a DA $\beta^{\sharp} \colon \Powf(S) \to 2\times \Powf(S)^A$, i.e. an $\funF \Powf$-coalgebra into an $\funF$-coalgebra. For the determinised NDA $\beta^{\sharp}$, there exists a unique $\Fbeh$-coalgebra homomorphism $\sem{\cdot} \colon \Powf(S) \to 2^{A^*}$. This yields language semantics $S \to 2^{A^*}$ for the original statespace $S$, by precomposing with the unit for the monad $\eta_S \colon S \to \Powf(S)$, as in the diagram below.
\begin{equation}\label{eq:NDAexample}
\vcenter{
\xymatrix{
S \ar[r]^{\eta} \ar[dd]_{\beta} & \Powf(S) \ar[ddl]|{\beta^{\sharp}} \ar@{-->}[rr]^{\sem{\cdot}} & & 2^{A^*} \ar[dd]\\
\\
\funF\Powf(S) \ar[rrr]_{\funF(\sem{\cdot})} & & & \funF\Powf(2^{A^*})
}
}
\end{equation}
In automata theory, there is a concrete way of constructing $\beta^{\sharp}$ from $\beta$. In the categorical abstraction, the general principle underpinning this construction is that defining $\beta^{\sharp}$ requires the introduction of a \emph{distributive law} $\lambda \colon\Powf \Fbeh \Rightarrow\Fbeh \Powf$. For all sets $X$, $\lambda_X\colon \Powf(2 \times S^A) \to 2\times \Powf(X)^A$ is defined for any $b_1, \dots, b_n \in 2$ and for all $\phi_1,\dots, \phi_n \in S^A$,
\begin{equation}\label{eq:lawdet}\lambda_X (\{\langle b_1,\phi_1\rangle , \dots , \langle b_n,\phi_n\rangle \}) = \langle \bigsqcup_{i \in 1\dots n} b_i, \bigsqcup_{i \in 1\dots n}  \phi_i \rangle  \end{equation}
where the leftmost $\sqcup$ is the obvious join in $2$ (regarded as the partial order $0\sqsubseteq 1$), and the rightmost one is defined as the point-wise union: namely for all $a\in A$, $\bigsqcup_{i \in 1\dots n}  \phi_i (a)=\{\phi_1(a), \dots , \phi_n(a)\}$. Observe that, in the particular case of the empty set $\emptyset$, $\alpha(\emptyset)=\langle 0, \phi_{\emptyset}\rangle$ where $\phi_{\emptyset}(a)=\emptyset$ for all $a\in A$.

Given $\lambda$, $\beta^{\sharp}$ is defined as $\Powf(S) \tr{\Powf{\beta}} \Powf\Fbeh\Powf(S) \tr{\lambda_{\Powf(S)}} \Fbeh\Powf\Powf(S) \tr{\Fbeh\mu} \Fbeh\Powf(S)$. The reader could also check that, with this definition, $\beta^{\sharp}$ is a $\lambda$-bialgebra, where the algebraic structure on the state space $\Powf(S)$ is given by the multiplication $\mu_S \colon \Powf\Powf(S) \to \Powf(S)$. 

\smallskip

\noindent{\bf{Algebraic presentation of $\Powf$.}} In concrete, determinisation is about obtaining the language semantics of NDAs. The essence of its categorical abstraction is moving from coalgebras to bialgebras, thus taking into account the \emph{algebraic} structure of the statespace $\Powf(S)$. Indeed, it is well-known that the monad $\Powf$ is presented by the algebraic theory of join semilattice with bottom: these consist of a set $S$ equipped with an operation $\otimes \colon S \times S \to S$ and an element $\bm{0} \colon 1\to S$ such that 
\begin{equation} \label{eq:axiomssemilattices}
(x\otimes y) \otimes z = x \otimes (y \otimes z) \quad x\otimes y= y \otimes z \quad x\otimes x=x \quad x\otimes \bm{0}=x 
\end{equation}
for all $x,y,z\in S$. 
To make formal the link with $\Powf$, let us introduce the functor $\mathcal{T}_{\mathcal{S}}$ which maps every set $S$ to the set $\mathcal{T}_{\mathcal{S}}(S)$ of terms built with $\otimes$, $\bm{0}$ and variables in $S$. This functor carries the structure of a monad, where the unit $\eta_S\colon S \to \mathcal{T}_{\mathcal{S}}(S)$is given by inclusion (each variable is a term) and the multiplication $\mu_S \colon \mathcal{T}_{\mathcal{S}}\mathcal{T}_{\mathcal{S}}(S) \to \mathcal{T}_{\mathcal{S}}(S)$ by substitution.
The fact that $\Powf$ is presented by the algebraic theory of join semilattice with bottom means that $\Powf$ is (isomorphic to) the quotient of the monad $\mathcal{T}_{\mathcal{S}}$ by equations \eqref{eq:axiomssemilattices}.

\smallskip
\noindent{\bf{Modular construction of the distributive law.}} The monad $\Powf$ is associated with a rather elementary algebraic theory, which allows $\lambda$ to be defined in a rather straightforward way. The monads that we use in the main text to encode categorical structures are way more involved. For this reason, it is important to modularise the task of finding a distributive law. Again, we are going to use $\lambda \colon\Powf \funF \Rightarrow \funF \Powf$ in \eqref{eq:lawdet} as a proof of concept, and show how its definition can be divided into different stages. 

By the above discussion, $\Powf$ is the quotient of $\mathcal{T}_{\mathcal{S}}$ by equations \eqref{eq:axiomssemilattices}, but we can decompose this even further: $\mathcal{T}_{\mathcal{S}}$ is the free monad (see \S~\ref{sec:background})  $(\mathcal{S}_1+\mathcal{S}_2)^{\dagger}$ over the endofunctor expressing the semilattice signature: here $\mathcal{S}_1 \colon S \mapsto \{\bm{0}\}$ stands for the element $\bm{0}$, and $\mathcal{S}_2 \colon S \mapsto S \times S$ encodes the operation $\otimes$.

We can now divide the construction of $\lambda$ in three  stages:
\begin{description}
	\item[Distributive law on the signature.] The first step is to give separate distributive laws of $\mathcal{S}_1$ and $\mathcal{S}_2$ over $\Fbeh$. The first one is $\lambda^1 \colon \mathcal{S}_1\Fbeh \Rightarrow \Fbeh \mathcal{S}_1$. Given a set $S$, we let $\lambda^1_S \colon \{\bm{0}\} \to 2 \times \{\bm{0}\}^A$ map $\bm{0}$ to the pair $\langle 0, ! \rangle$ where $!$ is the unique function of type $ A \to \{\bm{0}\}$. This distributive law can also be conveniently presented as the following pair of derivation rules. The first describes the element $0$ (= $\bm{0}$ is not a terminating state) of the pair $\lambda^1_S\{\bm{0}\}$; the second describes the element $! \colon A \to \{0\}$ of the pair $\lambda^1_S\{\bm{0}\}$.
	\begin{equation}\label{eq:NDA-DerRule1}
\derivationRule{}
{\bm{0} \not \downarrow }{} 	 \qquad  \qquad 	\derivationRule{a \in A}
{\bm{0} \dtrans{a}{} \bm{0}}{}
	\end{equation}

The second distributive law is $\lambda^1 \colon \mathcal{S}_2\Fbeh \to \Fbeh \mathcal{S}_2$. Given a set $S$, we let $\lambda^2_S \colon (2\times S^A) \times (2\times S^A)  \to 2 \times (S\times S)^A$ be defined for all $\langle b_1,\phi_1\rangle, \langle b_2,\phi_2\rangle \in  2\times S^A$ as $\lambda^2 (\langle b_1,\phi_1\rangle, \langle b_2,\phi_2\rangle) = \langle b_1\sqcup b_2, \langle \phi_1, \phi_2\rangle \rangle$. The representation in terms of derivation rules is:
	\begin{equation}\label{eq:NDA-DerRule2}
\derivationRule{s_1 \downarrow}
{ s_1 \otimes s_2 \downarrow }{} 	 
\qquad  \qquad
\derivationRule{s_2 \downarrow}
{ s_1 \otimes s_2 \downarrow }{} 	 
\qquad  \qquad
 	\derivationRule{s_1 \dtrans{a}{} s_1' \quad s_2 \dtrans{a}{} s_2'}
{s_1 \otimes s_2 \dtrans{a}{} s_1' \otimes s_2'}{} 
	\end{equation}

	\item[Distributive law on the free monad.] We can now put together $\lambda^1$ and $\lambda^2$ to obtain a distributive law of type $\colon \mathcal{T}_{\mathcal{S}} \funF \Rightarrow \funF \mathcal{T}_{\mathcal{S}}$. First, by post-composing $\lambda^1$ and $\lambda^2$ with the unit of the free monads $\eta\colon \mathcal{S}_i \Rightarrow \mathcal{S}_i^{\dagger}$ associated with endofunctors $\mathcal{S}_i$, one obtains two GSOS specifications   $\overline{\lambda^1}\colon \mathcal{S}_1 \funF \Rightarrow \funF\mathcal{S}_1^{\dagger}$ and 
$\overline{\lambda^2}\colon \mathcal{S}_2 \funF \Rightarrow \funF\mathcal{S}_2^{\dagger}$. As explained in Appendix \ref{app:preliminaries}, one can take the coproduct of GSOS specification, so to obtain $\overline{\lambda} \colon \mathcal{S} \funF \Rightarrow \funF \mathcal{S}^{\dagger}$ and thus, by Proposition \ref{prop:GSOSToDistrLaw} , a distributive law $\overline{\lambda}^{\dagger} \colon \mathcal{S}^{\dagger} \funF \Rightarrow \funF \mathcal{S}^{\dagger}$. Since $\mathcal{S}^{\dagger} = \mathcal{T}_{\mathcal{S}}$, this distributive law is effectively of the desired type $\mathcal{T}_{\mathcal{S}} \funF \Rightarrow \funF \mathcal{T}_{\mathcal{S}}$.
	\item[Distributive law on the quotient.] As $\Powf$ is a quotient of $\mathcal{T}_{\mathcal{S}}$ by equations \eqref{eq:axiomssemilattices}, the last step is quotienting the $\overline{\lambda}^{\dagger} \colon \mathcal{T}_{\mathcal{S}} \funF \Rightarrow \funF \mathcal{T}_{\mathcal{S}}$ by such equations. There is a categorical approach, described in \S~\ref{sec:background}, which allows to turn $\overline{\lambda}^{\dagger}$ into a distributive law $\lambda$ on the quotient of $\mathcal{T}_{\mathcal{S}}$, provided that $\overline{\lambda}^{\dagger}$ is compatible with  equations  \eqref{eq:axiomssemilattices}. This is indeed the case for $\overline{\lambda}^{\dagger}$, which thus yields by Proposition \ref{prop:distrlawquotient} a distributive law $\lambda \colon\Powf \funF \Rightarrow \funF \Powf$ of the desired type. The definition \eqref{eq:lawdet} of $\lambda$ given above can be equivalently described by the derivation rules \eqref{eq:NDA-DerRule1}-\eqref{eq:NDA-DerRule2}, modulo the algebraic representation of $\Powf$ in terms of join semilattices.
\end{description}

\section{More Details on Bialgebraic Semantics of String Diagrams}

\subsection{Existence of Final Coalgebras}
\label{sec:final}

\begin{proposition}\label{prop:finalcoalg}
	There exists a final coalgebra $\Omega \to \Fbeh(\Omega)$.
\end{proposition}
\begin{proof}
    Being isomorphic to a presheaf category, $\SIG$ is accessible. Then one may construct $\Omega$ via a so-called terminal sequence \cite{Worrell99}, provided that  (i) $\Fbeh$ preserves monomorphisms and (ii) is accessible. Because functor composition preserve these properties, it suffices to check them componentwise on $\Fbeh = \overline{\Pow}(L \poi Id \poi L)$. First, $(- \poi -)$ is defined by pullbacks, which preserve monomorphisms and (filtered) colimits; it follows that $(- \poi -)$ satisfies (i) and (ii). Next, the functor $\overline{\Pow}$ satisfies (i) for the same reason as the arbitrary powerset functor, and it satisfies (ii) because it is a $\kappa$-bounded functor. Finally, satisfaction of (i) and (ii) is completely obvious for $Id$ and the constant functor $L$. 
\end{proof}

Note that the proof of Proposition~\ref{prop:finalcoalg} is constructive: following \cite{Worrell99}, for a fixed $(n,m) \in \N \times \N$, we can visualise the elements of $\Omega$ as $\kappa$-branching trees with edges labelled with pairs of labels $(l_1,l_2) \in L \times L$ with $|l_1| = n, |l_2| = m$. The unique $\Fbeh$-coalgebra map $\sem{\cdot}_{\beta} \colon \SPROP(\Sigma) \to \Omega$ sends a $\Sigma$-term $t$ to the tree whose branching describes all the executions from $t$ according to the transition rules defined by $\beta$.

\subsection{Failure of Compositionality for Lawvere Theories}
\label{sec:Lawvere-fail}

We can define a monad $\SL$ for Lawvere theories following the usual recipe. The signature of a Lawvere theory is that of a prop with the additional comonoids. As for the comonoid part of Frobenius bimonoids (Fig.\ref{fig:Frob}), we only need to introduce this comonoid on the generating object $1$ since all the others can be obtained by parallel and sequential composition of these basic morphisms and the symmetries. Let $\Bcomult\from \SIG\to\SIG$ and $\Bcounit\from \SIG\to\SIG$ be the constant functors defined exactly as in~\S~\ref{sec:cw-monad}. 
$\SL$ is is then obtained as the quotient of $\SPROP + \Bcomult+\Bcounit$ by the defining equations of cocommutative comonoids (first line of Fig.~\ref{fig:Frob}), and the requirement that the comonoids be natural in the sense that every morphism $f\from m\to n$ should be a comonoid homomorphism:
\begin{equation}
  \label{eq:natural-comonoid}
  
\tikzset{x=1em, y=2.1ex}
\begin{tikzpicture}[show background rectangle, {background rectangle/.style}={fill=backcolour}, scale=.7, baseline=-.5ex]
	\begin{pgfonlayer}{nodelayer}
		\node [style=black] (0) at (-0.75, -0) {};
		\node [style=none] (1) at (-2, -0) {};
		\node [style=none] (2) at (1, -1.25) {};
		\node [style=none] (3) at (1, 1.25) {};
		\node [style=none] (4) at (1.5, -1.25) {};
		\node [style=none] (5) at (1.5, 1.25) {};
		\node [style=box] (6) at (-2.5, -0) {$f$};
		\node [style=none] (7) at (-3, -0) {};
		\node [style=none] (8) at (-4.5, -0) {};
		\node [style=none] (9) at (-4.25, 0.5) {\scriptsize $m$};
		\node [style=none] (10) at (1.25, -0.75) {\scriptsize $n$};
		\node [style=none] (11) at (1.25, 1.75) {\scriptsize $n$};
	\end{pgfonlayer}
	\begin{pgfonlayer}{edgelayer}
		\draw (0) to (1.center);
		\draw [in=75, out=180, looseness=1.00] (3.center) to (0);
		\draw [in=-75, out=180, looseness=1.00] (2.center) to (0);
		\draw (3.center) to (5.center);
		\draw (2.center) to (4.center);
		\draw (8.center) to (7.center);
	\end{pgfonlayer}
\end{tikzpicture}}
\tikzset{x=1em, y=1.5ex}
\;\; \approx\;\; 
\tikzset{x=1em, y=2.1ex}
\begin{tikzpicture}[show background rectangle, {background rectangle/.style}={fill=backcolour}, scale=.7, baseline=-.5ex]
	\begin{pgfonlayer}{nodelayer}
		\node [style=black] (0) at (-0.75, -0) {};
		\node [style=none] (1) at (2, -1.25) {};
		\node [style=none] (2) at (1, 1.25) {};
		\node [style=none] (3) at (1, -1.25) {};
		\node [style=none] (4) at (3.5, 1.25) {};
		\node [style=none] (5) at (3.5, -1.25) {};
		\node [style=none] (6) at (2, 1.25) {};
		\node [style=none] (7) at (-2.75, -0) {};
		\node [style=box] (8) at (1.5, 1.25) {\scriptsize $f$};
		\node [style=box] (9) at (1.5, -1.25) {\scriptsize $f$};
		\node [style=none] (10) at (3.25, 1.75) {\scriptsize $n$};
		\node [style=none] (11) at (3.25, -0.75) {\scriptsize $n$};
		\node [style=none] (12) at (-2.5, 0.5) {\scriptsize $m$};
	\end{pgfonlayer}
	\begin{pgfonlayer}{edgelayer}
		\draw (6.center) to (4.center);
		\draw [in=75, out=180, looseness=1.00] (2.center) to (0);
		\draw [in=-75, out=180, looseness=1.00] (3.center) to (0);
		\draw (7.center) to (0);
		\draw (1.center) to (5.center);
	\end{pgfonlayer}
\end{tikzpicture}}
\tikzset{x=1em, y=1.5ex}
 \qquad\qquad 
\tikzset{x=1em, y=2.1ex}
\begin{tikzpicture}[show background rectangle, {background rectangle/.style}={fill=backcolour}, scale=.7, baseline=-.5ex]
	\begin{pgfonlayer}{nodelayer}
		\node [style=black] (0) at (2, -0) {};
		\node [style=none] (1) at (0.5, -0) {};
		\node [style=box] (2) at (0, -0) {\scriptsize $f$};
		\node [style=none] (3) at (-0.5, -0) {};
		\node [style=none] (4) at (-2, -0) {};
		\node [style=none] (5) at (-1.75, 0.5) {\scriptsize $m$};
	\end{pgfonlayer}
	\begin{pgfonlayer}{edgelayer}
		\draw (0) to (1.center);
		\draw (4.center) to (3.center);
	\end{pgfonlayer}
\end{tikzpicture}}
\tikzset{x=1em, y=1.5ex}
\; \approx\; \Bcounitn{\scriptsize $m$}
\end{equation}
The intuitive interpretation is that every morphism can be copied and deleted using the comonoid structure.





To give a GSOS specification for $\SL$ it is natural to extend the specification for props with specifications for the comonoids. These are entirely analogous to the black interpretation of the comonoid part of CW props as copying and deleting operations as given in Fig.~\ref{sos1}. 
However, this specification turns out to be incompatible with the required naturality of comonoids. To see this, consider the following counter-example.

Let $\Sigma$ be the signature with a single symbol $d\from 0\to 1$. We specify its behaviour as the coalgebra $\beta\from \Sigma\to \Fbeh\SL\Sigma$ given by the two transitions
\[d \dtrans{\epsilon}{a} d \qquad d \dtrans{\epsilon}{b} d\]
Note the use of nondeterminism here: it is this form of nondeterminism combined with the naturality requirement of~\eqref{eq:natural-comonoid} that will lead to a contradiction. The derivation rules of \S~\ref{sec:comp-prop} give us only two possible transitions for $d\poi\Bcomult$:
\[\derivationRule{\diagEffect{}{d}\dtrans{\varepsilon}{a} \diagEffect{}{d} \quad \Bcomult\dtrans{a}{a \labelSep a} \Bcomult}
{
\tikzset{x=1em, y=2.1ex}
\begin{tikzpicture}[{background rectangle/.style}={fill=backcolour}, show background rectangle, scale=.7, baseline=-.5ex]
	\begin{pgfonlayer}{nodelayer}
		\node [style=black] (0) at (-0.75, -0) {};
		\node [style=none] (1) at (-1.75, -0) {};
		\node [style=none] (2) at (1, -1) {};
		\node [style=none] (3) at (1, 1) {};
		\node [style=coreg] (4) at (-2.5, -0) {\scriptsize $d$};
	\end{pgfonlayer}
	\begin{pgfonlayer}{edgelayer}
		\draw (0) to (1.center);
		\draw [in=75, out=180, looseness=1.00] (3.center) to (0);
		\draw [in=-75, out=180, looseness=1.00] (2.center) to (0);
	\end{pgfonlayer}
\end{tikzpicture}}
\tikzset{x=1em, y=1.5ex}
 \dtrans {\varepsilon}{a \labelSep a} 
\tikzset{x=1em, y=2.1ex}
\begin{tikzpicture}[{background rectangle/.style}={fill=backcolour}, show background rectangle, scale=.7, baseline=-.5ex]
	\begin{pgfonlayer}{nodelayer}
		\node [style=black] (0) at (-0.75, -0) {};
		\node [style=none] (1) at (-1.75, -0) {};
		\node [style=none] (2) at (1, -1) {};
		\node [style=none] (3) at (1, 1) {};
		\node [style=coreg] (4) at (-2.5, -0) {\scriptsize $d$};
	\end{pgfonlayer}
	\begin{pgfonlayer}{edgelayer}
		\draw (0) to (1.center);
		\draw [in=75, out=180, looseness=1.00] (3.center) to (0);
		\draw [in=-75, out=180, looseness=1.00] (2.center) to (0);
	\end{pgfonlayer}
\end{tikzpicture}}
\tikzset{x=1em, y=1.5ex}
}{\lambdaseqcomp} \qquad \derivationRule{\diagEffect{}{d}\dtrans{\varepsilon}{b} \diagEffect{}{d} \quad \Bcomult\dtrans{b}{b \labelSep b} \Bcomult}
{
\tikzset{x=1em, y=2.1ex}
\begin{tikzpicture}[{background rectangle/.style}={fill=backcolour}, show background rectangle, scale=.7, baseline=-.5ex]
	\begin{pgfonlayer}{nodelayer}
		\node [style=black] (0) at (-0.75, -0) {};
		\node [style=none] (1) at (-1.75, -0) {};
		\node [style=none] (2) at (1, -1) {};
		\node [style=none] (3) at (1, 1) {};
		\node [style=coreg] (4) at (-2.5, -0) {\scriptsize $d$};
	\end{pgfonlayer}
	\begin{pgfonlayer}{edgelayer}
		\draw (0) to (1.center);
		\draw [in=75, out=180, looseness=1.00] (3.center) to (0);
		\draw [in=-75, out=180, looseness=1.00] (2.center) to (0);
	\end{pgfonlayer}
\end{tikzpicture}}
\tikzset{x=1em, y=1.5ex}
 \dtrans {\varepsilon}{b \labelSep b} 
\tikzset{x=1em, y=2.1ex}
\begin{tikzpicture}[{background rectangle/.style}={fill=backcolour}, show background rectangle, scale=.7, baseline=-.5ex]
	\begin{pgfonlayer}{nodelayer}
		\node [style=black] (0) at (-0.75, -0) {};
		\node [style=none] (1) at (-1.75, -0) {};
		\node [style=none] (2) at (1, -1) {};
		\node [style=none] (3) at (1, 1) {};
		\node [style=coreg] (4) at (-2.5, -0) {\scriptsize $d$};
	\end{pgfonlayer}
	\begin{pgfonlayer}{edgelayer}
		\draw (0) to (1.center);
		\draw [in=75, out=180, looseness=1.00] (3.center) to (0);
		\draw [in=-75, out=180, looseness=1.00] (2.center) to (0);
	\end{pgfonlayer}
\end{tikzpicture}}
\tikzset{x=1em, y=1.5ex}
}{\lambdaseqcomp} 
\]
whereas $f\tns f$ can also perform the following transitions:
\[\derivationRule{\diagEffect{}{d}\dtrans{\varepsilon}{a} \diagEffect{}{d} \quad \diagEffect{}{d}\dtrans{\varepsilon}{b} \diagEffect{}{d}}
{
\tikzset{x=1em, y=2.1ex}
\begin{tikzpicture}[show background rectangle, {background rectangle/.style}={fill=backcolour}]
	\begin{pgfonlayer}{nodelayer}
		\node [style=coreg] (0) at (-0.5, 0.75) {\scriptsize $d$};
		\node [style=none] (1) at (1.5, 0.75) {};
		\node [style=coreg] (2) at (-0.5, -0.75) {\scriptsize $d$};
		\node [style=none] (3) at (1.5, -0.75) {};
	\end{pgfonlayer}
	\begin{pgfonlayer}{edgelayer}
		\draw (1.center) to (0);
		\draw (3.center) to (2);
	\end{pgfonlayer}
\end{tikzpicture}}
\tikzset{x=1em, y=1.5ex}
\dtrans{\varepsilon}{a \labelSep b} 
\tikzset{x=1em, y=2.1ex}
\begin{tikzpicture}[show background rectangle, {background rectangle/.style}={fill=backcolour}]
	\begin{pgfonlayer}{nodelayer}
		\node [style=coreg] (0) at (-0.5, 0.75) {\scriptsize $d$};
		\node [style=none] (1) at (1.5, 0.75) {};
		\node [style=coreg] (2) at (-0.5, -0.75) {\scriptsize $d$};
		\node [style=none] (3) at (1.5, -0.75) {};
	\end{pgfonlayer}
	\begin{pgfonlayer}{edgelayer}
		\draw (1.center) to (0);
		\draw (3.center) to (2);
	\end{pgfonlayer}
\end{tikzpicture}}
\tikzset{x=1em, y=1.5ex}
}{\lambdaparcomp} \qquad \derivationRule{\diagEffect{}{d}\dtrans{\varepsilon}{b} \diagEffect{}{d} \quad \diagEffect{}{d}\dtrans{\varepsilon}{a} \diagEffect{}{d}}
{
\tikzset{x=1em, y=2.1ex}
\begin{tikzpicture}[show background rectangle, {background rectangle/.style}={fill=backcolour}]
	\begin{pgfonlayer}{nodelayer}
		\node [style=coreg] (0) at (-0.5, 0.75) {\scriptsize $d$};
		\node [style=none] (1) at (1.5, 0.75) {};
		\node [style=coreg] (2) at (-0.5, -0.75) {\scriptsize $d$};
		\node [style=none] (3) at (1.5, -0.75) {};
	\end{pgfonlayer}
	\begin{pgfonlayer}{edgelayer}
		\draw (1.center) to (0);
		\draw (3.center) to (2);
	\end{pgfonlayer}
\end{tikzpicture}}
\tikzset{x=1em, y=1.5ex}
\dtrans{\varepsilon}{b \labelSep a} 
\tikzset{x=1em, y=2.1ex}
\begin{tikzpicture}[show background rectangle, {background rectangle/.style}={fill=backcolour}]
	\begin{pgfonlayer}{nodelayer}
		\node [style=coreg] (0) at (-0.5, 0.75) {\scriptsize $d$};
		\node [style=none] (1) at (1.5, 0.75) {};
		\node [style=coreg] (2) at (-0.5, -0.75) {\scriptsize $d$};
		\node [style=none] (3) at (1.5, -0.75) {};
	\end{pgfonlayer}
	\begin{pgfonlayer}{edgelayer}
		\draw (1.center) to (0);
		\draw (3.center) to (2);
	\end{pgfonlayer}
\end{tikzpicture}}
\tikzset{x=1em, y=1.5ex}
}{\lambdaparcomp} 
\]
This is in contradiction with the leftmost equation of~\eqref{eq:natural-comonoid} which requires that $$
\tikzset{x=1em, y=2.1ex}
\begin{tikzpicture}[{background rectangle/.style}={fill=backcolour}, show background rectangle, scale=.7, baseline=-.5ex]
	\begin{pgfonlayer}{nodelayer}
		\node [style=black] (0) at (-0.75, -0) {};
		\node [style=none] (1) at (-1.75, -0) {};
		\node [style=none] (2) at (1, -1) {};
		\node [style=none] (3) at (1, 1) {};
		\node [style=coreg] (4) at (-2.5, -0) {\scriptsize $d$};
	\end{pgfonlayer}
	\begin{pgfonlayer}{edgelayer}
		\draw (0) to (1.center);
		\draw [in=75, out=180, looseness=1.00] (3.center) to (0);
		\draw [in=-75, out=180, looseness=1.00] (2.center) to (0);
	\end{pgfonlayer}
\end{tikzpicture}}
\tikzset{x=1em, y=1.5ex}
 \approx 
\tikzset{x=1em, y=2.1ex}
\begin{tikzpicture}[show background rectangle, {background rectangle/.style}={fill=backcolour}]
	\begin{pgfonlayer}{nodelayer}
		\node [style=coreg] (0) at (-0.5, 0.75) {\scriptsize $d$};
		\node [style=none] (1) at (1.5, 0.75) {};
		\node [style=coreg] (2) at (-0.5, -0.75) {\scriptsize $d$};
		\node [style=none] (3) at (1.5, -0.75) {};
	\end{pgfonlayer}
	\begin{pgfonlayer}{edgelayer}
		\draw (1.center) to (0);
		\draw (3.center) to (2);
	\end{pgfonlayer}
\end{tikzpicture}}
\tikzset{x=1em, y=1.5ex}
$$.

\section{Proof of Theorem \ref{thmHMbw}}

\begin{proposition} Let $n \vdash P$.
\begin{enumerate}
\item If $P\Htr{\alpha} P'$ then $\semProc{n}{P}\Bdtrans{\zerodv}{n\vdash \alpha}\semProc{n}{P'}$
\item If $P\Mtr{\alpha}P'$ then $\semProc{n}{P}\Wdtrans{\zerodv}{n \vdash \alpha}\semProc{n}{P'}$
\end{enumerate}
\end{proposition}

\begin{proof}
We prove 1; the proof for 2 is analogous. Hereafter, given $n\vdash \alpha$, we write $n+n \vdash \alpha \tns \alpha$ for the function mapping $a_i$ to $\alpha(a_i)$ if $i\leq n$ and to $\alpha(a_{i-n})$ if $i>n$.

Suppose that $P\Htr{\alpha} P'$. We prove by induction on the rules of~\eqref{eq:typingdiscipline} that $$\semProc{n}{P}\Bdtrans{\zerodv}{n\vdash \alpha}\semProc{n}{P'}.$$
\begin{itemize}
\item For the rule
$$\derivationRule{n\vdash P \quad n\vdash Q}{n\vdash P|Q}{}$$
we observe that by the leftmost rule in \eqref{eq:Hoare}, $(P|Q) \Htr{\gamma}R$ iff $R=P'|Q'$, $P \Htr{\alpha} P'$ and $Q \Htr{\beta} Q'$ with $\gamma=\alpha\cup \beta$ and  $\alpha\cap \Alphabet(Q)=\beta \cap \Alphabet(P)$. By the induction hypothesis we have that $\semProc{n}{P}\Bdtrans{\zerodv}{n\vdash \alpha}\semProc{n}{P'}$ and $\semProc{n}{Q}\Bdtrans{\zerodv}{n\vdash \beta}\semProc{n}{Q'}$. Since $\alpha\cap \Alphabet(Q)=\beta \cap \Alphabet(P)$ then $(\beta \setminus \alpha) \cap \Alphabet(P)=\emptyset$: by Lemma \ref{lemma:namesHoare}, we have that $\semProc{n}{P}\Bdtrans{\zerodv}{n\vdash \gamma}\semProc{n}{P'}$. By a symmetric argument, $\semProc{n}{Q}\Bdtrans{\zerodv}{n\vdash \gamma}\semProc{n}{Q'}$. Now, $\enc{n\vdash P|Q}$ is equal by definition to $\parProc{black}{n}{P}{Q}$. Since $\Bcomultn{n} \Bdtrans{n+n \vdash \gamma \tns \gamma}{n \vdash \gamma}\Bcomultn{n}$, then 
$$\parProc{black}{n}{P}{Q}\Bdtrans{\zerodv}{n \vdash \gamma} \parProc{black}{n}{P'}{Q'}\text{.}$$

\item For the rule
$$\derivationRule{n+1\vdash P }{n\vdash \nu a_{n+1}(P)}{}$$
we observe that by the rightmost rule in \eqref{eq:Hoare}, $\nu a_{n+1}(P) \Htr{\alpha}P'$ iff $P'=\nu a_{n+1} (P'')$ and $P \Htr{\beta} P''$ with $\beta \setminus \{a_{n+1}\} = \alpha$. By the induction hypothesis 
$\semProc{n+1}{P}\Bdtrans{\zerodv}{n+1\vdash \beta}\semProc{n+1}{P''}$. Now $\enc{n\vdash \nu a_{n+1}(P)}$ is equal by definition to $\hideProc{black}{n}{P}$. Since $
\tikzset{x=1em, y=2.1ex}
\begin{tikzpicture}
	\begin{pgfonlayer}{nodelayer}
		\node [style=none] (0) at (1, -0.25) {};
		\node [style=black] (1) at (-0.25, -0.25) {};
		\node [style=none] (2) at (-1.75, 0.25) {};
		\node [style=none] (3) at (1, 0.25) {};
		\node [style=none] (4) at (-1.25, 0.75) {\scriptsize $n$};
	\end{pgfonlayer}
	\begin{pgfonlayer}{edgelayer}
		\draw (0.center) to (1.center);
		\draw (2.center) to (3.center);
	\end{pgfonlayer}
\end{tikzpicture}}
\tikzset{x=1em, y=1.5ex}
\Bdtrans{n+1\vdash \beta}{n \vdash \alpha}
\tikzset{x=1em, y=2.1ex}
\begin{tikzpicture}
	\begin{pgfonlayer}{nodelayer}
		\node [style=none] (0) at (1, -0.25) {};
		\node [style=black] (1) at (-0.25, -0.25) {};
		\node [style=none] (2) at (-1.75, 0.25) {};
		\node [style=none] (3) at (1, 0.25) {};
		\node [style=none] (4) at (-1.25, 0.75) {\scriptsize $n$};
	\end{pgfonlayer}
	\begin{pgfonlayer}{edgelayer}
		\draw (0.center) to (1.center);
		\draw (2.center) to (3.center);
	\end{pgfonlayer}
\end{tikzpicture}}
\tikzset{x=1em, y=1.5ex}
$, then 
$$\hideProc{black}{n}{P} \Bdtrans{\zerodv}{n \vdash \alpha} \hideProc{black}{n}{P''} = \semProc{n+1}{P'}\text{.}$$

\item For the rule 
$$\derivationRule{ar(\procvar{f})=n}{n\vdash \procvar{f}}{}$$ the result is immediate by the definition of $\beta_b$ and the two leftmost rules in \eqref{eq:opgeneral}.

\item For the rule $$\derivationRule{n\vdash P\quad degree(\perm)\leq n}{n\vdash P\perm}{}$$ we observe that by the rightmost rule in \eqref{eq:opgeneral}, $P\perm \Htr{\alpha} P'$ iff $P'=P'' \perm$ and $P'\Htr{\alpha \circ \perm^{-1}}P''$. By the induction hypothesis, we have that $\semProc{n}{P}\Bdtrans{\zerodv}{n\vdash \alpha  \circ \perm^{-1}}\semProc{n}{P''}$. Observe that $n\vdash \alpha  \circ \perm^{-1}$, since $n \vdash \alpha$ and both $\perm$ and $\perm^{-1}$ have degree $n$.
Now, $\enc{n\vdash P\perm}$ is equal by definition to $\permProc{n}{\sigma}{P}$. Since $ \overline{\perm} \Bdtrans{n\vdash \alpha  \circ \perm^{-1}}{n\vdash \alpha} \overline{\perm} $, then 
$$\semProc{n}{P\perm} \Bdtrans{\zerodv}{n \vdash \alpha} \semProc{n}{P''\perm}= \semProc{n}{P'}\text{.}$$

\item For the rule
$$\derivationRule{n\vdash P}{n+1\vdash P}{}$$
observe that for all processes $P$, if $P \Htr{\alpha}P'$, then $\Alphabet(\alpha)\subseteq \Alphabet(P)$. Since $n\vdash P$, $\Alphabet(P)\subseteq \{a_1,\dots, a_n\}$, then $\Alphabet(\alpha) \subseteq \{a_1,\dots a_n\}$. So $n\vdash \alpha$. By the induction hypothesis $\semProc{n}{P}\Bdtrans{\zerodv}{n\vdash \alpha}\semProc{n}{P'}$. Now, $\enc{n+1\vdash P}$ is equal by definition to $\newvarProc{black}{n}{P}  $. Since $
\tikzset{x=1em, y=2.1ex}
\begin{tikzpicture}
	\begin{pgfonlayer}{nodelayer}
		\node [style=none] (0) at (-1.75, -0.25) {};
		\node [style=black] (1) at (-0.5, -0.25) {};
		\node [style=none] (2) at (-0.75, 0.25) {};
		\node [style=none] (3) at (0.5, 0.25) {};
		\node [style=none] (4) at (-1.25, 0.75) {\scriptsize $n$};
		\node [style=none] (5) at (-1.75, 0.25) {};
	\end{pgfonlayer}
	\begin{pgfonlayer}{edgelayer}
		\draw (0.center) to (1.center);
		\draw [in=180, out=0, looseness=1.25] (2.center) to (3.center);
		\draw (5.center) to (2.center);
	\end{pgfonlayer}
\end{tikzpicture}}
\tikzset{x=1em, y=1.5ex}
\Bdtrans{n\vdash \alpha}{n+1 \vdash \alpha}
\tikzset{x=1em, y=2.1ex}
\begin{tikzpicture}
	\begin{pgfonlayer}{nodelayer}
		\node [style=none] (0) at (-1.75, -0.25) {};
		\node [style=black] (1) at (-0.5, -0.25) {};
		\node [style=none] (2) at (-0.75, 0.25) {};
		\node [style=none] (3) at (0.5, 0.25) {};
		\node [style=none] (4) at (-1.25, 0.75) {\scriptsize $n$};
		\node [style=none] (5) at (-1.75, 0.25) {};
	\end{pgfonlayer}
	\begin{pgfonlayer}{edgelayer}
		\draw (0.center) to (1.center);
		\draw [in=180, out=0, looseness=1.25] (2.center) to (3.center);
		\draw (5.center) to (2.center);
	\end{pgfonlayer}
\end{tikzpicture}}
\tikzset{x=1em, y=1.5ex}
$, then 
$$\newvarProc{black}{n}{P} \Bdtrans{\zerodv}{n+1 \vdash \alpha} \newvarProc{black}{n}{P'}\text{.}$$
\end{itemize}

\end{proof}

\begin{proposition} Let $n \vdash P$, $n \vdash \alpha$ such that $\Alphabet(\alpha)\subseteq \Alphabet(P)$
\begin{enumerate}
\item If $\semProc{n}{P}\Bdtrans{\zerodv}{n\vdash \alpha} \diagState{n}{d}$, then there exists $n\vdash P'$ such that $P\Htr{\alpha} P'$ and $\semProc{n}{P'}=\diagState{n}{d}$.
\item If $\semProc{n}{P} \!\Wdtrans{\zerodv}{n\vdash \alpha}\diagState{n}{d}$
then there is $n\vdash P'$ such that $P\!\Mtr{\alpha}\! P'$ and $\semProc{n}{P'}\!=\!\diagState{n}{d}$.
\end{enumerate}
\end{proposition}
\begin{proof}
We prove 1, the proof for 2 is analogous. As before, given $n\vdash \alpha$, we write $n+n \vdash \alpha \tns \alpha$ for the function mapping $a_i$ to $\alpha(a_i)$ if $i\leq n$ and to $\alpha(a_{i-n})$ if $i>n$.

Suppose that $\semProc{n}{P}\Bdtrans{\zerodv}{n\vdash \alpha} \diagState{n}{d}$. We prove by induction on the rules of~\eqref{eq:typingdiscipline} that there exists $n\vdash P'$ such that $P\Htr{\alpha} P'$ and $\semProc{n}{P'}=\diagState{n}{d}$. 

\begin{itemize}
\item For the rule
$$\derivationRule{n\vdash P \quad n\vdash Q}{n\vdash P|Q}{}$$
we observe that by definition $\enc{n\vdash P|Q} = \parProc{black}{n}{P}{Q}$. Since $\Bcomultn{n} \Bdtrans{n+n\vdash \alpha \tns \alpha }{n\vdash \alpha} \Bcomultn{n}$, then $\parProc{black}{n}{P}{Q} \Bdtrans{\zerodv}{\alpha} \diagState{n}{d}$ iff $\semProc{n}{P}\Bdtrans{\zerodv}{n\vdash \alpha} \diagState{n}{d_1}$ and $\semProc{n}{Q}\Bdtrans{\zerodv}{n\vdash \alpha} \diagState{n}{d_2}$ with $\diagState{n}{d}=\parStates{black}{n}{d_1}{d_2}$. 
Now take $\beta'=\alpha \cap \Alphabet(Q)$ and $\alpha'=\alpha \cap \Alphabet(P)$. We have that $\beta' \subseteq \Alphabet(Q)$ and $\alpha'\subseteq \Alphabet(P)$. Moreover, since $\alpha\subseteq \Alphabet(P) \cup \Alphabet(Q)$ by hypothesis, it holds that $\alpha' \cup \beta'=\alpha$ and that  $\alpha'\cap \Alphabet(Q)=\beta' \cap \Alphabet(P)$.

 By Lemma \ref{lemma:namesHoare} and $\semProc{n}{P}\Bdtrans{\zerodv}{n\vdash \alpha} \diagState{n}{d_1}$ and $\semProc{n}{Q}\Bdtrans{\zerodv}{n\vdash \alpha} \diagState{n}{d_2}$, we obtain that $\semProc{n}{P}\Bdtrans{\zerodv}{n\vdash \alpha'} \diagState{n}{d_1}$ and $\semProc{n}{Q}\Bdtrans{\zerodv}{n\vdash \beta'} \diagState{n}{d_2}$ with $\diagState{n}{d}=\parStates{black}{n}{d_1}{d_2}$.
We can now use induction hypothesis to get a $P'$ and $Q'$ such that $\semProc{n}{P'}=\diagState{n}{d_1}$, $\semProc{n}{Q'}=\diagState{n}{d_2}$, $P \tr{\alpha' }P'$ and $Q \tr{\beta'}Q'$. 
By the leftmost rule in \eqref{eq:Hoare}, we have that $$P|Q \tr{\alpha} P'|Q'\text{.}$$
Finally, by definition of  $\enc{\cdot}$,  $$\enc{n \vdash P'|Q' }=\parProc{black}{n}{P'}{Q'}  =  \parStates{black}{n}{d_1}{d_2} = \diagState{n}{d}\text{.}$$

\item For the rule
$$\derivationRule{n+1\vdash P }{n\vdash \nu a_{n+1}(P)}{}$$
we observe that by definition $\enc{n\vdash \nu a_{n+1} P} =  \hideProc{black}{n}{P}$. Since $
\tikzset{x=1em, y=2.1ex}
\begin{tikzpicture}
	\begin{pgfonlayer}{nodelayer}
		\node [style=none] (0) at (1, -0.25) {};
		\node [style=black] (1) at (-0.25, -0.25) {};
		\node [style=none] (2) at (-1.75, 0.25) {};
		\node [style=none] (3) at (1, 0.25) {};
		\node [style=none] (4) at (-1.25, 0.75) {\scriptsize $n$};
	\end{pgfonlayer}
	\begin{pgfonlayer}{edgelayer}
		\draw (0.center) to (1.center);
		\draw (2.center) to (3.center);
	\end{pgfonlayer}
\end{tikzpicture}}
\tikzset{x=1em, y=1.5ex}
 \Bdtrans{n+1\vdash \beta}{n\vdash \alpha} 
\tikzset{x=1em, y=2.1ex}
\begin{tikzpicture}
	\begin{pgfonlayer}{nodelayer}
		\node [style=none] (0) at (1, -0.25) {};
		\node [style=black] (1) at (-0.25, -0.25) {};
		\node [style=none] (2) at (-1.75, 0.25) {};
		\node [style=none] (3) at (1, 0.25) {};
		\node [style=none] (4) at (-1.25, 0.75) {\scriptsize $n$};
	\end{pgfonlayer}
	\begin{pgfonlayer}{edgelayer}
		\draw (0.center) to (1.center);
		\draw (2.center) to (3.center);
	\end{pgfonlayer}
\end{tikzpicture}}
\tikzset{x=1em, y=1.5ex}
$ with $\alpha=\beta \setminus \{a_{n+1}\}$, then $\hideProc{black}{n}{P}\Bdtrans{\zerodv}{\alpha} \diagState{n}{d}$ iff $\procplusone{n}{\enc{P}}\Bdtrans{\zerodv}{n+1\vdash \beta}\stateplusone{n}{d'}$ with $\diagState{n}{d}= \hideState{black}{n}{d'}$. We can now use the induction hypothesis to get a $P'$ such that $\procplusone{n}{\enc{P'}}=\stateplusone{n}{d'}$ and $P \tr{\alpha }P'$. By the fact that $\alpha=\beta \setminus \{a_{n+1}\}$ and the rightmost rule in \eqref{eq:Hoare}, we have that $$\nu a_{n+1}(P) \tr{\alpha} \nu a_{n+1}(P')\text{.}$$
By definition of  $\enc{\cdot}$,  $$\enc{n \vdash \nu a_{n+1}(P') }=  \hideProc{black}{n}{P'} = \hideState{black}{n}{d'} =\diagState{n}{d}\text{.}$$

\item For the rule 
$$\derivationRule{ar(\procvar{f})=n}{n\vdash \procvar{f}}{}$$ the result is immediate by the definition of $\beta_b$ and the two leftmost rules in \eqref{eq:opgeneral}.

\item For the rule $$\derivationRule{n\vdash P\quad degree(\perm)\leq n}{n\vdash P\perm}{}$$ we observe that by definition $ \enc{n\vdash P\perm} = \permProc{n}{\sigma}{P}$. Since $ \overline{\perm} \Bdtrans{n\vdash \alpha  \circ \perm^{-1}}{n\vdash \alpha} \overline{\perm} $, then 
$\permProc{n}{\sigma}{P} \Bdtrans{\zerodv}{n\vdash \alpha} \diagState{n}{d}$ iff $\semProc{n}{P} \Bdtrans{\zerodv}{n\vdash \alpha \circ \perm^{-1}} \diagState{n}{d'}$ with $\diagState{n}{d}= \permState{n}{\sigma}{d'}$. We can now use the induction hypothesis to get a $P'$ such that $\semProc{n}{P'}=\diagState{n}{d'}$ and $P \tr{\alpha \circ \perm^{-1}}P'$. By the rightmost rule in \eqref{eq:opgeneral}, we have that $$P\perm \tr{\alpha} P' \perm$$ and, by the definition of $\enc{\cdot}$,  $$\enc{n\vdash P' \perm}=  \permProc{n}{\sigma}{P'} = \permState{n}{\sigma}{d'} = \diagState{n}{d}\text{.}$$

\item For the rule
$$\derivationRule{n\vdash P}{n+1\vdash P}{}$$
we observe that since $n\vdash P$, then $\Alphabet(P)\subseteq \{a_1,\dots, a_n\}$ and thus $a_{n+1}\notin \alpha$. By definition, $\enc{n+1\vdash P} = \newvarProc{black}{n}{P}$. Since $
\tikzset{x=1em, y=2.1ex}
\begin{tikzpicture}
	\begin{pgfonlayer}{nodelayer}
		\node [style=none] (0) at (-1.75, -0.25) {};
		\node [style=black] (1) at (-0.5, -0.25) {};
		\node [style=none] (2) at (-0.75, 0.25) {};
		\node [style=none] (3) at (0.5, 0.25) {};
		\node [style=none] (4) at (-1.25, 0.75) {\scriptsize $n$};
		\node [style=none] (5) at (-1.75, 0.25) {};
	\end{pgfonlayer}
	\begin{pgfonlayer}{edgelayer}
		\draw (0.center) to (1.center);
		\draw [in=180, out=0, looseness=1.25] (2.center) to (3.center);
		\draw (5.center) to (2.center);
	\end{pgfonlayer}
\end{tikzpicture}}
\tikzset{x=1em, y=1.5ex}
 \Bdtrans{n\vdash \alpha}{n+1\vdash \alpha} 
\tikzset{x=1em, y=2.1ex}
\begin{tikzpicture}
	\begin{pgfonlayer}{nodelayer}
		\node [style=none] (0) at (-1.75, -0.25) {};
		\node [style=black] (1) at (-0.5, -0.25) {};
		\node [style=none] (2) at (-0.75, 0.25) {};
		\node [style=none] (3) at (0.5, 0.25) {};
		\node [style=none] (4) at (-1.25, 0.75) {\scriptsize $n$};
		\node [style=none] (5) at (-1.75, 0.25) {};
	\end{pgfonlayer}
	\begin{pgfonlayer}{edgelayer}
		\draw (0.center) to (1.center);
		\draw [in=180, out=0, looseness=1.25] (2.center) to (3.center);
		\draw (5.center) to (2.center);
	\end{pgfonlayer}
\end{tikzpicture}}
\tikzset{x=1em, y=1.5ex}
$, then $\newvarProc{black}{n}{P} \dtrans{\zerodv}{n+1\vdash \alpha}{} \stateplusone{n}{d}$ iff $\semProc{n}{P} \Bdtrans{\zerodv}{n\vdash \alpha} \diagState{n}{d'}$ with $\stateplusone{n}{d}= \newvarState{black}{n}{d'}$. We can now use induction hypothesis to get a $P'$ such that $\semProc{n}{P'}=\diagState{n}{d'}$ and $$P \tr{\alpha }P'.$$ By 
definition of  $\enc{\cdot}$,  $$\enc{n+1 \vdash P' }= \newvarProc{black}{n}{P'}  =\newvarState{black}{n}{d'} =\diagState{n}{d}\text{.}$$
\end{itemize}

\end{proof}

\begin{lemma}\label{lemma:namesHoare}
Let $n\vdash P$, $n\vdash \alpha$, $i\leq n$ and $a_i \notin \Alphabet(P)$.
$$\enc{n\vdash P}\Bdtrans{\zerodv}{n\vdash \alpha}\enc{n\vdash Q} \quad \text{iff}\quad \enc{n\vdash P}\Bdtrans{\zerodv}{n\vdash \alpha\cup\{a_i\}} \enc{n\vdash Q}.$$
\end{lemma}
\begin{proof}
First, using Lemma~\ref{lemma:i-disconnect}, we can find $\diagState{n-1}{d}$ such that
\[\semProc{n}{P} \;= \;\vardisconnect{black}{n}{i}{d}\]
Then, we use a permutation to relocate the $\Bcounit$: let $\perm$ be the transposition $(i\;\; n)$. By definition of $\enc{-}$ we have
\[\semProc{n}{P\perm}\;=\;\permProc{n}{\perm}{P} \;=\; \newvarState{black}{n-1}{d}\]
From the derivation rules for props (\S~\ref{sec:comp-prop}), we see that any transition out of $\semProc{n}{P\sigma}$ must come from the derivation rule for $\oplus$, namely:
\[\derivationRule{\diagState{n-1}{d}\Bdtrans{\zerodv}{\alpha} \diagState{n-1}{e}\quad \Bcounit\Bdtrans{\zerodv}{v} \Bcounit}
{\newvarState{black}{n-1}{d} \Bdtrans{\zerodv}{\alpha \labelSep v} \newvarState{black}{n-1}{e}}{\lambda^3}\]
Hence, we must have $\semProc{n}{Q\sigma} = \newvarState{black}{n-1}{e}$ for some $\diagState{n}{e}$. Now, notice that we have $\Bcounit\Bdtrans{\zerodv}{1}\:\Bcounit$ and $\Bcounit\Bdtrans{\zerodv}{0}\:\Bcounit$ so that
\[\semProc{n}{P\perm}\Bdtrans{\zerodv}{n\vdash \alpha}\semProc{n}{Q\perm}\;\text{ iff }\;\semProc{n}{P\perm}\Bdtrans{\zerodv}{n\vdash \alpha\cup\{a_n\}} \semProc{n}{Q\perm}.\]
Applying the transposition $(i\;\; n)$ again we can conclude that
\[\semProc{n}{P}\Bdtrans{\zerodv}{n\vdash \alpha}\semProc{n}{Q}\;\text{ iff }\;\semProc{n}{P}\Bdtrans{\zerodv}{n\vdash \alpha\cup\{a_i\}} \semProc{n}{Q}.\]
\end{proof}

\begin{lemma}\label{lemma:i-disconnect}
Let $n\vdash P$, $i\leq n$ and $a_i \notin \Alphabet(P)$. Then there exists $\diagState{n-1}{d}$ such that
\[\semProc{n}{P} \;= \;\vardisconnect{black-faded}{n}{i}{d}\]
\end{lemma}
\begin{proof}
We prove this by structural induction on the rules of~\eqref{eq:typingdiscipline}.
\begin{itemize}
\item For the rule
$$\derivationRule{n\vdash P_1 \quad n\vdash P_2}{n\vdash P_1|P_2}{}$$
$\Alphabet(P_1|P_2) = \Alphabet(P_1)\cup \Alphabet(P_2)$ and therefore $a_i \notin \Alphabet(P)$ and $a_i \notin \Alphabet(P_2)$. We can apply the induction hypothesis to obtain $d_1$ and $d_2$ such that
\[\semProc{n}{P_1}\; = \; \vardisconnect{black-faded}{n}{i}{d_1}\qquad \text{ and } \qquad \semProc{n}{P_2}\; = \; \vardisconnect{black-faded}{n}{i}{d_2}\]
and we can deduce
\[\semProc{n}{P_1|P_2} \; = \; \parProc{black-faded}{n}{P_1}{P_2} \; =\;  
\tikzset{x=1em, y=2.1ex}
\begin{tikzpicture}[show background rectangle, {background rectangle/.style}={fill=backcolour}]
	\begin{pgfonlayer}{nodelayer}
		\node [style=black-faded] (0) at (0.25, 1.5) {};
		\node [style=none] (1) at (0.5, 2.25) {};
		\node [style=none] (2) at (2.25, 1.75) {};
		\node [style=none] (3) at (0.25, -2.75) {\scriptsize $n-i$};
		\node [style=none] (4) at (0.25, 2.25) {};
		\node [style=none] (5) at (0.25, 0.75) {};
		\node [style=none] (6) at (0.5, 0.75) {};
		\node [style=none] (7) at (2.25, 1.25) {};
		\node [style=reg] (8) at (2.5, 1.5) {$d_1$};
		\node [style=black-faded] (9) at (0.25, -1.5) {};
		\node [style=none] (10) at (0.25, -2.25) {};
		\node [style=none] (11) at (0.25, 2.75) {\scriptsize $i-1$};
		\node [style=none] (12) at (0.5, -0.75) {};
		\node [style=none] (13) at (2.25, -1.75) {};
		\node [style=none] (14) at (2.25, -1.25) {};
		\node [style=none] (15) at (0.25, -0.75) {};
		\node [style=none] (16) at (0.5, -2.25) {};
		\node [style=reg] (17) at (2.5, -1.5) {$d_2$};
		\node [style=black-faded] (18) at (-1.75, -0) {};
		\node [style=black-faded] (19) at (-1.75, 1) {};
		\node [style=black-faded] (20) at (-1.75, -1) {};
		\node [style=none] (21) at (-3, -1) {};
		\node [style=none] (22) at (-3, -0) {};
		\node [style=none] (23) at (-3, 1) {};
	\end{pgfonlayer}
	\begin{pgfonlayer}{edgelayer}
		\draw [in=180, out=0, looseness=1.25] (1.center) to (2.center);
		\draw (4.center) to (1.center);
		\draw [in=180, out=0, looseness=1.25] (6.center) to (7.center);
		\draw (5.center) to (6.center);
		\draw [in=180, out=0, looseness=1.25] (12.center) to (14.center);
		\draw (15.center) to (12.center);
		\draw [in=180, out=0, looseness=1.25] (16.center) to (13.center);
		\draw (10.center) to (16.center);
		\draw (23.center) to (19);
		\draw (22.center) to (18);
		\draw (21.center) to (20);
		\draw [in=180, out=-60, looseness=1.00] (20) to (10.center);
		\draw [in=180, out=60, looseness=0.75] (20) to (5.center);
		\draw [in=165, out=-60, looseness=1.00] (18) to (9);
		\draw [in=-165, out=60, looseness=1.00] (18) to (0);
		\draw [in=180, out=60, looseness=1.00] (19) to (4.center);
		\draw [in=180, out=-75, looseness=0.75] (19) to (15.center);
	\end{pgfonlayer}
\end{tikzpicture}}
\tikzset{x=1em, y=1.5ex}
 \; = \; 
\tikzset{x=1em, y=2.1ex}
\begin{tikzpicture}[show background rectangle, {background rectangle/.style}={fill=backcolour}]
	\begin{pgfonlayer}{nodelayer}
		\node [style=none] (0) at (0.5, 2.25) {};
		\node [style=none] (1) at (2.25, 1.75) {};
		\node [style=none] (2) at (0.25, -2.75) {\scriptsize $n-i$};
		\node [style=none] (3) at (0.25, 2.25) {};
		\node [style=none] (4) at (0.25, 0.75) {};
		\node [style=none] (5) at (0.5, 0.75) {};
		\node [style=none] (6) at (2.25, 1.25) {};
		\node [style=reg] (7) at (2.5, 1.5) {$d_1$};
		\node [style=none] (8) at (0.25, -2.25) {};
		\node [style=none] (9) at (0.25, 2.75) {\scriptsize $i-1$};
		\node [style=none] (10) at (0.5, -0.75) {};
		\node [style=none] (11) at (2.25, -1.75) {};
		\node [style=none] (12) at (2.25, -1.25) {};
		\node [style=none] (13) at (0.25, -0.75) {};
		\node [style=none] (14) at (0.5, -2.25) {};
		\node [style=reg] (15) at (2.5, -1.5) {$d_2$};
		\node [style=black-faded] (16) at (-1.75, -0) {};
		\node [style=black-faded] (17) at (-1.75, 1) {};
		\node [style=black-faded] (18) at (-1.75, -1) {};
		\node [style=none] (19) at (-3, -1) {};
		\node [style=none] (20) at (-3, -0) {};
		\node [style=none] (21) at (-3, 1) {};
	\end{pgfonlayer}
	\begin{pgfonlayer}{edgelayer}
		\draw [in=180, out=0, looseness=1.25] (0.center) to (1.center);
		\draw (3.center) to (0.center);
		\draw [in=180, out=0, looseness=1.25] (5.center) to (6.center);
		\draw (4.center) to (5.center);
		\draw [in=180, out=0, looseness=1.25] (10.center) to (12.center);
		\draw (13.center) to (10.center);
		\draw [in=180, out=0, looseness=1.25] (14.center) to (11.center);
		\draw (8.center) to (14.center);
		\draw (21.center) to (17);
		\draw (20.center) to (16);
		\draw (19.center) to (18);
		\draw [in=180, out=-60, looseness=1.00] (18) to (8.center);
		\draw [in=180, out=60, looseness=0.75] (18) to (4.center);
		\draw [in=180, out=60, looseness=1.00] (17) to (3.center);
		\draw [in=180, out=-75, looseness=0.75] (17) to (13.center);
	\end{pgfonlayer}
\end{tikzpicture}}
\tikzset{x=1em, y=1.5ex}
\]
\item For the rule
$$\derivationRule{n+1\vdash P }{n\vdash \nu a_{n+1}(P)}{}$$
we have $\Alphabet(\nu a_{n+1}(P)) = \Alphabet(P)\setminus \{a_{n+1}\}$ which implies that $a_i \notin \Alphabet(P)$. We can therefore apply the induction hypothesis to find $d$ such that
\[\semProc{n+1}{P}\; = \;\;\vardisconnect{black-faded}{n+1}{i}{d}\]
and we can deduce
\[\semProc{n}{\nu a_{n+1}P} \; = \; \hideProc{black-faded}{n}{P} \; = \; 
\tikzset{x=1em, y=2.1ex}
\begin{tikzpicture}[show background rectangle, {background rectangle/.style}={fill=backcolour}]
	\begin{pgfonlayer}{nodelayer}
		\node [style=none] (0) at (1, -0.25) {};
		\node [style=none] (1) at (-0.75, -1.25) {};
		\node [style=none] (2) at (-2, -0) {};
		\node [style=reg] (3) at (1.5, -0.25) {$d$};
		\node [style=none] (4) at (1, -0) {};
		\node [style=black-faded] (5) at (-0.5, -0) {};
		\node [style=none] (6) at (-1, 1.25) {\scriptsize $i-1$};
		\node [style=none] (7) at (-2, 0.75) {};
		\node [style=none] (8) at (-0.75, 0.75) {};
		\node [style=none] (9) at (-2, -1.25) {};
		\node [style=none] (10) at (-1, -0.5) {\scriptsize $n-i$};
		\node [style=black-faded] (11) at (-1.25, -2) {};
		\node [style=none] (12) at (-0.75, -2) {};
		\node [style=none] (13) at (1, -0.5) {};
	\end{pgfonlayer}
	\begin{pgfonlayer}{edgelayer}
		\draw (2.center) to (5);
		\draw [in=180, out=0, looseness=1.25] (8.center) to (4.center);
		\draw (7.center) to (8.center);
		\draw [in=180, out=0, looseness=1.25] (1.center) to (0.center);
		\draw (9.center) to (1.center);
		\draw [in=180, out=0, looseness=1.25] (12.center) to (13.center);
		\draw (11.center) to (12.center);
	\end{pgfonlayer}
\end{tikzpicture}}
\tikzset{x=1em, y=1.5ex}
 \]
\item $P$ cannot be equal to a process variable $\procvar{f}$ since, for $n\vdash \procvar{f}$ we have $\Alphabet(\procvar{f}) = ar(\procvar{f}) = n$.
\item For the rule $$\derivationRule{n\vdash P\quad degree(\perm)\leq n}{n\vdash P\perm}{}$$ we have $\Alphabet(P\sigma) = \perm[\Alphabet(P)]$. Therefore $a_{j}\notin \Alphabet(P)$ for $j\df \perm^{-1}(i)$. By the induction hypothesis, we can find $d$ such that
\[\semProc{n}{P} = \vardisconnect{black-faded}{n}{j}{d}\]
and conclude that
\[\semProc{n}{P\perm}\;=\; \permProc{n}{\perm}{P}\; =\;
\tikzset{x=1em, y=2.1ex}
\begin{tikzpicture}[{background rectangle/.style}={fill=backcolour}, show background rectangle]
	\begin{pgfonlayer}{nodelayer}
		\node [style=none] (0) at (1, -0.25) {};
		\node [style=none] (1) at (-0.75, -0.75) {};
		\node [style=none] (2) at (-2, -0) {};
		\node [style=reg] (3) at (1.5, -0) {$d$};
		\node [style=none] (4) at (1, 0.25) {};
		\node [style=black-faded] (5) at (-0.75, -0) {};
		\node [style=none] (6) at (-0.75, 1.25) {\scriptsize $j-1$};
		\node [style=none] (7) at (-2, 0.75) {};
		\node [style=none] (8) at (-0.75, 0.75) {};
		\node [style=none] (9) at (-2, -0.75) {};
		\node [style=none] (10) at (-0.75, -1.25) {\scriptsize $n-j$};
		\node [style=tallbox] (11) at (-2.5, -0) {\scriptsize $\overline\perm$};
		\node [style=none] (12) at (-3, -0.75) {};
		\node [style=none] (13) at (-3, -0) {};
		\node [style=none] (14) at (-3, 0.75) {};
		\node [style=none] (15) at (-4.5, 0.75) {};
		\node [style=none] (16) at (-4.5, -0) {};
		\node [style=none] (17) at (-4.5, -0.75) {};
		\node [style=none] (18) at (-4, -1.25) {\scriptsize $n-i$};
		\node [style=none] (19) at (-4, 1.25) {\scriptsize $i-1$};
	\end{pgfonlayer}
	\begin{pgfonlayer}{edgelayer}
		\draw (2.center) to (5);
		\draw [in=180, out=0, looseness=1.25] (8.center) to (4.center);
		\draw (7.center) to (8.center);
		\draw [in=180, out=0, looseness=1.25] (1.center) to (0.center);
		\draw (9.center) to (1.center);
		\draw (17.center) to (12.center);
		\draw (16.center) to (13.center);
		\draw (15.center) to (14.center);
	\end{pgfonlayer}
\end{tikzpicture}}
\tikzset{x=1em, y=1.5ex}
\; =\;
\tikzset{x=1em, y=2.1ex}
\begin{tikzpicture}[{background rectangle/.style}={fill=backcolour}, show background rectangle]
	\begin{pgfonlayer}{nodelayer}
		\node [style=none] (0) at (1, -0.25) {};
		\node [style=none] (1) at (-0.75, -0.75) {};
		\node [style=none] (2) at (-1, -0) {};
		\node [style=reg] (3) at (1.5, -0) {$d$};
		\node [style=none] (4) at (1, 0.25) {};
		\node [style=none] (5) at (0.25, 1.25) {\scriptsize $j-1$};
		\node [style=none] (6) at (-1, 0.75) {};
		\node [style=none] (7) at (-0.75, 0.75) {};
		\node [style=none] (8) at (-1, -0.75) {};
		\node [style=none] (9) at (0.25, -1.25) {\scriptsize $n-j$};
		\node [style=tallbox] (10) at (-1.5, -0) {\scriptsize $\overline{\perm}'$};
		\node [style=none] (11) at (-2, -0.75) {};
		\node [style=black-faded] (12) at (-3, -0) {};
		\node [style=none] (13) at (-2, 0.75) {};
		\node [style=none] (14) at (-4.25, 0.75) {};
		\node [style=none] (15) at (-4.25, -0) {};
		\node [style=none] (16) at (-4.25, -0.75) {};
		\node [style=none] (17) at (-3.5, 1.25) {\scriptsize $i-1$};
		\node [style=none] (18) at (-3.5, -1.25) {\scriptsize $n-i$};
	\end{pgfonlayer}
	\begin{pgfonlayer}{edgelayer}
		\draw [in=180, out=0, looseness=1.25] (7.center) to (4.center);
		\draw (6.center) to (7.center);
		\draw [in=180, out=0, looseness=1.25] (1.center) to (0.center);
		\draw (8.center) to (1.center);
		\draw (16.center) to (11.center);
		\draw (15.center) to (12);
		\draw (14.center) to (13.center);
	\end{pgfonlayer}
\end{tikzpicture}}
\tikzset{x=1em, y=1.5ex}
 \]
for some permutation $\perm'\from n-1\to n-1$.
\item For the rule
$$\derivationRule{n+1\vdash P }{n\vdash \nu a_{n+1}(P)}{}$$
we have $\semProc{n}{P}\; = \; \newvarProc{black-faded}{n-1}{P}$
and, using the induction hypothesis we can find $d$ such that
\[\semProc{n-1}{P} \;= \;\vardisconnect{black-faded}{n-1}{i}{d}\]
The conclusion follows immediately. 
\end{itemize}
\end{proof}

\end{document}